\documentclass[11pt,reqno,letterpaper]{amsart}
\usepackage{booktabs}
\usepackage{color}
\usepackage[colorlinks=true, allcolors=blue, hypertexnames=false]{hyperref}
\usepackage{amsmath, amssymb, amsthm}
\usepackage{mathrsfs}
\usepackage{mathtools}
\usepackage[capitalize,nameinlink]{cleveref}
\usepackage{fullpage}
\usepackage[noadjust]{cite}
\usepackage{graphics}
\usepackage{pifont}
\usepackage{tikz}
\usepackage{bbm}
\usepackage{thm-restate}
\usepackage[T1]{fontenc}

\usetikzlibrary{arrows.meta}

\usepackage{environ}
\usepackage{framed}
\usepackage{url}
\usepackage[linesnumbered,ruled,vlined]{algorithm2e}
\usepackage[noend]{algpseudocode}
\usepackage[labelfont=bf]{caption}
\usepackage{cite}
\usepackage{framed}
\usepackage[framemethod=tikz]{mdframed}
\usepackage{appendix}
\usepackage{graphicx}
\usepackage[textsize=tiny]{todonotes}
\usepackage{tcolorbox}
\usepackage{enumerate}
\allowdisplaybreaks[1]
\usepackage{enumerate}
\usepackage{stmaryrd}
\usepackage[margin=1in]{geometry}

\DeclareSymbolFont{symbolsC}{U}{pxsyc}{m}{n}
\SetSymbolFont{symbolsC}{bold}{U}{pxsyc}{bx}{n}
\DeclareFontSubstitution{U}{pxsyc}{m}{n}
\DeclareMathSymbol{\medcircle}{\mathbin}{symbolsC}{7}

\newcommand{\fpras}{\mathsf{FPRAS}}

\crefname{algocf}{Algorithm}{Algorithms}

\AtBeginEnvironment{appendices}{\crefalias{section}{appendix}} 

\usepackage[color,final]{showkeys} 

\colorlet{refkey}{orange!20}
\colorlet{labelkey}{blue!30}

\crefname{algocf}{Algorithm}{Algorithms}

\numberwithin{equation}{section}
\newtheorem{theorem}{Theorem}[section]
\newtheorem{proposition}[theorem]{Proposition}
\newtheorem{lemma}[theorem]{Lemma}
\newtheorem{claim}[theorem]{Claim}
\newtheorem{observation}[theorem]{Observation}
\crefname{claim}{Claim}{Claims}

\newtheorem{corollary}[theorem]{Corollary}

\newtheorem*{question*}{Question}
\newtheorem{fact}[theorem]{Fact}

\theoremstyle{definition}
\newtheorem{definition}[theorem]{Definition}

\newtheorem*{definition*}{Definition}

\newtheorem{remark}[theorem]{Remark}

\newtheorem{condition}[theorem]{Condition}



\newcommand{\mb}{\mathbb}
\newcommand{\mbf}{\mathbf}

\newcommand{\mc}{\mathcal}

\newcommand{\mr}{\mathrm}

\newcommand{\ol}{\overline}
\newcommand{\on}{\operatorname}
\newcommand{\wh}{\widehat}
\newcommand{\wt}{\widetilde}

\newcommand{\eps}{\varepsilon}

\let\originalleft\left
\let\originalright\right
\renewcommand{\left}{\mathopen{}\mathclose\bgroup\originalleft}
\renewcommand{\right}{\aftergroup\egroup\originalright}

\DeclareMathOperator{\id}{id}
\newcommand{\Bernoulli}{\on{Ber}}

\renewcommand{\index}{\on{index}}

\newcommand{\Tmix}{T_{\mathrm{mix}}}

\newcommand{\vecsigma}{{\vec{\sigma}}}

\newcommand{\Tcouple}{{T_{\mathrm{cp}}}}
\newcommand{\Ccouple}{{C_{\mathrm{cp}}}}
\newcommand{\Tmod}{{T_{\mathrm{m}}}}
\newcommand{\Cmod}{{C_{\mathrm{buffer}}}}

\newcommand{\Exchange}{\mathsf{Ex}}
\newcommand{\Pure}{\mathsf{PureEx}}
\newcommand{\Swappable}{\mathsf{Swap}}

\newcommand{\Avoid}{\mathsf{Avoid}}
\newcommand{\LU}{\on{LU}}

\allowdisplaybreaks

\newcommand{\ignore}[1]{}

\title{Sampling Colorings Close to the Maximum Degree: Non-Markovian Coupling and Local Uniformity}

\author[A1]{Vishesh Jain}
\address{Department of Mathematics, Statistics, and Computer Science, University of Illinois Chicago, Chicago, IL, 60607 USA}
\email{visheshj@uic.edu}

\author[A2]{Clayton Mizgerd}
\address{Department of Mathematics, Statistics, and Computer Science, University of Illinois Chicago, Chicago, IL, 60607 USA}
\email{cmizge2@uic.edu}

\author[A3]{Eric Vigoda}
\address{Department of Computer Science, University of California, Santa Barbara, CA, 93106 USA}
\email{vigoda@ucsb.edu}

\DeclareMathOperator{\BC}{BC}
\DeclareMathOperator{\NM}{NM}
\DeclareMathOperator{\Jerrum}{Jerrum}
\newcommand{\Temp}{\mathsf{Temp}}

\newcommand{\paren}[1]{\left(#1\right)}

\begin{document}

\begin{abstract}

Sampling graph colorings via local Markov chains is a central problem in approximate counting and Markov chain Monte Carlo (MCMC).
We address the problem of sampling a random $k$-coloring of a graph with maximum degree $\Delta$.  
The simplest algorithmic approach is to establish rapid mixing of the single-site update chain known as the Metropolis Glauber dynamics, which at each step chooses a random vertex $v$ and proposes a random color $c$, recoloring $v$ to $c$ if the resulting coloring remains proper.  
It is a long-standing open problem to prove that the Glauber dynamics has polynomial mixing time on all graphs whenever $k\geq\Delta+2$.

We prove that for every $\delta>0$ and all $\Delta \geq \Delta_0(\delta)$, if $k\ge (1+\delta)\Delta$ then the Glauber dynamics has optimal mixing time of $O_{\delta}(|V| \log |V|)$ on any graph of girth $\geq 7$ and maximum degree~$\Delta$.  
Our approach builds on a non-Markovian coupling introduced by Hayes and Vigoda (2003) for the large-degree regime $\Delta=\Omega(\log n)$ and girth $11$, in which updates at time $t$ may depend on and modify proposed updates at future times.  A complete analysis of this framework requires resolving substantial technical obstacles that remain in the original argument, and extending it to the constant-degree regime introduces further difficulties, since non-Markovian updates may fail with constant probability.

We overcome these obstacles by developing and analyzing a refined local non-Markovian coupling, 
and by establishing new local-uniformity results for the Metropolis dynamics, extending prior results for the heat-bath chain due to Hayes (2013). Together, these ingredients provide a complete analysis of the non-Markovian coupling framework in the large-degree regime, while simultaneously strengthening it substantially to obtain optimal mixing all the way down to the constant-degree setting.
\end{abstract}

\maketitle

\thispagestyle{empty}

\newpage

\setcounter{page}{1}

\section{Introduction}
A problem of long-standing interest in the fields of approximate counting and Markov chain Monte Carlo (MCMC) is estimating the number of $k$-colorings of a graph $G=(V,E)$ of maximum degree~$\Delta$. Since the seminal work of Jerrum~\cite{jerrum1995very}, it has been an outstanding open problem to efficiently sample $k$-colorings when $k$ is close to $\Delta$. Existing results achieve such guarantees only under strong structural assumptions, such as large girth (the length of the shortest cycle) or when $\Delta$ grows with $n=|V|$. Progress on sampling $k$-colorings has also served as a driving force for the development of new algorithms for approximate counting and new techniques for analyzing the mixing times of MCMC algorithms. In this paper we make significant progress on sampling $k$-colorings when $k$ is close to the maximum degree $\Delta$.

Throughout this paper, let $[k]=\{1,\dots,k\}$, and 
a $k$-coloring is a proper vertex coloring $\chi:V\to[k]$ such that for all $\{v,w\}\in E$, $\chi(v)\neq\chi(w)$.  
For a graph $G=(V,E)$ and an integer $k\geq 2$, let 
$\Omega$ denote the set of $k$-colorings of $G$, and let $\pi$ denote the uniform distribution over $\Omega$.  

In the approximate counting problem we are aiming for a fully-polynomial randomized approximation scheme ($\fpras$) for estimating $|\Omega|$.  In the corresponding approximate sampling problem, which we study in this paper,
we are given as input a graph $G=(V,E)$ of maximum degree $\Delta$, an integer $k\geq 2$, and an error parameter $\eps>0$, and our goal is to sample a $k$-coloring from a distribution $\mu$ 
which is within total variation distance $\eps$ of the uniform distribution $\pi$ over $\Omega$, in time polynomial in $n=|V|$ and $\log(1/\eps)$; this then yields an $\fpras$ for the approximate counting problem (e.g., see \cite{SVV:annealing,Huber,Kolmogorov}). 

The simplest and most widely studied approach for sampling colorings is the Glauber dynamics, which is the single-site update Markov chain.  In this paper we analyze the Metropolis version of the Glauber dynamics, which we simply refer to as the Glauber dynamics.  

\begin{definition}[Discrete-time Metropolis Glauber Dynamics]\label{def:metropolis} Let $G = (V,E)$ be a graph and let $k\geq 2$ be the number of colors. Let $X_0$ be an arbitrary proper $k$-coloring of $G$.  
The dynamics evolves as a Markov chain $\{X_t\}_{t\geq 0}$: 
\begin{enumerate}
    \item At each time step $t$, choose a vertex $u \in V$ uniformly at random and a candidate color  $c \in [k]$ uniformly at random.
    \item Define the candidate configuration $X'$ by  \[ 
    X'(w) = \begin{cases} c & \text{if } w = u, \\ X_t(w) & \text{otherwise.} \end{cases} 
    \]  
    \item If $X'$ is a proper coloring, set $X_{t+1} = X'$; else, set $X_{t+1} = X_t$.
\end{enumerate}
\end{definition}
It is straightforward to verify that the Glauber dynamics is ergodic whenever $k\ge\Delta+2$. Since the transitions are symmetric, the unique stationary distribution is $\pi$, the uniform distribution over~$\Omega$.

\begin{remark}
    An alternative version of the Glauber dynamics is the heat-bath Glauber dynamics. The main difference between the heat-bath and Metropolis dynamics is in the update probability. The heat-bath chain selects a color uniformly from the \textit{currently available} colors $A(X_t, v) := [k] \setminus \{X_t(w):w\in N(v)\}$, guaranteeing a valid move. In contrast, the Metropolis process attempts any color in $[k]$. 
\end{remark}

The mixing time $\Tmix$ is the minimum number of steps, from the worst initial state $X_0$, to ensure that the distribution of $X_{\Tmix}$ is within total variation distance $1/4$ of the stationary distribution.  Hayes and Sinclair~\cite{HayesSinclair} proved that for every constant $\Delta$, for any graph of maximum degree $\Delta$ the mixing time is $\Omega(n\log n)$.  
Therefore we say the Glauber dynamics has optimal mixing time when $\Tmix=O(n\log n)$.

The first major progress on sampling $k$-colorings was by Jerrum~\cite{jerrum1995very} who proved $O(n\log n)$ mixing time of the Glauber dynamics whenever $k>2\Delta$.  This was improved by Vigoda~\cite{Vigoda} to $k>(11/6)\Delta$, and further improved to $k\geq 1.809\Delta$ by~\cite{CDMPP19,CarlsonVigoda25} using a more general Markov chain known as the flip dynamics; this is currently the best poly-time sampling result which applies for general graphs.

Further progress on sampling $k$-colorings typically relies either on strong structural assumptions (such as large girth and constant $\Delta$) or applies only when $k/\Delta$ is large. We briefly review the most relevant prior results.
Chen, Liu, Mani, and Moitra~\cite{CLMM23} proved that for every {\em constant} $\Delta$ and all $k\geq\Delta+3$, for any graph with girth $\geq g(\Delta)$ the Glauber dynamics has $O(n\log n)$ mixing time.  Regarding the girth requirement in \cite{CLMM23}, in the simpler case when $k\geq (1+\delta)\Delta$, the required girth grows at least as $C(\delta)\log^2 \Delta$.   
Other notable results include polynomial mixing time of the Glauber dynamics when $k>1.764\Delta$ for all $\Delta\geq 3$ for triangle-free graphs~\cite{CGSV21,FGYZ21, jain2022spectral} using spectral independence and a result of Dyer, Frieze, Hayes, and Vigoda~\cite{DFHV},
which establishes optimal mixing time of the Glauber dynamics when either: $k > 1.764\Delta$ and girth $\geq 5$, or $k> 1.489\Delta$ and girth $\geq 7$, and $\Delta\geq C(\delta)$ where $\delta$ is the gap from the relevant threshold; see also \cite{dyer2003randomly, Molloy, hayes-old-girth5} for earlier results in this direction.

Most closely related to our work, Hayes and Vigoda~\cite{hayesvigoda-nonmarkovian} claimed
to establish $O(n\log n)$ mixing of the Glauber dynamics for every fixed $\delta>0$ and for $k\geq(1+\delta)\Delta$ when the girth $g\geq 11$ and $\Delta=\Omega(\log n)$; their proof contains several non-trivial technical issues (see~\cref{sub:corrections,rem:HV-error}), 
which we address and resolve in our work.
In particular, our work requires revisiting the construction of the bounding chain,
the definition and validity of the non-Markovian coupling itself, 
and how these components are combined in the analysis.
In addition, we extend the local uniformity results of Hayes~\cite{hayes2013local} to the Metropolis Glauber dynamics, which is also required but not established in~\cite{hayesvigoda-nonmarkovian}.  
Furthermore, extending the Hayes--Vigoda non-Markovian coupling to the constant-degree regime introduces additional challenges in both the definition and analysis of the coupling.

Our main result establishes optimal mixing of the Glauber dynamics
in the near-critical regime $k \ge (1+\delta)\Delta$, assuming only girth $\ge 7$ and $\Delta \geq \Delta_0(\delta)$.

\begin{theorem}\label{thm:main-mixing}
For every $\delta>0$, there exist constants $\Delta_0(\delta),C(\delta)>0$ such that the following holds. For any graph $G=(V,E)$ of maximum degree $\Delta \geq \Delta_0(\delta)$ and girth $\geq 7$, for any $k\geq(1+\delta)\Delta$, the Metropolis Glauber dynamics on $k$-colorings has mixing time
\[
\Tmix \leq C(\delta)\, n \log n .
\]
\end{theorem}

\subsection{Technical Overview}
A standard approach to proving rapid mixing of the Glauber dynamics is via path coupling~\cite{BubleyDyer}.
Consider a pair of colorings $X_t,Y_t$ that differ at a single vertex~$z$.  Let $c_X=X_t(z)$ and $c_Y=Y_t(z)$ denote the disagree colors.  Under the identity coupling we update the same vertex $v_t$ with the same attempted color $c_t$ in both chains.  For each $w\in N(z)$, there are at most two color choices which can create a new disagreement.  
Since there are at least $k-\Delta$ colors that can successfully recolor $z$ in both chains, while each neighbor $w\in N(z)$ can create at most two new disagreements, contraction requires $k-\Delta > 2\Delta$, which yields rapid mixing when $k>3\Delta$.

In Jerrum's coupling, for $v_t\in N(z)$, we pair color $c_X$ in one chain with $c_Y$ in the other chain.  All other $(v_t,c)$ updates are coupled using the identity coupling.
Now each neighbor can create at most one new disagreement (from color $c_Y$ in $X_t$ and $c_X$ in $Y_t$).  Therefore, we need $k-\Delta>\Delta$, which yields Jerrum's $k>2\Delta$ bound.

Subsequent improvements rely on first burning-in the chains to avoid worst-case configurations, as utilized in \cite{dyer2003randomly, Molloy, hayes-old-girth5, DFHV}. In a star graph on $\Delta+1$ vertices with center vertex $v$, the expected number of available colors for $v$ in a random coloring is $\approx k\exp(-\Delta/k)$, since a color does not appear in the neighborhood with probability $(1-1/k)^\Delta\approx \exp(-\Delta/k)$.  In fact,  Dyer and Frieze \cite{dyer2003randomly} showed such a lower bound on the expected number of available colors for any triangle-free graph.  
This heuristic leads to the condition $k\exp(-\Delta/k)>\Delta$, which holds for $k>1.764\Delta$.
This forms the basis of rapid mixing results for graphs of girth at least $5$; see \cite{dyer2003randomly,hayes-old-girth5,DFHV}.
The additional girth assumption ensures that these estimates hold for the burned-in dynamics, rather than only in the stationary distribution.
  
Molloy~\cite{Molloy}, and subsequently Hayes~\cite{hayes-old-girth5,hayes2013local}, further improved this bound by proving stronger burn-in properties on the distance-2 neighborhood of a vertex (and hence further requiring girth $\geq 7$).  In particular, for a neighbor $w\in N(z)$ of the disagreement, by considering the probability that both colors $c_X$ and $c_Y$ appear in $N(w)\setminus\{z\}$, and hence there is no bad update for $w$, we obtain the bound $k\exp(-\Delta/k)>\Delta(1-\exp(-\Delta/k))^2$, which holds for $k>1.489\Delta$. 

The remaining obstacle in the near-critical regime is handling singly-blocked configurations.  Suppose that the disagree color $c_Y$ appears in $N(w)\setminus\{z\}$, and the other disagree color $c_X$ does not appear.  
Then the coupled update for $w$ which attempts color $c_Y$ in $X_t$ and color $c_X$ in $Y_t$ is blocked in $X_t$ but succeeds in $Y_t$, creating a new disagreement.  

Suppose instead we could modify the coloring of $Y_t$ on $N(w)\setminus\{z\}$ so that $c_X$ appears and $c_Y$ does not (which is the opposite of $X_t$).  Then the same coupled update would be blocked in both chains, preventing this disagreement.  
Consequently, a new disagreement at $w$ can only occur if $c_Y$ does not appear in $N(w)\setminus\{z\}$ in $X_t$ (and symmetrically for $c_X$ in $Y_t$).
  The coupling then requires $k\exp(-\Delta/k)>\Delta\exp(-\Delta/k)$, which holds for $k>\Delta$.

The goal of the non-Markovian coupling is to realize this ideal:
 when $w\in N(z)$ is singly blocked for one of the disagree colors in $X_t$ then we'll couple the update sequence for $Y_t$ (at earlier times) so that it is singly blocked for the other disagree color.  The challenge is ensuring the disagreements we introduce to obtain this do not propagate; this requires looking at and modifying updates at future times, hence the non-Markovian aspect of the coupling (see \cref{fig:nm} for an illustration). 

To implement this strategy, we first bring the chains into a well-behaved regime via a burn-in phase,
 after which the refined coupling can be applied effectively.

\begin{figure}
    \centering
    \includegraphics[width=0.5\linewidth]{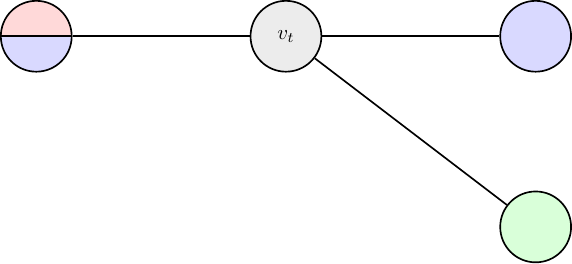}
    \caption{This chain is following the Jerrum coupling.  Vertex $v_t$ will attempt (blue/red) in the upper and lower chains respectively, leaving it as a (gray/red) discrepancy after time $t$.}
    \label{fig:placeholder}
\includegraphics[width=0.5\linewidth]{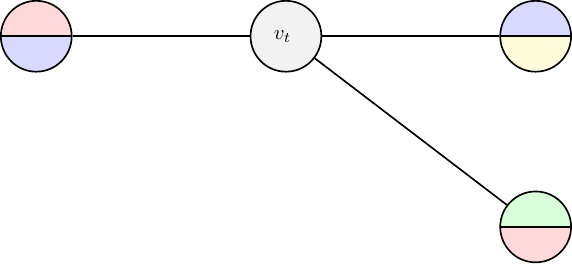}
    \caption{This chain is following the non-Markovian coupling.  By changing past updates, when $v_t$ attempts (blue/red) it will remain (gray/gray).  We have prevented this discrepancy at the cost of introducing two ``temporary'' discrepancies in the neighborhood which we have cleverly chosen so that they will be fixed.  The color ``yellow'' is somewhat arbitrary; it encodes that green was the original color of the other vertex to ensure validity of the coupling.}
    \label{fig:nm}
\end{figure}

 \medskip

{\bf \noindent Stage I: Burn-in Phase:}  
To couple a pair $(X_0,Y_0)$ that differ at a single vertex $z$, we first run a burn-in phase of $O(n\log{\Delta})$ steps.
The burn-in phase transforms worst-case initial colorings into configurations that are locally close to the stationary distribution; this is formalized via local uniformity.
During this phase, we show that the initial disagreement at $z$ does not spread beyond a ball of radius $\Delta^{c}$ for a small constant $c \in (0,1/10)$,
allowing us, via a union bound, to conclude that all vertices in this local ball satisfy local uniformity with high probability. A key aspect here is that for a fixed (sufficiently large) time and a fixed vertex, local uniformity properties hold with probability $1-\exp(-\Omega(\Delta))$, which is what permits us to take various union bounds over time and space. This aspect is also utilized in \cite{DFHV}. 
  
One of our key technical ingredients is establishing local uniformity for the Metropolis Glauber dynamics.
A local uniformity result shows that after a suitable burn-in period, the distribution of colors in the neighborhood of every vertex is close to the distribution induced by a random coloring of a tree rooted at that vertex.

Previously, Hayes~\cite{hayes2013local} established an analogous local uniformity result for the heat-bath Glauber dynamics.  Hayes's proof utilizes the similarity of the heat-bath dynamics to the stationary distribution, and it was unclear how to extend his proof approach to the Metropolis dynamics.  
In particular, in the heat-bath dynamics the vertex update probabilities are independent of the current coloring, whereas in the Metropolis dynamics the probability that a proposed recoloring succeeds depends on the number of available colors at the chosen vertex. Controlling this dependency on the current configuration is the central difficulty in establishing the burn-in result. We discuss this in a bit more detail in \cref{sub:corrections}.

We state here a simplified version of our burn-in result, 
and refer the interested reader to \cref{thm:local-uniformity-discrete} for the general statement.
For a coloring $X_t$, and a vertex $v$, let 
\[ A(X_t,v):=[k]\setminus X_t(N(v))
\]
denote the set of {\em available colors}.  The following result shows that after a suitable burn-in period, the number of available colors is determined only by the degree of $v$.  The burn-in result requires girth $g\geq 7$ as in Hayes~\cite{hayes2013local}, where the girth is the number of edges in the shortest cycle.

\begin{theorem}\label{thm:SIMPLER-local-uniformity-discrete}
    Given $\delta,\eps > 0$, there exist $\Delta_0 = \Delta_0(\delta,\eps),C = C(\delta,\eps)$  such that for all $\Delta > \Delta_0$, $k \geq (1+\delta)\Delta$, and for any graph $G = (V,E)$ of maximum degree $\Delta$ and girth $g\geq 7$
    we have the following.  Let $\{X_t\}_{t \geq 0}$ be the discrete-time Metropolis Glauber dynamics.  Fix $v \in V$. Suppose $T \ge n$ and $T_0 \ge C n \log \Delta$. Then,
    \begin{equation}\label{eq:simpler-main-number-colors}
        \mb{P}\left[ \exists t \in [T_0,T_0+T] : \big||A(X_t,v)| - k e^{-d(v)/k} \big| > \eps k\right] \leq \frac Tn e^{-\Delta/C}.
    \end{equation}
\end{theorem}

    Before proceeding, let us interpret the conclusion of \cref{thm:SIMPLER-local-uniformity-discrete}.  As noted earlier, in a star graph the expected number of available colors for the center vertex $v$ is $ \approx k\exp(-d(v)/k)$.  
     \eqref{eq:simpler-main-number-colors} 
     shows that the number of available colors at $v$ is sharply concentrated around this quantity.  

The refinement of \cref{thm:SIMPLER-local-uniformity-discrete} stated in \cref{thm:local-uniformity-discrete} 
establishes a finer-grained description of the distribution
of colors in the neighborhood of a vertex $v$. The probability that a neighbor $w\in N(v)$ has exactly $i_1$ neighbors colored $c_1$ and $i_2$ neighbors colored $c_2$ is close to a Poisson distribution, 
which shows that correlations between colors in the neighborhood are negligible.

\begin{remark}
\label{rem:HV-error}
    Hayes~\cite{hayes2013local} established the analog of~\cref{thm:SIMPLER-local-uniformity-discrete} (and the more general, \cref{thm:local-uniformity-discrete}) for the heat-bath dynamics.  
    The result of Hayes and Vigoda~\cite{hayesvigoda-nonmarkovian} relied on a local uniformity result for Metropolis Glauber dynamics; however Hayes's proof approach does not extend to the Metropolis dynamics.
\end{remark}

{\bf \noindent Stage II: Non-Markovian Coupling:}
To translate local uniformity into rapid mixing, we refine the non-Markovian coupling introduced by Hayes and Vigoda~\cite{hayesvigoda-nonmarkovian}.   We provide a high-level overview of the coupling here, and identify the key differences from \cite{hayesvigoda-nonmarkovian} in the subsequent subsection \cref{sub:corrections}.

Fix a pair $X_0,Y_0$ that differ only at a single vertex $z$.  Let $\Tcouple$ denote the coupling time, and   
let \[ \vecsigma=((v_1,c_1),\dots,(v_{\Tcouple},c_{\Tcouple}))
\]
denote the update sequence defining the evolution from $X_0$ to $X_1,\dots,X_{\Tcouple}$.  Thus, in the
transition $X_{t-1}\rightarrow X_{t}$ we recolor~$v_t$ to color $c_t$ if $c_t\notin X_{t-1}(N(v_t))$.  

Our coupling defines a coupled update sequence:
\[ F(\vecsigma)=\vecsigma'=((v_1,c'_1),\dots,(v_{\Tcouple},c'_{\Tcouple})),
\]
which will define the evolution from $Y_0$ to $Y_1,\dots,Y_{\Tcouple}$.  Notice that the vertex updates are the same for both chains, only the color updates possibly differ.  We will ensure that $F$ is bijective and hence this is a valid coupling.

We can now state our main coupling result, which we present in a slightly informal manner here in the introduction, and we refer the interested reader to a more formal statement in \cref{sec:nonmarkovian}, which is presented after presenting the necessary definitions.

\begin{theorem}
    \label{thm:SIMPLER-main-nonmarkovian}
Let $\delta>0$, and let $G=(V,E)$ be a graph of maximum degree $\Delta$ and girth $\geq 7$, where $\Delta>\Delta_0(\delta)$.  Let  $k\geq(1+\delta)\Delta$.
    Consider a pair of colorings $X_0,Y_0$ which differ at a single vertex $z$ and 
    which are ``burned in'' in the ball of radius $\Delta^{1/10}$
    around $z$.
Then, for $\Tcouple=O_{\delta}(n)$,
there exists a $\Tcouple$-step coupling of the Glauber dynamics where 
    \[  \mb{E}[|X_{\Tcouple}\oplus Y_{\Tcouple}|] \leq 1/3.
\]
\end{theorem}

A more precise version of \cref{thm:SIMPLER-main-nonmarkovian} is stated in \cref{thm:main-nonmarkovian} in \cref{sec:nonmarkovian}, where the coupling is formally defined and the requisite burn-in properties are formalized.  We provide a high-level description of the coupling here.

We construct our global coupling $F$ by composing local non-Markovian couplings. 
For each time $1\le t\le \Tcouple$ we define a local coupling $f_t$ which may modify updates at multiple times, both earlier and later than $t$.
Because these transformations act beyond the current step, the resulting construction is referred to as a non-Markovian coupling.

The high-level idea of the non-Markovian coupling is as follows.  Suppose there is an update of vertex $v_t$ at time $t$ which potentially propagates a disagreement.  We will ensure our disagreements propagate in a tree-like manner (in order to apply our non-Markovian coupling) and hence there will be a unique neighbor $p\in N(v_t)$ which differs.  Let $\{c_b,c_u\}=\{X_{t-1}(p),Y_{t-1}(p)\}$ denote the disagree colors.  The non-Markovian coupling will apply in the case where one of the disagree colors, say $c_b$, appears in $X_{t-1}(N(v_t)\setminus\{p\})$, and hence is a blocked color; additionally, the other color $c_u$ does not appear among the neighbors $N(v_t)\setminus\{p\}$, and hence is an unblocked color.

In our non-Markovian coupling we create a pair of temporary disagreements in the neighborhood of $v_t$ so that in chain $Y_{t-1}$ the color $c_u$ appears and $c_b$ no longer appears (and hence in $Y_t$ color $c_u$ is now blocked and $c_b$ is unblocked, thus the chain $Y_t$ has the opposite behavior as $X_t$ with respect to this pair of colors).  Consequently, the coupled update $c_b$ in $X_t$ and $c_u$ in $Y_t$ is blocked in both chains, thereby preventing a disagreement that would occur under the standard coupling; this is the key mechanism enabling contraction in the near-critical regime, as outlined earlier.

Creating these temporary disagreements in a bijective manner is non-trivial.
Moreover, we must ensure that these temporary disagreements do not propagate.
We defer the detailed construction to \cref{sub:coupling-overview}.

Our construction differs significantly from Hayes and Vigoda~\cite{hayesvigoda-nonmarkovian}, where the non-Markovian coupling is applied either at all relevant times or not at all.
When $\Delta$ is constant, a non-Markovian update may fail (i.e., cease to be bijective) with constant probability.
Consequently, our global coupling, which is the composition of local couplings at each time $t$,
decides for each relevant time whether to apply the local non-Markovian coupling (when it is well-defined), or use the identity/Jerrum coupling if not. 

We discuss some more aspects of our construction and analysis, and how they contrast with \cite{hayesvigoda-nonmarkovian} in the next subsection.

\subsection{Comparison with previous work}
\label{sub:corrections}

We highlight several key issues in the construction and analysis of
\cite{hayesvigoda-nonmarkovian} and explain how our work addresses them.

\medskip

\paragraph{\bf Bounding chain}

To apply our non-Markovian coupling we need to ensure the associated updates are sufficiently disjoint from each other.  To achieve that we define a bounding chain $Z_t$ (\cref{sub:bounding-chain}) where each vertex has a subset of colors at each time, 
capturing the possible colors of $X_t$ and $Y_t$ under all coupling choices.  The set of vertices with more than one color at a time, denoted as $\mc{P}$ for {\em potential persistent
disagreements},  will be a superset of the disagreements between $X_t$ and $Y_t$ (excluding the temporary disagreements introduced during the non-Markovian coupling).  

Hayes--Vigoda utilize a different bounding chain, where the color sets $|Z_t(v)|$ are of size at most $2$. However, for the bounding chain to be sufficiently powerful to capture all potentially propagating disagreements, one must allow the color set sizes $|Z_t(v)|$ to grow, which we incorporate in our construction.
This presents an additional obstacle for maintaining that $|\mc{P}|$ is sufficiently small with high probability (see \cref{lem:P-total-size-bound}).

\medskip

\paragraph{\bf Validity of the coupling}

A central requirement of the non-Markovian framework is that the global
coupling map is bijective. In Hayes and Vigoda~\cite{hayesvigoda-nonmarkovian},
this property is argued via a more intricate global construction.
In contrast, our formulation yields a simpler and more transparent validity
argument: by composing local transformations that are each explicitly
bijective, we obtain a global coupling whose bijectivity follows directly.
This local structure is also essential for handling failures of individual
non-Markovian updates in the constant-degree regime. See \cref{sec:validity} for details. 

\medskip

\paragraph{\bf Measurability and drift analysis}

Our aim is to analyze the expected Hamming distance at time~$t$ given the coupled updates at all times prior to $t$. However the coupling is looking ahead (and potentially modifying updates) at times after $t$; this appears somewhat contradictory.  To overcome this, in Hayes--Vigoda they condition on the event that all non-Markovian updates succeed. 
However, the coupling in Hayes--Vigoda restricts non-Markovian updates to a smaller time interval $\Tmod=O(n)$ around the particular update, which leads to higher failure probability of non-Markovian updates than claimed in \cite[Theorem~17]{hayesvigoda-nonmarkovian}; the subsequent drift analysis is therefore not complete as stated.

We rectify these issues as a byproduct of our improved coupling, which eliminates the smaller $\Tmod$ window, and our enhanced analysis; see \cref{sec:analysis-coupling} for details. 

\medskip

\paragraph{\bf Local uniformity}
\label{sub:proof-overview-uniformity}

Let us briefly describe the proof overview of \cref{thm:SIMPLER-local-uniformity-discrete} (and its more formal version \cref{thm:local-uniformity-discrete}) which shows that after a burn-in period, the Metropolis dynamics exhibits certain local uniformity properties.  
Namely, we show that the distribution of colors in the neighborhood of a vertex
is well approximated for certain relevant statistics by the distribution induced by a random coloring of a tree
rooted at that vertex.

Similar results were proved for the heat-bath dynamics by Hayes~\cite{hayes2013local}.  We largely follow the proof strategy of \cite{hayes2013local}.  The main strategy is to use the continuous time dynamics, fix a distinguished vertex~$v$, and switch to a directed graph $G^* = G_\mr{in}(v,3)$ (see \cref{defn:directed-neighborhoods})
in which edges within $B_3(v)$ are oriented towards $v$ and all other edges are bidirected.  
Then, conditioning on $\mc{F}$, the $\sigma$-algebra generated by the configuration outside $B_2(v)$, the colors of the vertices $w\in N(v)$ become independent.
As in \cite{hayes2013local}, this auxiliary process allows us to analyze the local behavior around a fixed vertex $v$, after which we compare it to the original Glauber dynamics via a coupling argument.

However, analyzing the Metropolis dynamics introduces an additional difficulty.
For the continuous-time version of the heat-bath dynamics, each vertex $u$ has a Poisson clock of rate $1$, and when its clock ticks it chooses a uniform random color in $A(X_t,u)$.  For the Metropolis dynamics,
we instead give each vertex $u$ a Poisson clock of rate $|A(X_t,u)|/k$, and when the clock ticks we choose a uniform random color in $A(X_t,u)$; notice that this is the same process as the Metropolis dynamics but corresponds to ignoring rejected proposals.
However, information about the Poisson ringing of $u$ 
reveals information about the configuration in $N(u)$.
  To circumvent this dependency, we instead introduce another auxiliary process in which $u$ has a Poisson clock of rate $\mb{E}[|A(X_t,u)| \mid \mc{F}] /k$.  
  Under this auxiliary process, the events we analyze become conditionally independent given $\mc{F}$.
  The auxiliary process can then be coupled with the standard Metropolis dynamics using Poisson thinning.

  This auxiliary process allows us to recover the conditional independence structure used in Hayes's analysis. 
Using this framework, we establish the concentration bounds on $|A(X_t,u)|$ stated in \cref{thm:SIMPLER-local-uniformity-discrete} (and more generally in \cref{thm:local-uniformity-discrete}), see \cref{par:local-uniformity-overview} for more detailed discussion.
Similar auxiliary processes are also used to obtain conditional independence for various properties of our non-Markovian couplings; see \cref{sec:non-markovian-inequalities}.

\medskip

\paragraph{\bf Girth requirement.} The local uniformity argument requires tree-like structure up to radius $3$ (which translates to girth $g\geq 7$),
while the coupling analysis requires that the discrepancy region remains acyclic and its neighborhood remains acyclic (see \cref{rmk:girth}).
These constraints combine to yield the girth
$\ge 7$ condition in our main result, namely \cref{thm:main-mixing}.

\subsection{Notation}
We will use standard graph theory notation. For a graph $G = (V,E)$, a vertex $v \in V$ and a non-negative integer $r$, we let $B_r(v)$ denote the ball of radius $r$ around $v$, i.e.~the set of vertices which can be reached from $v$ in $r$ edges or fewer. We extend this to $T \subseteq V$ via $B_r(T) = \bigcup_{v \in T} B_r(v)$.  We also define the sphere $S_r(T) = B_r(T) \setminus B_{r-1}(T)$ and the neighborhood $N^r(T) = B_r(T) \setminus T$. For a subset $T\subseteq V$, we let $G[T]$ denote the induced graph on $T$.
For labelings $X,Y:V\to[k]$, write $X\oplus Y:=\{v\in V:X(v)\ne Y(v)\}$.

Throughout, we will let $[N]$ denote the interval $\{1,\dots, N\}$. We will also make use of asymptotic notation. For functions $f,g$, $f = O_\alpha(g)$ means that $f \le C_\alpha g$, where $C_\alpha$ is some constant depending on $\alpha$; $f = \Omega_\alpha(g)$  means that $f \ge c_\alpha g$, where $c_\alpha > 0$ is some constant depending on $\alpha$, and $f = \Theta_\alpha(g)$ means that both $f = O_\alpha(g)$ and $f = \Omega_\alpha(g)$ hold. For parameters $\varepsilon, \delta$, we write $\varepsilon \ll \delta$ to mean that $\varepsilon \le c(\delta)$ for a sufficient function $c$. A chain $\alpha \ll \beta \ll \gamma$ should be read from right to left.

\subsection{Outline of Paper}

In the following \cref{sec:nonmarkovian} we define our coupling, including the bounding chain, the local non-Markovian coupling, and the global coupling which is a composition of local couplings. We provide an overview of the construction in \cref{sub:coupling-overview}. Then in \cref{sec:validity} we prove that the coupling we constructed is in fact a valid coupling.  We then analyze the coupling in \cref{sec:analysis-coupling}, and conclude \cref{thm:main-nonmarkovian} (the formal version of \cref{thm:SIMPLER-main-nonmarkovian}) which shows that the coupling contracts from a ``burned-in'' configuration.  We prove the local uniformity properties, including \cref{thm:SIMPLER-local-uniformity-discrete}, in \cref{sec:local-uniformity}. In \cref{sec:non-markovian-inequalities}, we prove certain properties of the non-Markovian coupling which require the decoupling techniques from \cref{sec:local-uniformity} and are stated and used in \cref{sec:analysis-coupling}.  Finally, in \cref{sec:DFHV}, we combine the contracting coupling result from \cref{sec:analysis-coupling} with the local uniformity results in \cref{sec:local-uniformity} to conclude the main result of \cref{thm:main-mixing}, following the burn-in and contraction framework of Dyer--Frieze--Hayes--Vigoda \cite{DFHV}.

\subsection*{Acknowledgments}

The authors met at the Rocky Mountain Summer Workshop in Algorithms, Probability, and Combinatorics at Colorado State University.  V.J.~is partially supported by NSF grant DMS-2237646. C.M.~is partially supported by a Simons Dissertation Fellowship.  E.V.~is partially supported by NSF grant CCF~2147094.

The authors acknowledge the use of various versions of ChatGPT and Gemini for proofreading and copyediting. The mathematical content is entirely due to the authors.

\section{Construction of the Non-Markovian Coupling}
\label{sec:nonmarkovian}

In the next sections we will prove our main coupling result, which is stated informally in \cref{thm:SIMPLER-main-nonmarkovian}, and is stated formally in \cref{thm:main-nonmarkovian} below.  Recall, \cref{thm:SIMPLER-main-nonmarkovian} says that for a ``burned-in'' pair of initial states disagreeing at a single vertex, there is a coupling of length $O(n)$
so that the coupled pair contracts in terms of the expected Hamming distance.  This coupling result is the main ingredient in the subsequent fast mixing results, see~\cref{sec:DFHV}. 

The goal of this section is to construct the appropriate coupling.  In \cref{sub:coupling-overview} we provide an overview of the coupling.  We begin by presenting the bounding chain in \cref{sub:bounding-chain}.  Then we present the local non-Markovian coupling in \cref{sub:non-markovian-local-coupling}.  Finally, we compose a set of local couplings to obtain our global coupling $F$ in \cref{sub:global}.

The proof that the coupling is, in fact, a valid coupling (which we do by showing it is a bijective map) is presented in \cref{sec:validity}.
We then prove that the coupling contracts, with respect to the expected Hamming distance, in \cref{sec:analysis-coupling}, which will complete the proof of \cref{thm:main-nonmarkovian}.

In order to apply path coupling with respect to Hamming distance for colorings, we need to expand the state space as is done in all prior works in the literature using path coupling for colorings.  We consider the auxiliary state space:
\[
    \widehat\Omega := [k]^V,
\]
which consists of all labelings, not only the (proper) colorings. The Metropolis update rule extends in a straightforward manner to $\widehat\Omega$.  From a labeling $X_t\in\widehat\Omega$, choose a vertex $u\in V$ and color $c\in [k]$ uniformly at random.  For all $w \neq u$, we set $X_{t+1}(w) = X_{t}(w)$, and we update the color of $u$ as follows:
\[ X_{t+1}(u) = \begin{cases}
c & \mbox{if } c \notin X_t(N(u)) \\
X_t(u) & \mbox{if } c \in X_t(N(u)) 
\end{cases}
\]

Notice that if $X_t\in\Omega$ (i.e., $X_t$ is a proper coloring) then $X_{t+1}\in\Omega$. Thus, $\Omega$ is a closed class for the chain on $\wh{\Omega}$, and the restriction of the chain to $\Omega$ is exactly the original Metropolis Glauber dynamics, which has the uniform distribution on $\Omega$ as its unique stationary distribution.  Therefore, any mixing-time upper bound proved for the extended update rule, when started from states in $\Omega$, yields the same mixing-time upper bound for the original chain. The expanded state space $\wh{\Omega}$ is used only in the path-coupling argument.

We now present the formal version of \cref{thm:SIMPLER-main-nonmarkovian}. First, we need to quantify what we mean by ``burned-in''. For a labeling $X \in \wh{\Omega}$ and a vertex $u \in V$, we write
\[A(X,u) := [k] \setminus X(N(u))\]
for the set of colors available at $u$ under $X$. 

\begin{definition}
\label{def:local-uniformity}
Let $G = (V,E)$ be a graph of maximum degree $\Delta$. 
We say that a labeling $X \in \wh{\Omega}$ is $\eps$-uniform at $z^* \in V$ for radius $R$  if for all $v \in B_R(z^*)$, the following three conditions hold. 
\begin{equation}
\label{eq:main-number-colors}
\big| |A(X,v)| - k e^{-d(v)/k} \big| \leq \eps k; 
\end{equation}
for all colors $c \in [k]$,
\begin{equation}\label{eq:main-blockers} \left| \paren{\sum_{\substack{w \in N(v)  \\ c \in A_v(X,w)}} e^{d(w)/k}} - d(v) \right| < \eps\Delta; \end{equation}
where $A_v(X,w) = [k] \setminus X(N(w) \setminus \{v\})$
and for all $c \in [k]$,
\begin{equation}
\label{eq:two-nbd-uniformity}
|X^{-1}(c) \cap B_2(v)| \leq 400\Delta.
\end{equation}
\end{definition}

\begin{remark}
    If neighbors of a vertex $v$ choose colors uniformly and independently, the probability a color $c$ is missing from $N(v)$ is $(1-1/k)^{d(v)} \approx e^{-d(v)/k}$. Summing over $k$ colors yields the expectation $k e^{-d(v)/k}$ for the number of missing colors from $N(v)$. \eqref{eq:main-number-colors} requires that the number of available colors at $v$ is approximately equal to this quantity. \eqref{eq:main-blockers} requires a finer-grained control over the local structure and demands.
\end{remark}

\begin{definition}
\label{def:event-LU}
Let $G = (V,E)$ be a graph of maximum degree $\Delta$. For a labeling $X \in \wh{\Omega}$ and an update sequence \[ \vecsigma=((v_1,c_1),\dots,(v_{T},c_{T}))\]
we say that $\LU(X_0, \vec{\sigma}, \eps, z^*)$ holds if for all $0\leq t \leq T$, $X_{t}$ is $\eps$-uniform at $z^* \in V$ for radius $\Delta^{1/10}$.
\end{definition}

\begin{theorem}
\label{thm:main-nonmarkovian}
For every $\delta > 0$, there exists constants
\[1/\Delta_0 \ll \eps \ll 1/\Ccouple \ll \delta\]
such that the following holds. 
Let $G=(V,E)$ be a graph on $n$ vertices of maximum degree $\Delta\ge \Delta_0(\delta)$ and girth at least $7$, and let $k\ge (1+\delta)\Delta$.
Let $X_0,Y_0\in\wh\Omega$ be neighboring labelings with unique disagreement at vertex $z^* \in V$. 
Set
\[
    \Tcouple := \Ccouple\,n.
\]
Then there exists a $T_{\mathrm{cp}}$-step coupling of the Metropolis dynamics such that
\[
    \mb E\bigl[|X_\Tcouple \oplus Y_\Tcouple| \cdot \mbf1\{\LU(X_0,\vecsigma,\eps,z^*) \}\bigr] \le \frac13.
\]
\end{theorem}

\subsection{Coupling Overview}
\label{sub:coupling-overview}

Fix a pair of labelings $X_0,Y_0 \in \wh{\Omega}$ that differ only at a single vertex $z^*$.
Let $\Tcouple$ denote the coupling time from \cref{thm:main-nonmarkovian}.  
Let \[ \vecsigma=((v_1,c_1),\dots,(v_{\Tcouple},c_{\Tcouple}))
\]
denote the update sequence defining $X_0,\dots,X_{\Tcouple}$.  
Thus, in the
transition $X_{t-1}\rightarrow X_{t}$ we recolor~$v_t$ to color $c_t$ if $c_t\notin X_{t-1}(N(v_t))$.

Given $X_0,Y_0$ and the update sequence $\vecsigma$, we will define a coupled update sequence:
\[ \vecsigma'=((v_1,c'_1),\dots,(v_{\Tcouple},c'_{\Tcouple})),
\]
which will define $Y_0,\dots,Y_{\Tcouple}$.  Thus the two chains always update the same vertex $v_t$ at time $t$; only the
proposed colors may differ.

Before defining $\vecsigma'$, we first construct an auxiliary process
$Z_t$, called the {bounding chain} (\cref{sub:bounding-chain}). This process depends only on
$X_0,Y_0$, and the original sequence $\vecsigma$.
It is not itself a labeling: rather, for each vertex $v$, the value $Z_t(v)$ is a set of colors, i.e.
$Z_t(v)\subseteq [k]$.
The role of $Z_t$ is to identify the region where disagreements between the two
chains may persist. More precisely, once the coupling has been defined, we will
show that (see \cref{prop:main-global-induction-rewrite})
\[
\{X_t(v),Y_t(v)\}\subseteq Z_t(v)
\]
for every vertex $v$ and every time $t\le \Tcouple$, except for the temporary
disagreements that are intentionally created during certain non-Markovian edits.
Consequently, apart from these temporary disagreements, any actual disagreement
between $X_t$ and $Y_t$ can occur only at a vertex $v$ with $|Z_t(v)|>1$.

This motivates the definition
\[
\mc{P}_t:= \{v\in V:\ |Z_t(v)|>1\},
\]
and we write
\[
\mc{P}:=\bigcup_{t\le \Tcouple}\mc{P}_t, \qquad \mc{P}_{\leq t} = \bigcup_{s\leq t} \mc{P}_s.
\]
The set $\mc{P}_t$ should be viewed as a set of {\em potential persistent
disagreements} at time $t$. The key point is that $\mc{P}_t$ is determined before we define
the coupled sequence $\vecsigma'$: it is built only from $X_0,Y_0$, and the update sequence $\vecsigma$, and it is designed to dominate the disagreements in every coupling
that we will later allow (again, ignoring the temporary disagreements deliberately introduced
inside a non-Markovian edit.)

The geometry of $\mc{P}$ determines whether we attempt the non-Markovian
construction. If the induced subgraph on $\mc{P}$ together with its local
neighborhood contains a cycle, then we do not use any non-trivial edits: we
simply take the identity coupling at all times, i.e.,
\[
\vecsigma'=\vecsigma.
\]

If, on the other hand, the subgraph is tree-like, 
then we use a non-trivial coupling.  Formally, later (\cref{sub:global}) we will track a set $D_t \subseteq \mc{P}_t$ of actual persistent
disagreements. A time $t$ is potentially propagating if the updated vertex
$v_t$ lies on the boundary of the current disagreement set, say, with unique
neighbor $p \in D_{t-1}$, and the proposed color $c_t$ is exactly the disagreeing
color across that edge.
In this situation the disagreement at $p$ can spread to $v_t$. These are the times at which we
consider a non-Markovian edit. In particular, they are among the times when
the bounding chain indicates that the set of potential disagreements may grow.

At all other times we use simpler local couplings. If $v_t$ is far from the
current disagreement set, or if $v_t$ already belongs to it, then we use the
identity coupling. If $v_t$ is adjacent to the current disagreement set but the
update is not potentially propagating, then we use the standard one-step
Jerrum coupling at time $t$; this modifies only the color choice at time $t$
and leaves all other times unchanged. Even at a potentially propagating time,
we apply the non-Markovian map only when an additional local well-definedness
condition is satisfied; otherwise we again fall back to Jerrum's coupling.

In summary, the full coupling will be defined as a composition of local maps $f_t$.
There are three kinds of local maps to keep in mind. The first is the
identity coupling, which leaves the color update unchanged. The second is a
standard one-step coupling, namely Jerrum's coupling, where the coupled color
choice at time $t$ is determined from the update at time $t$ and the past
history. The third is our non-Markovian coupling: at certain exceptional times, $f_t$ reads (and writes to) more coordinates than just $c_t$, so we might change the color choices $c_{t'}$ for times $t'\neq t$ and some of these times might be in the future $t'>t$. However, none of these
local maps ever changes the vertex-update sequence $(v_t)$.

We will define the full coupling
\[
F=F_{X_0,Y_0}:(V\times [k])^{\Tcouple}\to (V\times [k])^{\Tcouple}
\]
by composing these local maps at the appropriate times.

\subsection{Constructing the bounding chain}
\label{sub:bounding-chain}
To apply our non-Markovian coupling, we want to ensure that the set of vertices involved in the non-Markovian updates induce an essentially tree-like subgraph.  To this end, we will now introduce a ``bounding chain'' that defines a set $\mc{P}$ which is a superset of all discrepancies that will occur between the subsequently coupled chains $(X_t)$ and $(Y_t)$ (ignoring temporary disagreements).  We will then define a predicate $\BC(X_0,Y_0,\vecsigma)$  which is $\tt{True}$ if certain conditions are satisfied (roughly, the discrepancy region $\mc{P}$ and its local neighborhood evolve in a tree-like manner), and is $\tt{False}$ otherwise.

Let $X_0, Y_0\in\wh{\Omega}$ where $X_0\oplus Y_0=\{z^*\}$ 
for some $z^*\in V$, and let $\vecsigma=\vecsigma_{\Tcouple}=((v_1,c_1),\dots,(v_{\Tcouple},c_{\Tcouple}))$
denote the update sequence defining $X_0,\dots,X_{\Tcouple}$.  For each such $X_0,Y_0,\vecsigma$, we will define a bounding chain $(Z_t)$ where for all $v\in V, 0\leq t\leq\Tcouple$, 
\[
Z_t(v)\subseteq [k].
\]

   Let $Z_0$ be defined by $Z_0(v) = \{ X_0(v), Y_0(v)\}$ for all $v\in V$. In particular, $|Z_0(z^*)| = 2$ and $|Z_0(v)| = 1$ for all $v\neq z^*$.  

   For $t>0$ we define $Z_t$ inductively as follows.
   For all $v\in V$, we decompose the colors $[k]$ into three disjoint sets $A(Z_{t-1},v),B(Z_{t-1},v),H(Z_{t-1},v)$ which correspond to the available, blocked, and potentially hazardous colors defined as follows:
    \begin{align*} 
    A(Z_{t-1},v) & = [k] \setminus \bigcup_{w \in N(v)} Z_{t-1}(w) = \{c\in [k]: c\notin Z_{t-1}(w)\mbox{ for any }w\in N(v)\},
    \\
         H(Z_{t-1},v) & = \bigcup_{\substack{w \in N(v): \\ |Z_{t-1}(w)| > 1}} Z_{t-1}(w),
         \\
         B(Z_{t-1},v) &= [k]\setminus (A(Z_{t-1},v)\cup H(Z_{t-1},v)).  \text{ Thus, } B(Z_{t-1},v) \subseteq \bigcup_{\substack{w \in N(v):\\ |Z_{t-1}(w)|=1}} Z_{t-1}(w)
    \end{align*}
      We emphasize that hazardous colors take precedence over blocked colors because $H(Z_{t-1},v)$ is
meant to record the colors through which unresolved uncertainty can propagate: if $v$ has a neighbor $w$ such that $Z_{t-1}(w) = \{R,B\}$ and another neighbor $w'$ such that $Z_{t-1}(w') = \{R\}$, then color $R$ is included in $H(Z_{t-1},v)$ 
    and considered hazardous, not blocked. This is a deliberate over-approximation, and it ensures that later
local couplings can swap between colors carried by a discrepancy without
changing the evolution of the bounding chain. 
    
    We define $Z_t$, for $1\leq t \leq \Tcouple$, based on the update $(v_t,c_t)$ as follows:
    \begin{enumerate}
       \item For all $w\neq v_t$, set $Z_t(w)=Z_{t-1}(w)$.
        \item If $c_t\in A(Z_{t-1},v_t)$, set $Z_{t}(v_t) = \{c_t\}$.
        \item If $c_t \in B(Z_{t-1},v_t)$, set $Z_{t}(v_t) = Z_{t-1}(v_t)$.
        \item If $c_t\in H(Z_{t-1},v_t)$, do the following.
        \begin{enumerate}
            \item If $v_t \in \mc{P}_{{t-1}}$, set $Z_{t}(v_t) = Z_{t-1}(v_t)$.
            \item If $v_t \notin \mc{P}_{{t-1}}$, then choose
            \[ { w \in \{ u \in N(v_t) \cap \mc{P}_{t-1} : c_t \in Z_{t-1}(u) \} } \]
            according to some fixed, but otherwise arbitrary, total order.  Set $Z_{t}(v_t) = Z_{t-1}(w) \cup Z_{t-1}(v_t)$.
        \end{enumerate}
    \end{enumerate}

    For an intuitive understanding of these update rules, note that the sets $A(Z_{t-1},v)$, $B(Z_{t-1},v)$, and $H(Z_{t-1},v)$
represent three different levels of certainty about what happens if the color
$c_t$ is proposed at the vertex $v_t$. If $c_t \in A(Z_{t-1},v_t)$, then no neighbor of $v_t$ can possibly use
$c_t$, according to the information recorded in $Z_{t-1}$.  Thus the proposal
is certainly legal, so after the update the color of $v_t$ is completely
determined, and we may collapse to $
Z_t(v_t)=\{c_t\}$. If $c_t \in B(Z_{t-1},v_t)$, then $c_t$ is not available, but it also does not
appear in the color set of any unresolved neighbor.  In this case no new uncertainty is created at $v_t$, so we simply
leave its set of possible colors unchanged:
$
Z_t(v_t)=Z_{t-1}(v_t).
$ The interesting case is when $c_t \in H(Z_{t-1},v_t)$.  By definition, this
means that some neighbor $w$ of $v_t$ has $|Z_{t-1}(w)|>1$ and lists
$c_t$ as a possible color.  Thus the proposal at $v_t$ interacts with an
already unresolved part of the configuration.  This is precisely the mechanism
by which uncertainty can spread from one vertex to another.  In this situation
we deliberately over-approximate: rather than trying to determine the exact
outcome from the partial information in $Z_{t-1}$, we allow $v_t$ both to keep
any of its previous possible colors and also to acquire any color that is
currently carried by an unresolved neighbor.

The purpose of the bounding chain is to track the region where
disagreements between the coupled chains may persist.
We therefore define the set of potential persistent disagreements using $Z_t$ as follows. For $0\leq t\leq\Tcouple$, let 
\[  \mc{P}_t = \{v\in V: |Z_t(v)|>1\},
\]
and 
\[
\mc{P} = \cup_{0\le t\le \Tcouple} \mc{P}_t, \qquad \mc{P}_{\leq t} = \bigcup_{s\leq t}\mc{P}_s.
\]
Note, $\mc{P}_0=\{z^*\}$ where $z^*$ is the initial disagreement between $X_0$ and $Y_0$ and the set $\mc{P}_{\leq s}$ ``grows'' from $z^*$. The next condition captures that this growth happens in a tree-like manner, outwards from $z^*$.

For the bounding chain $Z_t$ on inputs $X_0,Y_0,\vecsigma$, we set $\BC(X_0,Y_0,\vecsigma)=\tt{True}$ if the following properties hold:
    \begin{enumerate}
        \item
        $G[\mc{P}]$
        is acyclic.  Furthermore, there is no vertex $v$ so that $G[\mc{P} \cup \{v\}]$ has a cycle.
        \item $\forall\, 1\leq t\leq\Tcouple, v_t\in \cup_{s<t} \mc{P}_s \implies c_t\notin H(Z_{t-1},v_t)$. 
        \item We have the size bound $|\mc{P}| \leq \exp(\exp(O(\Ccouple)))$, where the implicit constant in $O(\cdot)$ is universal, independent of any parameters.
    \end{enumerate}
    Otherwise, 
we set $\BC(X_0,Y_0,\vecsigma)=\tt{False}$.

The condition $\BC(X_0,Y_0,\vecsigma)=\texttt{True}$ singles out the update
sequences for which the potential disagreement region evolves in a controlled
way. Condition~(1) is a geometric requirement. 
This condition enforces a local tree-like quality to ensure that our discrepancies do not interact.
Condition~(2) is a dynamical requirement.  As the update rule above shows, a
hazardous update is the only kind of update that can enlarge the bounding chain vertex. The implication
\[
v_t\in \bigcup_{s<t}\mc P_s \implies c_t\notin H(Z_{t-1},v_t)
\]
says that if the color $c_t$ is potentially hazardous then the updated vertex was not a discrepancy at an earlier time, hence a vertex $v_t$ already in the discrepancy set cannot be updated with a hazardous color.  Intuitively, once a vertex has entered
$\mc P$, later updates at that vertex may resolve its uncertainty or leave it
unchanged, but they are not allowed to make it hazardous again.  Thus $\mc{P}_{\leq t}$
grows outward from $z^*$ in a ``one-pass manner'', rather than repeatedly
re-propagating through vertices that are already in the discrepancy set. Finally, condition (3) ensures that the bounding chain does not grow too rapidly, and will be used in \cref{sec:analysis-coupling} to show that our non-Markovian coupling is contractive. 

\begin{fact}
\label{lem:bc-basic-geometry-rewrite}
Assume $\BC(X_0,Y_0,\vec\sigma)=\tt{True}$.
\begin{enumerate}
\item For every $t$, the induced subgraph $G[\mc P_{\le t}]$ is connected.
\item If $v_t\in \mc P_t\setminus \mc P_{\le t-1}$, then $N(v_t)\cap \mc P_{t-1}$ is a singleton.
\end{enumerate}
In particular, every vertex that first enters $\mc P$ has a unique parent in the previously constructed set.
\end{fact}

\begin{proof}
For part~(1), note that $\mc P_{\le 0}=\{z^*\}$ is connected. Suppose inductively that $\mc P_{\le t-1}$ is connected. If $\mc P_t\subseteq \mc P_{\le t-1}$ there is nothing to prove. Otherwise let $x\in \mc P_t\setminus \mc P_{\le t-1}$. Then the update at time $t$ is hazardous, so by definition of the bounding chain there exists a neighbor $y\in N(x)$ with $|Z_{t-1}(y)|>1$, i.e.
$y\in \mc P_{t-1}\subseteq \mc P_{\le t-1}$. Thus every new vertex is attached by an edge to the previously constructed set, and connectivity is preserved.

For part~(2), existence of at least one neighbor in $\mc P_{t-1}$ follows from the same observation. For uniqueness, suppose that $v_t$ had two distinct neighbors $p,q\in \mc P_{t-1}$. By part~(1), the graph induced by $\mc P_{\le t-1}$ is connected, so there is a path in $G[\mc P_{\le t-1}]$ from $p$ to $q$. Together with the two edges $pv_t$ and $qv_t$, this creates a cycle in $G[\mc P \cup \{ v_t \}]$, contradicting Condition~(1) in the definition of $\BC$.
\end{proof}

\subsection{Local non-Markovian coupling}\label{sub:non-markovian-local-coupling}

The bounding chain from \cref{sub:bounding-chain} identifies the region
$\mc P$ where persistent discrepancies may potentially occur and, when
$\BC(X_0,Y_0,\vecsigma)=\texttt{True}$, guarantees that this region grows
outward from the initial discrepancy in a controlled tree-like manner.
We now define the basic local operation that will be used to try to prevent a propagation
in the potential disagreement set $\mc{P}_t$ from becoming a propagation in the actual disagreement set $X_t \oplus Y_t$. The resulting local
transformation, denoted $\NM_{t,H^*}$, is the basic building block of the
global coupling constructed later.\\

We begin with some preliminary definitions.  For a set $S \subset V$ and $v\in V$, let $N_S(v)$ be the neighborhood of $v$ in $G[(V \setminus S) \cup \{v\}]$, and similarly define $N_S^2(v)$.  Notice that $N_S(v) = N(v) \setminus S$ but $N_S^2(v) \ne N^2(v) \setminus S$ in general.  For simplicity, for a vertex $w$, we will write $N_w(v) = N_{\{w\}}(v)$.

For a vertex $w\in V$ and time $0<t\leq\Tcouple$, define the last and next successful updates of $w$ as follows: \begin{align*}
    \tau_w^- &= \tau_w^-(t) := \max\{s:1\leq s \leq t,v_s=w, c_s \in A(X_{s-1},v_s) \} 
\\
\tau_w^+ & = \tau_w^+(t) := \min\{s: t < s \leq \Tcouple,v_s=w, c_s \in A(X_{s-1},v_s) \},
\end{align*}
whenever the sets are nonempty, and take $\tau_w^- = 0$ or $\tau_w^+ = \Tcouple+1$ otherwise.  We define the update epoch of $w$ around time $t$ by
\[
I(w,t) = 
\begin{cases} 
[\tau_w^-,\tau_w^+)
 & \mbox{if } \tau_w^- > 0,  
 \\
 \mbox{undefined} & \mbox{otherwise.}
 \end{cases}
\]
Note that $t \in I(w,t)$ and $X_s(w)$ is constant (in particular, equal to $X_t(w)$) for all $s \in I(w,t)$. We also define the punctured interval
\[I^{\circ}(w,t) := I(w,t)\setminus \{t\}.\]

Fix initial labelings $X_0,Y_0\in\wh{\Omega}$ that differ at a single vertex $z^*\in V$, and fix an update sequence $\vecsigma \in (V \times [k])^\Tcouple$.  Let $\{Z_t\}$ denote the corresponding bounding chain, and let $\mc{P} = \cup_{t \leq \Tcouple} \mc{P}_t$ be its associated potential persistent discrepancy set. Let $X_1,\dots,X_{\Tcouple}$ denote the evolution of $X_0$ according to $\vecsigma$.

\begin{condition} 
\label{condition:nm-prelim}
Our non-Markovian update is not defined unless the following conditions hold (and even if they are all satisfied, we may still require additional properties, see \cref{assumption:non-markovian-succeed}).
\begin{itemize}

\item $1\leq t \leq \Tcouple$ is an update time such that $c_t \in H(Z_{t-1},v_t)$ (i.e., the chosen update vertex $v_t$ is attempting a hazardous update). When $\BC(X_0,Y_0,\vecsigma)=\tt{True}$, this implies that $v_t \in \mc{P}_{\leq t} \setminus \mc{P}_{\leq (t-1)}$.  

\item Let $p$ denote the first element in $N(v_t) \cap \mc{P}_{t-1}$ according to some fixed, but otherwise arbitrary, total ordering. We assume that $c_t \neq X_{t-1}(p)$. (Whenever $\BC(X_0, Y_0, \vecsigma) = \tt{True}$, $p$ is uniquely determined by \cref{lem:bc-basic-geometry-rewrite}. We will only use our non-Markovian coupling whenever this holds, see the first bullet of \cref{assumption:non-markovian-succeed}.) 

\item Let $H^* = \{ X_{t-1}(p), c_t \}$, so that $H^*$ is a subset of $H(Z_{t-1},v_t)$ of size $2$ according to the previous bullet. 

We assume that $|X_{t-1}(N_p(v_t)) \cap H^*| = 1$. In this case, let
\[c_b \in X_{t-1}(N_p(v_t)) \cap H^*\]
denote the unique ``blocked'' color and let
\[c_u \in H^*\setminus \{c_b\}\]
denote the unique unblocked color. 
\item For each $w \in X_{t-1}^{-1}(c_b) \cap N(v_t)$, we have $\tau_w^- > 0$.
\end{itemize}
\end{condition}

\begin{remark}
Notice that whether \cref{condition:nm-prelim} is satisfied or not is fully determined by $(v_s,c_s)$ for $1 \leq s \leq t$. This is in contrast to the later condition \cref{assumption:non-markovian-succeed}, which requires ``looking into the future''. 
\end{remark}

In the definitions below, for lightness of notation, we will omit the dependence on $t$ and $H^*$ in various places where there is no risk of confusion. We emphasize that we will only use these definitions on the event that \cref{condition:nm-prelim} holds and further, that $\BC(X_0, Y_0, \vecsigma) = \tt{True}$ (see the first bullet of \cref{assumption:non-markovian-succeed}).  

\medskip

For a vertex $w\in N_p(v_t)$, we denote the {\em exchangeable colors} for $w$ as follows:
\[  \Exchange(w,t) = \{X_{t-1}(w)\}\cup\{c \in A(X_{\tau_w^--1},w) : \mbox{for all }s\in I^{\circ}(w,t), \mbox{ if } v_s\in N(w) \mbox{ then } c_s\neq c\},
\]
whenever $\tau_w^-$ is defined, and $\Exchange(w,t) = \emptyset$ otherwise. 
In words, $\Exchange(w,t)$ denotes those colors $c$ which were available to $w$ at its last successful recoloring (prior to $t$) and no neighbor of $w$ attempted $c$ during $I^{\circ}(w,t)$, except that we include the current color $X_{t-1}(w)$ regardless of updates attempted in the epoch.

We now define an \emph{avoided set}
\begin{equation*}
    \Avoid(t) = \{ w \in N(v_t) : \exists s \in I^\circ(w,t) \text{ where } v_s \in (N(w) \cap N_w^2(\mc{P})) \cup \{ v_t \}, c_s = X_{s-1}(w) \} \cup N_{v_t}^2(\mc{P}).
\end{equation*}
The purpose of $\Avoid(t)$ is to ensure that non-Markovian updates do not interact with each other.  If a vertex $w$ is distance $2$ from a vertex $z \in \mc{P}$, then $w$ is ineligible.  If a vertex $w$ is distance $3$ from a vertex $z \in \mc{P}$, then $w$ is eligible only as long as editing the color of $w$ does not force one to edit any updates in $N^2(z)$.

Next, for a color $c\in [k]$, we define the set of neighbors which are {\em swappable for color $c$ at time $t$} as
\begin{equation*} \Swappable_\mc{P}(c,t) = \{w\in N_\mc{P}(v_t): c\in\Exchange(w,t) \} \setminus \Avoid(t).
\end{equation*}

We define the \emph{complementary neighbor} mapping
\[ \alpha : \Swappable_{\mc{P}}(c_b,t) \to \Swappable_{\mc{P}}(c_u,t) \]
in the following manner.  
Order $\Swappable_{\mc{P}}(c_b,t) = \{ z_1,\ldots,z_\ell \}$ by increasing $\deg z_i \cdot (|I(z_i,t)|/n + 1)$ and $\Swappable_{\mc{P}}(c_u,t) = \{ z_1',\ldots,z_{\ell'}' \}$ by increasing $\deg z_i' \cdot (|I(z_i',t)|/n + 1)$, breaking ties according to some arbitrary (but otherwise globally fixed) total order on $V$.  For $i\leq\min\{\ell,\ell'\}$, let $\alpha(z_i)=z'_i$; and for $i>\min\{\ell,\ell'\}$, let $\alpha(z_i)$ be undefined.

For any $z \in \Swappable_{\mc{P}}(c_b,t)$ for which $\alpha(z)$ is defined, we further define the \emph{complementary color} mapping
\[\beta_{z}: \Exchange(\alpha(z),t)\rightarrow\Exchange(z,t)\]
as follows.  Let $\Exchange(\alpha(z),t)=\{c'_1,\dots,c'_{m'}\}$ and $\Exchange(z,t)=\{c_1,\dots,c_{m}\}$ denote ordered sets according to the natural order on $[k]$.  For $1\leq j \leq \min\{m,m'\}$, let $\beta_z(c'_j)=c_j$; and for $j>\min\{m,m'\}$, let $\beta_z(c'_j)$ be undefined.  

Finally, we define the set of {\em blocking neighbors} as
\[ B=\{w\in N_{\mc{P}}(v_t): X_{t-1}(w)=c_b\}.
\]
If \cref{assumption:non-markovian-succeed} holds, its second bullet together with \cref{condition:nm-prelim} gives $c_b\in X_{t-1}(N_\mc{P}(v_t))$, and hence $|B|\geq 1$.

\begin{condition}[non-Markovian well-defined]
\label{assumption:non-markovian-succeed}
The transformation $\NM_{t,H^*}$ is applied only if \cref{condition:nm-prelim} holds and furthermore the following hold:
\begin{itemize}
\item $\BC(X_0, Y_0, \vecsigma) = \tt{True}$;
\item $X_{t-1}(N_p(v_t) \cap \mc{P}) \cap H^* = \emptyset$;
\item $B \cap \Avoid(t) = \emptyset$;
\item the mapping $\alpha$ is defined on $B$;
\item for each $w\in B$ with $\alpha(w)\ne w$, the color
$\beta_w(X_{t-1}(\alpha(w)))$ is defined and satisfies  $\beta_w(X_{t-1}(\alpha(w))) \notin H^*$; 
\item $\alpha(B)\cap B=\{ w \in B : \alpha(w) = w \}$.
\end{itemize}
\end{condition}

Notice that the third bullet is equivalent to $B \subset \Swappable_\mc{P}(c_b,t)$.  This follows as $c_b = X_t(w)$ for all $w \in B$ and $X_t(w) \in \Exchange(w,t)$ by definition of $\Exchange(w,t)$.

We can now define our local non-Markovian coupling.

\begin{definition}[local non-Markovian transformation]\label{def:local-NM-coupling}
Given $X_0, Y_0, \vecsigma$ and $t \in [\Tcouple]$ satisfying \cref{assumption:non-markovian-succeed}, the transformation
\[
\NM_{t,H^*}(\vecsigma)=((v_1,c'_1),\dots,(v_{\Tcouple},c'_{\Tcouple}))
\]
is defined as follows.

\begin{enumerate}
\item For each $w \in B$ with $\alpha(w) = w$, set
\[ c'_{\tau_w^-} = c_u. \]
\item For each $w\in B$ with $\alpha(w) \ne w$, set
\[
c'_{\tau_w^-}=\beta_w\big(X_{t-1}(\alpha(w))\big),
\qquad
c'_{\tau_{\alpha(w)}^-}=c_u .
\]

\item For every $w\in B\cup \alpha(B)$ and every $s\in I^\circ(w,t)$ with
$v_{s}\in N(w)$ and $c_{s}=X_{s-1}(w)$, set
\[
c'_{s}=c'_{\tau_w^-}.
\]

\item At time $t$ set
\[
c'_t = H^*\setminus\{c_t\} = X_{t-1}(p).
\]

\item For all other times $s$, set $c'_s=c_s$.
\end{enumerate}

\end{definition}

\begin{remark}
\label{rmk:temporary-vs-persistent}
    The point of this coupling is that we introduce ``temporary'' discrepancies to block ``persistent'' discrepancies.  However, \cref{assumption:non-markovian-succeed} does not assume that $\tau_w^+$ is defined for $w \in B$; thus, some of these temporary discrepancies may persist until time $\Tcouple$.  This is intentional and is designed this way to overcome measurability obstacles.  While these errors may not be fixed, they cannot propagate as we have checked that their new colors were exchangeable in $I(w,t)$.
\end{remark}

The following simple lemma shows that the edits in the previous definition are pairwise disjoint, whenever \cref{assumption:non-markovian-succeed} holds. 

\begin{fact}
\label{lem:temporary-geometry}
Suppose \cref{assumption:non-markovian-succeed} holds. Let $\Temp_t:=B\cup \alpha(B)$. Then:
\begin{enumerate}
\item $\Temp_t\subseteq N_{\mc P}(v_t)$, and for every $w\in \Temp_t$ one has $N(w)\cap \mc P=\{v_t\}$;
\item If $u\in N(\Temp_t)\setminus\{v_t\}$, then $u \notin N(\mc P)$ and there is a unique $w\in \Temp_t$ with $u\in N(w)$;
\item If $u$ has some update $(u,c)$ changed by $\NM_{t,H}$, then $u \notin N^2_{v_t}(\mc{P})$;
\item $\Temp_t$ is an independent set.
\end{enumerate}
\end{fact}

\begin{proof}
The inclusion $\Temp_t\subseteq N_{\mc P}(v_t)$ is immediate from the definitions of $B$ and $\alpha$. Fix $w\in \Temp_t$. Since $v_t\in \mc P$, we already have $v_t\in N(w)\cap \mc P$. If there were another vertex $q\in N(w)\cap \mc P$ with $q\neq v_t$, then the connectedness of $G[\mc P]$ from \cref{lem:bc-basic-geometry-rewrite}(1) would give a path in $G[\mc P]$ from $q$ to $v_t$; together with the edges $qw$ and $wv_t$, this would create a cycle in $G[\mc P \cup \{w\}]$, contradicting $\BC(X_0,Y_0,\vecsigma)=\tt{True}$. This proves (1).

The first part of (2) follows immediately from the fact that $N_{v_t}^2(\mc{P}) \subset \Avoid(t)$.  If some $u\in N(\Temp_t)\setminus\{v_t\}$ were adjacent to two distinct vertices $w,w'\in \Temp_t$, then $u-w-v_t-w'-u$ would be a $4$-cycle in $G$, but $G$ has girth $\geq7$. This proves (2).

Bullet (3) of \cref{lem:temporary-geometry} is the purpose of the final part of $\Avoid(t)$, and follows immediately by checking the definition of $\Avoid(t)$ and \cref{def:local-NM-coupling}.  A vertex may not be put in $\Temp_t$ if exchanging its color forces it to change an update from a vertex in the $2$-neighborhood of $\mc{P}$.

Finally, if two distinct vertices $w,w'\in \Temp_t$ were adjacent, then $w-v_t-w'-w$ would be a triangle, contradicting our girth assumption. Hence $\Temp_t$ is independent.
\end{proof}

\subsection{Global Non-Markovian Coupling by Composing Local Couplings}
\label{sub:global}

Having defined our local non-Markovian coupling, we can now define our global (non-Markovian) coupling. As mentioned earlier, this will arise by composing three different types of maps: the identity map, the local non-Markovian map, and the following local Jerrum map.

\begin{definition}[local Jerrum coupling]
    Given $H \in \binom{[k]}2$, $t \in [\Tcouple]$, and \[\vecsigma=\vecsigma_{\Tcouple}=((v_1,c_1),\dots,(v_{\Tcouple},c_{\Tcouple})),\] define $\Jerrum_{t,H}(\vecsigma)$ to be the update sequence $\vecsigma' = ((v_s,c_s'))_{s \in [\Tcouple]}$ where $c_s' = c_s$ for all $s \ne t$ and
\[ c_t' = \begin{cases} H \setminus \{c_t\} & c_t \in H, \\ c_t & c_t \notin H. \end{cases} \]
\end{definition}

\begin{definition}[global coupling]\label{def:global-coupling}
Our global coupling is a partial function  $F : \wh{\Omega} \times \wh{\Omega} \times (V \times [k])^\Tcouple  \to (V \times [k])^\Tcouple$, 
\begin{align*}
    (X_0,Y_0,\vecsigma) & \xrightarrow[]{F} \vecsigma'
\end{align*}
which maps an update sequence $\vecsigma$, which will be applied to $X_0$, to a corresponding update sequence $\vecsigma'$, which will be applied to $Y_0$. The function $F$ is defined only if $X_0$ and $Y_0$ disagree at exactly one vertex, $z^*$. 

If $\BC(X_0,Y_0,\vecsigma) = \tt{False}$, then set $F(X_0,Y_0,\vecsigma) = \vecsigma$ (thus, we are using the identity coupling when the bounding chain ``fails'').  Otherwise, proceed as follows.

For $t \in \{0,1,\dots, \Tcouple\}$, we will inductively define functions
\[f_t: (V \times [k])^{\Tcouple} \to (V \times [k])^{\Tcouple},\]
with the initialization $f_0 = \on{id}$, and sets 
\[D_t \subseteq V\]
with the initialization $D_0 = \{z^*\}$, the unique ``root disagreement''. Let 
\[f_{\leq t}(\vec{\sigma}) = f_{t}\circ f_{t-1}\cdots \circ f_0(\vec{\sigma}).\] We  define
\[F(X_0, Y_0, \vec{\sigma}) = f_{\leq \Tcouple}(\vec{\sigma}).\]

For $t\geq 1$, suppose we have already defined $D_{t-1}$ and $f_{\leq t-1}(\vecsigma)$. 
Let $Y^{t-1}_s$ denote the configuration of the chain started from $Y_0$ after $s$ steps according to the intermediate
update sequence $f_{\le t-1}(\vecsigma) \in (V \times [k])^{\Tcouple}$.
In particular, $Y^{t-1}_{t-1}$ is the state after
$t-1$ steps under this intermediate update sequence. We now define $D_{t}$ and $f_{t}$ as follows.

\begin{enumerate}
    \item If $v_t \notin B_1(D_{t-1})$ (i.e.~$v_t \notin D_{t-1} \cup N(D_{t-1})$), then choose $f_t = \id$.  Set $D_t = D_{t-1}$.
    \item If $v_t \in D_{t-1}$, then choose $f_t = \id$.  Set 
    \[ D_{t} = \begin{cases} D_{t-1} \setminus \{v_t\} & c_t \in A(X_{t-1},v_t),  \\ D_{t-1} & c_t \notin A(X_{t-1},v_t). \end{cases} \]

    \item Otherwise, $v_t \in  N(D_{t-1})$.  Let $p$ be the first element in $N(v_t) \cap D_{t-1}$ (according to some fixed, but otherwise arbitrary, total order) and let $H = \{ X_{t-1}(p), Y_{t-1}^{t-1}(p) \}$. 

    \begin{enumerate}
        \item If $|H| = 1$, then $f_t = \on{id}$ and $D_{t} = D_{t-1}$.
    \item If $|H| = 2$ and $c_t \ne Y_{t-1}^{t-1}(p)$, set $f_t = \Jerrum_{t,H}$ and $D_t = D_{t-1}$. Otherwise, $|H| = 2$, $c_t = Y_{t-1}^{t-1}(p)$ and we proceed to the next case. 
    \item This is the ``danger zone''. Note that $c_t = Y_{t-1}^{t-1}(p) \ne X_{t-1}(p)$.  Set 
    \[ f_t = \begin{cases} \Jerrum_{t,H} & H \subset X_{t-1}(N_p(v_t)), \\ \NM_{t,H} & H \cap X_{t-1}(N_p(v_t)) = \{ c_t \} \text{ and \cref{assumption:non-markovian-succeed}  holds}, \\ \NM_{t,H} & H \cap X_{t-1}(N_p(v_t)) = H \setminus \{c_t\} \text{ and \cref{assumption:non-markovian-succeed} holds}, \\ \Jerrum_{t,H} & \text{otherwise}. \end{cases} \]
    We further track the discrepancy set
    \[ D_{t} = \begin{cases} D_{t-1} & H \subset X_{t-1}(N_p(v_t)), \\ D_{t-1} & H \cap X_{t-1}(N_p(v_t)) = \{ c_t \} \text{ and \cref{assumption:non-markovian-succeed} holds}, \\ D_{t-1} \cup \{v_t\} & H \cap X_{t-1}(N_p(v_t)) = H \setminus \{c_t\} \text{ and \cref{assumption:non-markovian-succeed} holds}, \\ D_{t-1} \cup \{ v_t \} & \text{otherwise}. \end{cases} \]
    \end{enumerate}
\end{enumerate}

\end{definition}

\begin{remark}
We will later see (\cref{prop:main-global-induction-rewrite}) that $N(v_t) \cap D_{t-1}$ is always a singleton and hence $p$ is uniquely determined in Case~(3). In the same proposition, we will also see that $|H| = 2$, so that one never enters Case~(3a). Nevertheless, we work with the more general definition above to ensure that the function is \emph{a priori} well-defined. 
\end{remark}

\begin{remark}
    Jerrum's coupling~\cite{jerrum1995very} may be defined exactly as above with two differences, both on step 3.  Replace $f_t = \Jerrum_{t,H}$ for all four cases.  For $D_t$, in the second case ($H \cap X_{t-1}(N_p(v_t)) = \{c_t\}$ and \cref{assumption:non-markovian-succeed} holds), still set $D_t = D_{t-1} \cup \{v_t\}$ instead of $D_t = D_{t-1}$.  This single case is our margin of victory.
\end{remark}

\section{Validity of the coupling}
\label{sec:validity}

In this section we prove that the coupling $F$, which was defined in \cref{sec:nonmarkovian}, is a valid coupling by proving that it is a bijective map on the space of update sequences.   We first show some basic properties of the local non-Markovian coupling in \cref{sub:localNM-properties} and of the bounding chain in \cref{sub:validity-bounding}.  We then prove in \cref{sub:validity-local} that the local non-Markovian coupling is bijective.
Finally, we prove that the global coupling $F$ is a valid coupling in \cref{sub:validity-global-bijective}.

\subsection{Properties of local coupling}
\label{sub:localNM-properties}

We will need the following properties of $\NM_{t,H^*}$ in later proofs.

\begin{lemma}\label{lem:non-markovian-local-is-bijection}
    For all $s \in [\Tcouple]$, let $X_s$ denote $X_0$ evolved according to $\vecsigma$ and $X'_s$ denote $X_0$ evolved according to $\vecsigma'=\NM_{t,H^*}(\vecsigma)$.
    \begin{enumerate}
        \item \label{item:local-far-away-no-change}
        For all $s\in[\Tcouple]$ with $v_s\notin B_1(B \cup \alpha(B))$, $c_s=c_s'$.
        \item For all $s\in [\Tcouple]$ the following hold:
       \begin{enumerate}
           \item \label{item:local-coupling-blocking-condition} If $v_s \in N_\mc{P}(B\cup \alpha(B))$ then $c_s \in A(X_{s-1},v_s)$ if and only if $c_s' \in A(X_{s-1}',v_s)$.  Furthermore, if both chains succeed, then $c_s = c_s'$.
           \item $X_s(v) = X_s'(v)$ unless $v \in B \cup \alpha(B) \cup \mc{C}_s$, where $\mc{C}_s$ is the (possibly empty) component of $\mc{P}_{\leq s} \setminus \mc{P}_{t-1}$ containing $v_t$.
           \label{item:local-coupling-only-vt-propagates}
       \end{enumerate}
    \end{enumerate}
\end{lemma}

These properties will be utilized in the later proof that the global coupling (formed by composing local couplings) is bijective, and hence a valid coupling.  Perhaps the most subtle item in \cref{lem:non-markovian-local-is-bijection} is \cref{item:local-coupling-only-vt-propagates}.  Essentially, this states that at time $t$, we will have the discrepancies created on $B \cup \alpha(B)$ as well as the discrepancy at $v_t$.  As time progresses past $t$, those on $B \cup \alpha(B)$ will not propagate outwards, and although $v_t$ will propagate, it will do so in a contained manner, following the bounding chain $\mc{P}_t$.  In particular, for $s<t$, we have $\mc{C}_s = \emptyset$.

\begin{proof}
Set
\[
W:=B\cup \alpha(B),
\qquad
\chi(w):=X_{t-1}(w),
\qquad
\chi'(w):=c'_{\tau_w^-}
\qquad (w\in W).
\]
We first record two basic consequences of the definition of $\NM_{t,H^*}$.

For every $w\in W$ one has
\begin{equation}
\label{eq:local-new-color-exchangeable-rewrite}
\chi'(w)\in \Exchange(w,t).
\end{equation}
Indeed, if $w\in B$ and $\alpha(w) \ne w$, then
$\chi'(w)=\beta_w(X_{t-1}(\alpha(w)))\in \Exchange(w,t)$ because $\beta_w$ maps into $\Exchange(w,t)$.
If $w\in \alpha(B)$, then $\chi'(w)=c_u$ and $w\in \Swappable_\mc P(c_u,t)$, so again $c_u\in \Exchange(w,t)$.
If $w \in B$ and $\alpha(w) = w$, then $w=\alpha(w)\in\Swappable_\mc P(c_u,t)$, and hence $\chi'(w)=c_u\in\Exchange(w,t)$.

We also note that $\chi'(w)\neq \chi(w)$ for every $w\in W$.
Indeed, if $w\in B$ and $\alpha(w)\ne w$, then $\chi(w)=c_b\in H^*$, whereas $\chi'(w)=\beta_w(X_{t-1}(\alpha(w)))\notin H^*$ by \cref{assumption:non-markovian-succeed}.
If $w\in \alpha(B)$, then $\chi'(w)=c_u$, while $\chi(w)=X_{t-1}(w)\neq c_u$ because $c_b$ is the unique color from $H^*$ appearing on $N_\mc P(v_t)$.

For $w\in B$, the condition $B\cap \mathsf{Avoid}_{\mc P}(t)=\emptyset$ gives $w\in \Swappable_\mc P(c_b,t)$, while for $w\in \alpha(B)$ we have $w\in \Swappable_\mc P(c_u,t)$ by definition of $\alpha$.
Hence for every $w\in W$ there is no $s\in I^\circ(w,t)$ with $v_s=v_t$ and $c_s=\chi(w)$, i.e.
\begin{equation}
\label{eq:local-vt-does-not-attempt-old-rewrite}
\text{for every }w\in W\text{ there is no }s\in I^\circ(w,t)\text{ with }v_s=v_t\text{ and }c_s=\chi(w).
\end{equation}
Since $\chi'(w)\neq \chi(w)$, the inclusion \eqref{eq:local-new-color-exchangeable-rewrite} places $\chi'(w)$ in the second part of the definition of $\Exchange(w,t)$, and therefore
\begin{equation}
\label{eq:local-no-neighbor-uses-new-rewrite-section24}
\chi'(w)\in A(X_{\tau_w^- -1},w)
\qquad\text{and}\qquad
\text{if }s\in I^\circ(w,t)\text{ and }v_s\in N(w),\text{ then }c_s\neq \chi'(w).
\end{equation}

The only coordinates edited by $\NM_{t,H^*}$ are:
(i) time $t$;
(ii) the times $\tau_w^-$ for $w\in W$; and
(iii) the times $s\in I^\circ(w,t)$ with $w\in W$ and $v_s\in N(w)$ at which $c_s=\chi(w)$.

\cref{item:local-far-away-no-change} follows immediately from \cref{def:local-NM-coupling}.



We now prove the remaining assertions by induction on time.
For $0\le r\le \Tcouple$, let $(\dagger_r)$  denote the conjunction of the following statements:
\begin{enumerate}
\item for every $w\in W$,
\[
X_r(w)=X_r'(w)\quad\text{if }r<\tau_w^-\text{ or }r\ge \tau_w^+,
\]
and
\[
X_r(w)=\chi(w),
\qquad
X_r'(w)=\chi'(w)
\quad\text{if }\tau_w^-\le r<\tau_w^+;
\]
\item if $u\notin W\cup \mc C_r$, then $X_r(u)=X_r'(u)$;
\item if $u\in \mc C_r$, then $\{X_r(u),X_r'(u)\}\subseteq Z_r(u)$.
\end{enumerate}
The base case $r=0$ is immediate because $W\cap \mc P=\emptyset$ and $\mc C_0=\emptyset$.
Assume $(\dagger_{s-1})$ and let $u:=v_s$.

\smallskip
\noindent\emph{Case 1: $u\in N(W)\setminus\{v_t\}$.}
Recall that by definition, $W \cap \Avoid(t) = \emptyset$ and $N_{v_t}^2(\mc{P}) \subset \Avoid(t)$.  By \cref{lem:temporary-geometry}, there is a unique $w\in W$ with $u\in N(w)$, and $u$ has no neighbor in $\mc P$.
Every neighbor of $u$ other than $w$ lies outside $W\cup \mc C_{s-1}$, and hence has the same color in the two runs by $(\dagger_{s-1})(2)$.

If $s\in I^\circ(w,t)$ and $c_s=\chi(w)$, then by definition $c_s'=\chi'(w)$.
By $(\dagger_{s-1})(1)$, the vertex $w$ currently carries $\chi(w)$ in the original run and $\chi'(w)$ in the modified run, so both proposals are blocked by the same neighbor $w$.

In every other subcase, $c_s'=c_s$.
If $s\notin I^\circ(w,t)$, then $(\dagger_{s-1})(1)$ gives $X_{s-1}(w)=X_{s-1}'(w)$.
If $s\in I^\circ(w,t)$, then either $c_s\neq \chi(w)$ by assumption or else we are in the previous subcase, and \eqref{eq:local-no-neighbor-uses-new-rewrite-section24} gives $c_s\neq \chi'(w)$.
Hence in every case $w$ has the same effect on the proposal in the two runs.
Since all other neighbors of $u$ also agree, the proposal has the same status in both chains.
This proves \cref{item:local-coupling-blocking-condition} at time $s$.
Moreover, when both chains succeed we are necessarily in the subcase $c_s'=c_s$, proving the final sentence in \cref{item:local-coupling-blocking-condition}.
Since $u$ has no neighbor in $\mc P$, the common status cannot be hazardous, and therefore $X_s(u)=X_s'(u)$.

\smallskip
\noindent\emph{Case 2: $u=w\in W$.}
We split according to the position of $s$ relative to the epoch $I(w,t)$.

If $s<\tau_w^-$, then $c_s'=c_s$.
Also $s<t$, and $(\dagger_{s-1})(2)$ gives equality on all neighbors of $w$.
So the proposal has the same status in the two runs, and $w$ remains equal.

If $s=\tau_w^-$, then the original proposal is $\chi(w)$, which is successful by definition of $\tau_w^-$.
The modified proposal is $\chi'(w)$, and this is also successful because \eqref{eq:local-no-neighbor-uses-new-rewrite-section24} gives $\chi'(w)\in A(X_{\tau_w^- -1},w)$, while before time $\tau_w^-$ the two runs coincide on $N(w)$.
Hence after time $\tau_w^-$ the two runs carry the colors $\chi(w)$ and $\chi'(w)$ at $w$.

Now assume $\tau_w^-<s<\tau_w^+$.
The proposal at $w$ is unchanged, so $c_s'=c_s$.
In the original run this update is not successful, by definition of $\tau_w^+$.
Suppose for contradiction that it were successful in the modified run.
Since every neighbor of $w$ other than $v_t$ has the same color in the two runs, the only possible cause of different status is the color at $v_t$.
If $s<t$, then the two runs still agree at $v_t$ by $(\dagger_{s-1})(2)$, contradiction.
Hence $s>t$, so $v_t\in \mc C_{s-1}$, and $(\dagger_{s-1})(3)$ gives
\[
\{X_{s-1}(v_t),X_{s-1}'(v_t)\}\subseteq Z_{s-1}(v_t).
\]
Because the proposal has different status in the two runs, it must be blocked by $v_t$ in exactly one run; in particular
\[
c_s\in \{X_{s-1}(v_t),X_{s-1}'(v_t)\}
\quad\text{and}\quad
X_{s-1}(v_t)\neq X_{s-1}'(v_t).
\]
Hence $|Z_{s-1}(v_t)|>1$, so $v_t\in \mc P_{s-1}$, and therefore
\[
c_s\in Z_{s-1}(v_t)\subseteq H(Z_{s-1},w).
\]
But then the update at time $s$ is hazardous for the original bounding chain, which would imply $w\in \mc P_s$, contradicting $w\in W\subseteq N_\mc P(v_t)$.
Hence the update is blocked in both runs, and $w$ continues to carry $\chi(w)$ and $\chi'(w)$.

If $s=\tau_w^+$, then again $c_s'=c_s$, and the original update is successful by definition of $\tau_w^+$.
If it failed in the modified run, then every neighbor of $w$ other than $v_t$ would still agree in the two runs, so the only possible cause would again be the color at $v_t$.
Since $s=\tau_w^+>t$, we have $v_t\in \mc C_{s-1}$, and $(\dagger_{s-1})(3)$ gives
\[
\{X_{s-1}(v_t),X_{s-1}'(v_t)\}\subseteq Z_{s-1}(v_t).
\]
Because the proposal has different status in the two runs, it must be blocked by $v_t$ in exactly one run; hence
\[
c_s\in \{X_{s-1}(v_t),X_{s-1}'(v_t)\}
\quad\text{and}\quad
X_{s-1}(v_t)\neq X_{s-1}'(v_t).
\]
Therefore $|Z_{s-1}(v_t)|>1$, so $v_t\in \mc P_{s-1}$ and
\[
c_s\in Z_{s-1}(v_t)\subseteq H(Z_{s-1},w),
\]
again contradicting $w\notin \mc P$.
So the update succeeds in both runs and sends $w$ to the same color $c_s$; from this time onward the two runs agree at $w$.

Finally, if $s>\tau_w^+$, then $w$ already agrees in the two runs by the previous subcase.
If the update at time $s$ created a new discrepancy at $w$, then every neighbor of $w$ other than $v_t$ would still agree in the two runs, so the discrepancy could only come from the color at $v_t$.
Since $s>t$, we have $v_t\in \mc C_{s-1}$, and $(\dagger_{s-1})(3)$ gives
\[
\{X_{s-1}(v_t),X_{s-1}'(v_t)\}\subseteq Z_{s-1}(v_t).
\]
Because the proposal has different status in the two runs, it must be blocked by $v_t$ in exactly one run; hence
\[
c_s\in \{X_{s-1}(v_t),X_{s-1}'(v_t)\}
\quad\text{and}\quad
X_{s-1}(v_t)\neq X_{s-1}'(v_t).
\]
Therefore $|Z_{s-1}(v_t)|>1$, so $v_t\in \mc P_{s-1}$ and
\[
c_s\in Z_{s-1}(v_t)\subseteq H(Z_{s-1},w),
\]
hence $w\in \mc P_s$, impossible.
So once equality is restored at $\tau_w^+$, it persists forever.
This proves $(\dagger_s)(1)$.

\smallskip
\noindent\emph{Case 3: $u=v_t$.}
If $s<t$, then $c_s'=c_s$: by \eqref{eq:local-vt-does-not-attempt-old-rewrite} and \eqref{eq:local-no-neighbor-uses-new-rewrite-section24}, the local edit never changes a proposal of $v_t$ before time $t$.
For every $w\in W$ that currently differs in the two runs, the common proposal $c_s$ is neither $\chi(w)$ nor $\chi'(w)$, so the temporary discrepancies on $W$ do not affect the proposal status at $v_t$.
Since $\mc C_{s-1}=\emptyset$ for $s<t$, all remaining neighbors of $v_t$ agree in the two runs, and therefore $X_s(v_t)=X_s'(v_t)$.

If $s=t$, then $c_t'=H^*\setminus\{c_t\}$.
Both colors $c_t$ and $c_t'$ belong to $H^*\subseteq H(Z_{t-1},v_t)$, so
\[
\{X_t(v_t),X_t'(v_t)\}\subseteq H(Z_{t-1},v_t)\cup Z_{t-1}(v_t)=Z_t(v_t).
\]
Also $v_t\in \mc C_t$ by definition.

Now assume $s>t$.
Then again $c_s'=c_s$.
As before, \eqref{eq:local-vt-does-not-attempt-old-rewrite} and \eqref{eq:local-no-neighbor-uses-new-rewrite-section24} show that the temporary discrepancies on $W$ are irrelevant to the proposal status at $v_t$.
Thus any discrepancy in the update outcome can only come from neighbors in $\mc C_{s-1}$.
For each $q\in \mc C_{s-1}$, $(\dagger_{s-1})(3)$ gives
\[
\{X_{s-1}(q),X_{s-1}'(q)\}\subseteq Z_{s-1}(q).
\]
Hence the only new color that can appear at $v_t$ is either the old color already contained in $Z_{s-1}(v_t)$, or the proposal color $c_s$.
If the two runs treat $c_s$ differently, then there exists a neighbor $q\in N(v_t)\cap \mc C_{s-1}$ such that $c_s$ equals exactly one of $X_{s-1}(q)$ and $X_{s-1}'(q)$.
In particular $X_{s-1}(q)\neq X_{s-1}'(q)$, so $|Z_{s-1}(q)|>1$ and hence $q\in \mc P_{s-1}$.
Therefore
\[
c_s\in Z_{s-1}(q)\subseteq H(Z_{s-1},v_t).
\]
In all cases
\[
\{X_s(v_t),X_s'(v_t)\}\subseteq Z_s(v_t),
\]
and $v_t\in \mc C_s$.
Thus $(\dagger_s)(2)$ and $(\dagger_s)(3)$ hold at $u=v_t$.

\smallskip
\noindent\emph{Case 4: $u\notin W\cup N(W)$.}
Then $c_s'=c_s$.
If $u\notin \mc C_{s-1}\cup N(\mc C_{s-1})$, every neighbor of $u$ has the same color in the two runs by $(\dagger_{s-1})(2)$, so $X_s(u)=X_s'(u)$.

Otherwise $u\in \mc C_{s-1}\cup N(\mc C_{s-1})$.
If $u\in \mc C_{s-1}$, then automatically $u\in \mc C_s$, because
\[
\mc P_{\le s}\setminus \mc P_{t-1}\supseteq \mc P_{\le s-1}\setminus \mc P_{t-1},
\]
so the component containing $v_t$ can only grow.
Since $u\notin N(W)$, the temporary discrepancies on $W$ are irrelevant.
Every neighbor of $u$ outside $\mc C_{s-1}$ agrees in the two runs, while for each $q\in \mc C_{s-1}$ the inductive hypothesis gives
\[
\{X_{s-1}(q),X_{s-1}'(q)\}\subseteq Z_{s-1}(q).
\]
Therefore the only new color that can appear at $u$ is again the common proposal $c_s$.
If the two runs treat that proposal differently, then there exists a neighbor $q\in N(u)\cap \mc C_{s-1}$ such that $c_s$ equals exactly one of $X_{s-1}(q)$ and $X_{s-1}'(q)$.
In particular $X_{s-1}(q)\neq X_{s-1}'(q)$, so $|Z_{s-1}(q)|>1$ and hence $q\in \mc P_{s-1}$.
Therefore
\[
c_s\in Z_{s-1}(q)\subseteq H(Z_{s-1},u).
\]
Consequently
\[
\{X_s(u),X_s'(u)\}\subseteq Z_s(u).
\]
If $u\notin \mc C_{s-1}$ and $X_s(u)\neq X_s'(u)$, then the differing status implies $c_s\in H(Z_{s-1},u)$, so $u\in \mc P_s$; since $u$ has a neighbor in $\mc C_{s-1}$, it follows that $u\in \mc C_s$.
This proves $(\dagger_s)(2)$ and $(\dagger_s)(3)$.

These four cases exhaust all possibilities, so $(\dagger_s)$ holds for every $s\le \Tcouple$.
Assertion \cref{item:local-coupling-only-vt-propagates} is exactly $(\dagger_s)(2)$.
\end{proof}

The proof above also gives the following conclusion, which we isolate since it will be needed in the involution argument later.

\begin{corollary}
\label{cor:epoch-behavior-W-rewrite}
In the setup of \cref{lem:non-markovian-local-is-bijection}, let $W=B\cup \alpha(B)$. 
\[
X_s(w)=X_s'(w)\quad\text{if }s<\tau_w^-\text{ or }s\ge \tau_w^+,
\]
and
\[
X_s(w)=X_{t-1}(w),
\qquad
X_s'(w)=c'_{\tau_w^-}
\quad\text{if } s \in I(w,t).
\]
In particular, the next successful recoloring time of $w$ after $t$ is still $\tau_w^+$ in the modified run, so the epoch $I(w,t)$ is unchanged by the local edit.
\end{corollary}

\begin{proof}
This is exactly the epoch-wise description established in the proof of \cref{lem:non-markovian-local-is-bijection}; the final sentence follows because the local edit neither creates nor destroys the successful recoloring time $\tau_w^+$.
\end{proof}

\subsection{Bounding chain properties}
\label{sub:validity-bounding}

Here we show that the bounding chain satisfies certain symmetry properties, which will be used in the later proofs showing that the coupling is valid.

\begin{lemma}
\label{lem:bounding-chain-is-symmetric}
    $\BC(\cdot,\cdot,\cdot)$ satisfies the following symmetry properties.
    \begin{itemize}
        \item $\BC(X_0, Y_0, \vecsigma) = \BC(Y_0, X_0, \vecsigma)$.
        \item If $\BC(X_0,Y_0,\vecsigma) = \tt{True}$, let $\vecsigma' = F(X_0,Y_0,\vecsigma)$.  Then $\BC(X_0,Y_0,\vecsigma') = \tt{True}$.  Furthermore, for all $t \in [\Tcouple]$, the sets $\mc{P}_t$ and $\wt{\mc{P}}_t$ created by $(X_0,Y_0,\vecsigma)$ and by $(X_0,Y_0,\vecsigma')$ are the same and further, the sets $Z_t(v)$ and $\wt{Z}_t(v)$ are also the same for all $v \in \mc{P}_t = \wt{\mc{P}}_t$.
    \end{itemize}
\end{lemma}

\begin{proof}
    The first bullet is immediate, as examining the definition reveals no difference in how $\BC$ handles its first two inputs.

    The second bullet is somewhat more subtle.  We will do this by analyzing the local functions $f_t$, and showing that none of them will meaningfully affect the evolution of the bounding chain $Z_t$.  There are three possibilities for $f_t$. 
    \begin{itemize}
        \item $f_t = \id$.  This is trivial.
        \item $f_t = \Jerrum_{t,H}$ for some $H \subset H(Z_{t-1},v_t)$.  This is also immediate from reviewing the definition of the bounding chain: a hazardous update never accesses which element of the hazardous set was actually attempted.
        \item $f_t = \NM_{t,H}$ is non-Markovian.  In this case, we must run an induction similar to proving \cref{lem:non-markovian-local-is-bijection}.  The key point is that in terms of updates, $\NM_{t,H}$ affects $\vecsigma$ in the $2$-neighborhood at potentially any time in $[\Tcouple]$ but affects $\vecsigma$ in the $1$-neighborhood at only times in $[0,t)$ (specifically, $\tau_w^-$ for each $w \in N_p(v_t)$).  Thus the changes are all in the past, before any times at which $v_t$ is hazardous (since $\BC(X_0,Y_0,\vecsigma) = \tt{True}$ means that $v_t$ was not yet hazardous).  Thus the hazardous set $\mc{P}_t$ constructed in $\{Z_s\}$ does not see them.
    \end{itemize}
Formally, let
\[
\eta:=\vecsigma=((v_s,c_s))_{s\in [\Tcouple]},
\qquad
\eta':=\vecsigma'=\NM_{t,H}(\vecsigma)=((v_s,c'_s))_{s\in [\Tcouple]}.
\]
Let $Z^\eta$ and $Z^{\eta'}$ be the corresponding bounding chains, and write $\widetilde Z:=Z^{\eta'}$ and $\widetilde{\mc P}:=\mc P^{\eta'}$. Let $W = B \cup \alpha(B)$.

For each $w\in W$, write
\[
\chi(w):=X_{t-1}(w),
\qquad
\chi'(w):=c'_{\tau_w^-}.
\]
As in the proof of \cref{lem:non-markovian-local-is-bijection}, we have
\begin{equation}
\label{eq:nm-exchangeable-color-available-rewrite}
\chi'(w)\in \Exchange(w,t)
\qquad\text{for every }w\in W,
\end{equation}
and therefore
\begin{equation}
\label{eq:nm-no-neighbor-uses-new-color-rewrite}
\text{for every }w\in W\text{ and every }s\in I^\circ(w,t)\text{ with }v_s\in N(w),\text{ one has }c_s\ne \chi'(w).
\end{equation}
Because every $w\in W$ belongs to a swappable set, we also have
\begin{equation}
\label{eq:nm-vt-does-not-attempt-old-color-rewrite}
\text{for every }w\in W\text{ there is no }s\in I^\circ(w,t)\text{ with }v_s=v_t\text{ and }c_s=\chi(w).
\end{equation}

Moreover, \cref{lem:temporary-geometry} gives the following geometry relative to the bounding chain generated by $\eta$: $W\subseteq N_{\mc P^{\eta}}(v_t)$, for every $w\in W$ one has $N(w)\cap \mc P^{\eta}=\{v_t\}$, and if $u\in N(W)\setminus\{v_t\}$ then $N(u)\cap \mc P^{\eta}=\emptyset$ and there is a unique $w\in W$ with $u\in N(w)$.

We now prove by induction on $s$ the following stronger statement, which we denote by $(\dagger_s)$:
\medskip

\noindent For every $0\le r\le s$:
\begin{enumerate}
\item $\widetilde{\mc P}_r=\mc P_r^{\eta}$ and $\widetilde Z_r(u)=Z_r^{\eta}(u)$ for every $u\in \mc P_r^{\eta}$;
\item for every $w\in W$,
\[
Z_r^{\eta}(w)=\{\chi(w)\}
\quad\text{and}\quad
\widetilde Z_r(w)=\{\chi'(w)\}
\qquad\text{whenever }\tau_w^-\le r<\tau_w^+,
\]
while $Z_r^{\eta}(w)=\widetilde Z_r(w)$ whenever $r<\tau_w^-$ or $r\ge \tau_w^+$;
\item if $u\in N(W)\setminus\{v_t\}$, then
\[
\widetilde Z_r(u)=Z_r^{\eta}(u)
\qquad\text{and}\qquad
u\notin \mc P_r^{\eta}.
\]
\end{enumerate}
\medskip
The base case $s=0$ is immediate. Assume $(\dagger_{s-1})$ and consider time $s$. Only the vertex $v_s$ can change at time $s$, so we distinguish cases.

\smallskip
\noindent\emph{Case 1: $v_s\notin W\cup N(W)$.}
Then $c_s'=c_s$, and every neighbor of $v_s$ has the same bounding-chain state in the two runs by $(\dagger_{s-1})$. Hence the proposal has the same status in the two bounding chains, and the update at time $s$ has the same effect in both. Therefore $(\dagger_s)$ holds.

\smallskip
\noindent\emph{Case 2: $v_s=w\in W$.}
We split according to the position of $s$ inside the epoch $I(w,t)$.

If $s=\tau_w^-$, then the original proposal is $\chi(w)$ and the modified proposal is $\chi'(w)$. Since $\tau_w^-<t$ and $v_t\notin \mc P_{\le t-1}^{\eta}$, the vertex $w$ has no multivalued neighbor at time $\tau_w^- -1$; hence every neighbor contributes a singleton color in both bounding chains. The proposal $\chi(w)$ is available because $\tau_w^-$ is a successful recoloring time in the reference run, and $\chi'(w)$ is also available by \eqref{eq:nm-exchangeable-color-available-rewrite}. Thus both updates set the bounding-chain state of $w$ to a singleton, namely $\{\chi(w)\}$ and $\{\chi'(w)\}$ respectively.

If $\tau_w^-<s<\tau_w^+$, then $c_s'=c_s$. By \cref{lem:temporary-geometry}, every neighbor of $w$ other than $v_t$ lies in $N(W)\setminus\{v_t\}$ and therefore has the same bounding-chain state in the two runs by $(\dagger_{s-1})(3)$, while the unique neighbor of $w$ in $\mc P^{\eta}$ is $v_t$, whose multivalued set is also the same in the two processes by $(\dagger_{s-1})(1)$. Therefore the proposal has the same status in the two bounding chains. Since $w\notin \mc P^{\eta}$ and there is no successful recoloring time of $w$ in the open interval $(\tau_w^-,\tau_w^+)$, the update at time $s$ is blocked in the reference bounding chain and hence also in the modified one. So the singleton status of $w$ is unchanged in both runs.

If $s=\tau_w^+$, then again $c_s'=c_s$ and the same neighborhood comparison shows that the proposal has the same status in both bounding chains. Because $w\notin \mc P^{\eta}$, the update at time $\tau_w^+$ cannot be hazardous; and since $\tau_w^+$ is a successful recoloring time in the reference run, it must therefore be available. Hence it is also available in the modified bounding chain, and both runs set the state of $w$ equal to the same singleton $\{c_s\}$. Thus $(\dagger_s)$ holds in Case~2.

\smallskip
\noindent\emph{Case 3: $v_s=u\in N(W)\setminus\{v_t\}$.}
By \cref{lem:temporary-geometry}, there is a unique $w\in W$ with $u\in N(w)$, and $u$ has no neighbor in $\mc P^{\eta}$ by definition of $\Avoid(t)$. Thus $u$ has no hazardous color in either bounding chain.

If $s\in I^\circ(w,t)$ and $c_s=\chi(w)$, then by definition of $\NM_{t,H^*}$ we have $c_s'=\chi'(w)$. In the reference run the proposal $\chi(w)$ is blocked by the singleton neighbor $w$, and in the modified run the proposal $\chi'(w)$ is blocked by the same singleton neighbor $w$. Hence in both runs the update at time $s$ leaves the state of $u$ unchanged, so $u\notin \mc P_s^{\eta}=\widetilde{\mc P}_s$ and $(\dagger_s)(3)$ continues to hold.

In all other subcases, $c_s'=c_s$. If $s\in I^\circ(w,t)$, then \eqref{eq:nm-no-neighbor-uses-new-color-rewrite} gives $c_s\ne \chi'(w)$, and because we are not in the previous subcase we also have $c_s\ne \chi(w)$. Thus the carried color of $w$ is irrelevant for the proposal status at $u$, and the update type is the same in both bounding chains. So again $(\dagger_s)(3)$ holds.

\smallskip
\noindent\emph{Case 4: $v_s=v_t$.}
If $s=t$, then $c_t'=H^*\setminus\{c_t\}$ and both colors lie in the same hazardous set $H^*\subseteq H(Z_{t-1}^{\eta},v_t)$. Hence time $t$ is hazardous in both sequences, and the resulting multivalued set at $v_t$ is the same in both bounding chains:
\[
\widetilde Z_t(v_t)=H(Z_{t-1}^{\eta},v_t)\cup Z_{t-1}^{\eta}(v_t)=Z_t^{\eta}(v_t).
\]
Thus $\widetilde{\mc P}_t=\mc P_t^{\eta}$.

Now assume $s\ne t$. Then $c_s'=c_s$. If $w\in W$ and the two runs currently disagree on the singleton color of $w$, then necessarily $s\in I^\circ(w,t)$. In that case, \eqref{eq:nm-vt-does-not-attempt-old-color-rewrite} implies $c_s\ne \chi(w)$, while \eqref{eq:nm-no-neighbor-uses-new-color-rewrite} implies $c_s\ne \chi'(w)$. So the neighbor $w$ blocks neither proposal. All other neighbors of $v_t$ have the same bounding-chain state in the two runs by $(\dagger_{s-1})$. Therefore $c_s$ has the same status in the two bounding chains, and the update at time $s$ has the same effect in both.

\smallskip
These four cases exhaust all possibilities, so $(\dagger_s)$ holds for every $s\le \Tcouple$. In particular,
\[
\widetilde{\mc P}_s=\mc P_s^{\eta}
\qquad\text{for all }0\le s\le \Tcouple,
\]
and
\[
\widetilde Z_s(u)=Z_s^{\eta}(u)
\qquad\text{for every }u\in \mc P_s^{\eta}.
\]
This proves that the two runs generate the same sets $\mc{P}_s$ and the same bounding-chain values on those sets.

It remains to prove that $\eta'$ also satisfies $\BC$. Note first that Condition~(1) in the definition of $\BC$ depends only on the union of the sets $\mc P_s$, and that union is the same for $\eta$ and $\eta'$. For Condition~(2), fix $1\le s\le \Tcouple$ such that
\[
v_s\in \bigcup_{0\le r<s}\widetilde{\mc P}_r=\bigcup_{0\le r<s}\mc P_r^{\eta}.
\]
By the case analysis above, the proposal at time $s$ has the same type (available, blocked, or hazardous) in the two bounding chains. In particular,
\[
 c_s'\in H(\widetilde Z_{s-1},v_s)
 \iff
 c_s\in H(Z_{s-1}^{\eta},v_s).
\]
Since $\eta$ satisfies Condition~(2) in the definition of $\BC$, the right-hand side is false, and therefore the left-hand side is also false. So $\eta'$ satisfies Condition~(2). Finally, since we have shown that $\mc{P} = {\wt{\mc{P}}}$, Condition~(3) is also satisfied, proving that $\BC(X_0,Y_0,\eta')=\tt{True}$.
\end{proof}

\subsection{Validity of local coupling}
\label{sub:validity-local}

Here we prove that the local coupling $\NM_{t,H^*}$ is bijective, and hence is a valid coupling.
In the lemma below, as throughout this section, $X_0$ and $Y_0$ are two labelings with an update sequence $\vecsigma \in (V \times [k])^\Tcouple$.  The discrepancy set $\mc{P}$ follows condition $\BC(X_0,Y_0,\vecsigma) = \tt{True}$.

\begin{lemma}\label{item:local-coupling-involution}
    Suppose $\BC(X_0,Y_0,\vecsigma) = \tt{True}$ and $t$ satisfies \cref{assumption:non-markovian-succeed}.  Then
    \[  \NM_{t,H^*}(\NM_{t,H^*}(\vecsigma))~=~\vecsigma.
    \]
\end{lemma}

\begin{proof}
Set
\[
\vec\sigma' := \NM_{t,H^*}(\vec\sigma),
\qquad
\vec\sigma'' := \NM_{t,H^*}(\vec\sigma').
\]
We will show that $\vec\sigma''=\vec\sigma$.

The single-edit case in the proof of \cref{lem:bounding-chain-is-symmetric} shows that the first pass preserves the bounding chain. Hence $\vec\sigma$ and $\vec\sigma'$ generate the same sets $\mc P_s$ and the same bounding-chain values on those sets. In particular, the auxiliary objects $p$, $H^*$, and the neighborhood $N_{\mc P}(v_t)$ are the same for the first and second pass.

Let
\[
W:=B\cup \alpha(B),
\qquad
\chi(w):=X_{t-1}(w),
\qquad
\chi'(w):=c'_{\tau_w^-}
\qquad (w\in W),
\]
where $\vec\sigma'=((v_s,c_s'))_{s\in [\Tcouple]}$.
By \cref{cor:epoch-behavior-W-rewrite}, the first pass changes the reference evolution on $W$ exactly as follows: for each $w\in W$, the times $\tau_w^-$ and $\tau_w^+$ are unchanged, and on the whole interval $[\tau_w^-,\tau_w^+)$ the color carried by $w$ is changed from $\chi(w)$ to $\chi'(w)$.

\begin{claim}
\label{claim:exchangeable-invariance-rewrite}
For every $z\in N_{\mc P}(v_t)$, the second pass sees the same epoch and the same exchangeable set as the first pass:
\[
I^{\vec\sigma'}(z,t)=I^{\vec\sigma}(z,t),
\qquad
\Exchange^{\vec\sigma'}(z,t)=\Exchange^{\vec\sigma}(z,t).
\]
Furthermore,
\[ \Avoid^{\vecsigma'}(t) = \Avoid^\vecsigma(t). \]
Consequently,
\[
\Swappable_{\mc P}^{\vec\sigma'}(c,t)=\Swappable_{\mc P}^{\vec\sigma}(c,t)
\qquad\text{for every }c\in [k].
\]
\end{claim}

\begin{proof}
Fix $z\in N_{\mc P}(v_t)$.

First suppose that $z\in W$.
By \cref{cor:epoch-behavior-W-rewrite}, the times $\tau_z^-$ and $\tau_z^+$ are unchanged by the first pass, so the epoch $I(z,t)$ is unchanged as well.
Before time $\tau_z^-$, the two runs coincide on the whole neighborhood of $z$, so the availability set at the last successful recoloring time is the same in the two runs.
For every color $d\notin\{\chi(z),\chi'(z)\}$, the first pass does not edit any proposal of color $d$ made by a neighbor of $z$ during $I^\circ(z,t)$, so membership of $d$ in the exchangeable set is unchanged.
The color $\chi'(z)$ belongs to $\Exchange^{\vec\sigma}(z,t)$ by construction, and it belongs to $\Exchange^{\vec\sigma'}(z,t)$ because it is the current color of $z$ at time $t-1$ under $\vec\sigma'$.
The color $\chi(z)$ belongs to $\Exchange^{\vec\sigma}(z,t)$ because it is the current color under $\vec\sigma$, and it belongs to $\Exchange^{\vec\sigma'}(z,t)$ because every interior neighbor-attempt of $\chi(z)$ was changed by the first pass to an attempt of $\chi'(z)$.
Hence
\[
\Exchange^{\vec\sigma'}(z,t)=\Exchange^{\vec\sigma}(z,t).
\]

Now suppose that $z\notin W$.
We claim that no coordinate edited by the first pass lies at $z$ or at a neighbor of $z$ during the open interval $I^\circ(z,t)$.
Indeed, steps~1 and~2 of $\NM_{t,H^*}$ edit only successful recolorings at vertices of $W$, and by acyclicity of $\mc{P}$, no vertex of $W$ is adjacent to $z$.
Step~3 edits only proposals made at neighbors of vertices of $W$; but since $G$ has girth $\geq 7$, this cannot lie in $N(z)$ either. 
Step~4 edits time $t$ itself, but time $t$ is excluded from $I^\circ(z,t)$ by definition.
Therefore the successful recoloring times of $z$, the neighborhood states relevant at those times, and the entire neighbor-attempt pattern on $I^\circ(z,t)$ are unchanged, so both the epoch and the exchangeable set are unchanged.

This proves the first two assertions.
For the swappable sets, note that their definition depends only on the exchangeable sets and on whether there exists $s\in I^\circ(z,t)$ with $v_s=v_t$ and $c_s$ equal to the current color of $z$.
For $z\notin W$, all of these are unchanged.
For $z\in W$, the current color changes from $\chi(z)$ to $\chi'(z)$, but the first pass never edits a time with updated vertex $v_t$ except time $t$ itself, and time $t$ is excluded from $I^\circ(z,t)$.
Moreover, by exchangeability neither $\chi(z)$ nor $\chi'(z)$ is proposed by $v_t$ on $I^\circ(z,t)$.

For $\Avoid(t)$, as $\mc{P}$ is the same in both cases, clearly $N_{v_t}^2(\mc{P})$ is the same in both cases.  The remaining part of $\Avoid(t)$ is
\[ \{ w \in N(v_t) : \exists s \in I^\circ(w,t) \text{ where } v_s \in (N(w) \cap N_w^2(\mc{P})) \cup \{ v_t \}, c_s = X_{s-1}(w) \}. \]

Since $\Avoid(t) \cap N(B \cup \alpha(B)) = \emptyset$, by \cref{lem:temporary-geometry}, vertices in $\Avoid(t)$ and their neighbors have all color attempts untouched by $\NM_{t,H}$.  Thus $\Avoid(t)$ is also invariant under $\NM_{t,H}$.

Thus the swappable sets are unchanged for every color.
\end{proof}

Let $B'$ denote the blocking set computed from $\vec\sigma'$ in the second pass.
By \cref{cor:epoch-behavior-W-rewrite}, every vertex $w\in B$ with $\alpha(w)\ne w$ carries the color
\[
\chi'(w)=\beta_w(X_{t-1}(\alpha(w)))\notin H^*
\]
at time $t-1$ under $\vec\sigma'$, while every vertex $\alpha(w)$ carries the color $c_u$. In the case that $\alpha(w) = w$, the second criterion takes precedence, i.e.~$w$ carries $c_u$.
No vertex of $N_{\mc P}(v_t)\setminus W$ changes color under the first pass.
Since originally the only color from $H^*$ present on $N_{\mc P}(v_t)$ was $c_b$, it follows that after the first pass the only color from $H^*$ present on $N_{\mc P}(v_t)$ is $c_u$, and it appears exactly on $\alpha(B)$.
Therefore
\[
B' = \alpha(B),
\qquad
c_b' = c_u,
\qquad
c_u' = c_b.
\]
By \cref{claim:exchangeable-invariance-rewrite}, the epochs $I(w,t)$ are the same in the two passes.
Hence the complementary-neighbor map in the second pass is the inverse ordered matching:
\[
\alpha'(\alpha(w))=w
\qquad\text{for every }w\in B.
\]
Likewise, because the ordered exchangeable sets are unchanged, the complementary-color map in the second pass is the inverse ordered matching on those ordered sets:
\[
\beta'_{\alpha(w)}\bigl(\chi'(w)\bigr) = X_{t-1}(\alpha(w))
\qquad\text{for every }w\in B\text{ with }\alpha(w)\ne w.
\]
In particular, the second pass is well-defined.

We now compare the coordinates edited by the two passes.
At time $t$,
\[
c''_t = H^*\setminus\{c'_t\} = c_t.
\]
For each $w\in B$, the second pass changes the two successful recolorings by
\[
c''_{\tau_w^-}=c_b=\chi(w),
\qquad
c''_{\tau_{\alpha(w)}^-}=X_{t-1}(\alpha(w))=\chi(\alpha(w)).
\]
These are exactly the original colors at those two times.
Finally, because
\[
W = B'\cup \alpha'(B') = \alpha(B)\cup B,
\]
the second pass acts on exactly the same collection of epochs as the first pass, but with the carried colors reversed.
Thus every interior edit made in the first pass is undone in the second pass: each proposal of $\chi'(w)$ created by the first pass is changed back to a proposal of $\chi(w)$, for every $w\in W$.
All coordinates untouched by the first pass remain untouched by the second pass.
Hence
\[
\vec\sigma'' = \vec\sigma,
\]
which proves the lemma.
\end{proof}

\subsection{Validity of Global Coupling}
\label{sub:validity-global-bijective}

In this section, we show that the global coupling is indeed a valid coupling.

\begin{definition}
    We define the set $\Temp_t$ as follows.  If $f_t = \NM_{t,H}$, let $\Temp_t = B \cup \alpha(B)$ as defined in \cref{sub:non-markovian-local-coupling}.  Else let $\Temp_t = \emptyset$. Let $\Temp_{\leq t} = \bigcup_{s \leq t} \Temp_s$.
\end{definition}

\begin{remark}
  As discussed in \cref{rmk:temporary-vs-persistent}, the set $\Temp_t$ captures the ``temporary'' discrepancies which are necessarily introduced during our non-Markovian coupling. However, since \cref{assumption:non-markovian-succeed} does not assume that $\tau_w^+$ is defined for $w \in B$ to avoid measurability issues, it may be the case that these ``temporary'' discrepancies may actually persist until time $\Tcouple$. Therefore, our argument will need to separately control the set $\Temp_{\leq\Tcouple}$ as well.
\end{remark}

Fix $t\in\{0,1,\dots,\Tcouple\}$. Besides the forward construction
\[
\eta^{(s)}:=f_{\le s}(\vecsigma)\qquad(0\le s\le t),
\]
we consider the \emph{reverse construction} obtained by running the global coupling
\cref{def:global-coupling} with initial ordered pair $(Y_0,X_0)$ and input sequence
$\eta^{(t)}$.  This construction creates local functions $g_s^{(t)}$ for $1\leq s \leq t$.
For this reverse run write 
\[
\zeta^{(0,t)}:=\eta^{(t)},
\qquad
\zeta^{(s,t)}:=g_s^{(t)}\bigl(\zeta^{(s-1,t)}\bigr)
\qquad(1\le s\le t),
\]
and let $D_s^{\leftarrow,t}$ denote the reverse discrepancy sets.
Thus
\[
\zeta^{(t,t)}=g_t^{(t)}\circ\cdots\circ g_1^{(t)}\bigl(\eta^{(t)}\bigr).
\]
When $t$ is fixed we suppress the superscript and simply write
$\zeta^{(s)}$, $g_s$, and $D_s^{\leftarrow}$.

The following observation packages the locality information used repeatedly in the induction.

\begin{remark}
\label{rem:stage-local-data-rewrite}
At stage $s$ of the global construction, once $D_{s-1}$ is known, the branch choice in
\cref{def:global-coupling} depends only on the states at time $s-1$ on
$D_{s-1}\cup N(D_{s-1})$.
If the chosen map is non-Markovian, then by \cref{lem:non-markovian-local-is-bijection} every coordinate it
changes has updated vertex in $N_{\mc P}^2(v_s)$, and the proof of
\cref{item:local-coupling-involution} shows that all auxiliary objects ($B$, $\alpha$,
$\beta$, the relevant epochs, and the exchangeable/swappable sets) are determined from the update
sequence restricted to that same local region.
The same statement applies verbatim to the reverse construction.
\end{remark}

\begin{proposition}
\label{prop:main-global-induction-rewrite}
For every $0\le t\le \Tcouple$, the following hold.
\begin{enumerate}[(i)]
\item If the reverse construction is run on input $\eta^{(t)}$, then
\[
\zeta^{(t,t)}=\vec\sigma.
\]
Equivalently,
\[
g_t^{(t)}\circ\cdots\circ g_1^{(t)}\bigl(\eta^{(t)}\bigr)=\vec\sigma.
\]
\item The forward and reverse discrepancy sets agree at time $t$:
\[
D_t^{\leftarrow,t}=D_t.
\]
\item Outside persistent and temporary discrepancies, the two chains agree:
\[
X_t(v)=Y_t^t(v)
\qquad\text{for every }v\notin D_t\cup \Temp_{\le t}.
\]
\item Every vertex in $D_t$ is a genuine persistent discrepancy, and its two actual colors are
recorded by the bounding chain:
\[
X_t(v)\neq Y_t^t(v)
\qquad\text{and}\qquad
\{X_t(v),Y_t^t(v)\}\subseteq Z_t(v)
\qquad\text{for every }v\in D_t.
\]
\end{enumerate}
\end{proposition}

\begin{proof}
We argue by induction on $t$.

For $t=0$ there is nothing to prove: $\eta^{(0)}=\vec\sigma$, the reverse construction has no steps,
$D_0^{\leftarrow,0}=D_0=\{z^*\}$, and the only discrepancy is the root discrepancy already recorded
by $Z_0$.

Fix $t\ge 1$ and assume the four statements have been proved at time $t-1$.
Set
\[
\eta^{(t)} = f_t\bigl(\eta^{(t-1)}\bigr).
\]
By the induction hypothesis, every vertex of $D_{t-1}$ lies in $\mc P_{t-1}$ and carries the two
actual colors recorded by $Z_{t-1}$.
We distinguish the three branches of \cref{def:global-coupling}. 

\medskip
\noindent\emph{Case 1: $v_t\notin D_{t-1}\cup N(D_{t-1})$.}
Then $f_t=\id$ and $D_t=D_{t-1}$ by definition.
Hence $\eta^{(t)}=\eta^{(t-1)}$, so item~(iii) and item~(iv) are unchanged from time $t-1$.
For the reverse construction, item~(ii) at time $t-1$ gives
\[
D_{t-1}^{\leftarrow,t}=D_{t-1}^{\leftarrow,t-1}=D_{t-1},
\]
and therefore stage $t$ in the reverse run also falls into Case~1 and chooses the identity map.
So
\[
\zeta^{(t,t)}=\zeta^{(t-1,t)}=\vec\sigma
\]
by the induction hypothesis applied at time $t-1$, and item~(i)--item~(ii) follow.

\medskip
\noindent\emph{Case 2: $v_t\in D_{t-1}$.}
Again $f_t=\id$.
By item~(iv) at time $t-1$,
\[
\{X_{t-1}(v_t),Y_{t-1}^{t-1}(v_t)\}\subseteq Z_{t-1}(v_t).
\]
Since $v_t\in \mc P_{t-1}$ and $\BC(X_0,Y_0,\vec\sigma)=\tt{True}$, Condition~(2) in the definition of the
bounding chain implies
\[
c_t\notin H(Z_{t-1},v_t).
\]
Thus the time-$t$ proposal is either available in both chains or blocked in both chains.
If $c_t\in A(Z_{t-1},v_t)$, both chains recolor $v_t$ to $c_t$, so the persistent discrepancy at
$v_t$ disappears and the update rule in Step~(2) gives
\[
D_t=D_{t-1}\setminus\{v_t\}.
\]
If $c_t\in B(Z_{t-1},v_t)$, both chains keep their old color at $v_t$, so
\[
D_t=D_{t-1}.
\]
Either way, item~(iii) and item~(iv) hold at time $t$.

For the reverse construction, item~(ii) at time $t-1$ again gives
$D_{t-1}^{\leftarrow}=D_{t-1}$.
Moreover, by the definition of $\NM_{r,H}$ and \cref{lem:temporary-geometry}, an earlier non-Markovian edit can change a proposal at time $t$ only if the updated vertex $v_t$ lies outside $\mc P$; since here $v_t\in \mc P$, the reverse run sees the same proposal $c_t$ at time $t$.
Hence stage $t$ of the reverse construction falls into the same branch of Step~(2), chooses the
identity map, and updates $D_t^{\leftarrow}$ in the same way.
Therefore item~(i)--item~(ii) hold.

\medskip
\noindent\emph{Case 3: $v_t\in N(D_{t-1})$.}
Because $D_{t-1}\subseteq \mc P_{t-1}$ by the induction hypothesis and because
$G[\mc P \cup \{ v_t \}]$ is acyclic, there is a unique vertex
\[
p\in N(v_t)\cap D_{t-1}.
\]
Set
\[
H:=\{X_{t-1}(p),Y_{t-1}^{t-1}(p)\}.
\]
By item~(iv) at time $t-1$ we have
\[
H\subseteq Z_{t-1}(p).
\]

We first record two structural facts that will be used in all subcases.

\smallskip
\noindent\textbf{Fact 1:} the current stage sees the same local neighborhood in $X$ and in the forward
$Y$-run.
Indeed, every earlier set $\Temp_r$ has $\Temp_r \cap \Avoid(r) = \emptyset$ (by \cref{assumption:non-markovian-succeed}).  If $v_r = p$, then $\Temp_r \subset N(p)$ is disjoint from $N(v_t)$ or $G$ would have a $3$-cycle.  If $v_r \ne p$, then $N_p(v_t) \subset N^2(p) \subset N^2_{v_r}(\mc{P}) \subset \Avoid(r)$
by definition.  Thus
\[
N_p(v_t)\cap \Temp_{\le t-1}=\emptyset.
\]
By item~(iii) at time $t-1$, we therefore have
\begin{equation}
\label{eq:global-local-equality-rewrite}
Y_{t-1}^{t-1}(u)=X_{t-1}(u)
\qquad\text{for every }u\in N_{\mc P}^2(v_t).
\end{equation}
In particular, the color pattern on $N_p(v_t)$ seen by the forward construction is
exactly the reference pattern $X_{t-1}(N_p(v_t))$.

\smallskip
\noindent\textbf{Fact 2:} later edits do not affect earlier reverse stages.
By \cref{rem:stage-local-data-rewrite}, stage $s<t$ of the reverse construction inspects only the
states on $D_{s-1}^{\leftarrow}\cup N(D_{s-1}^{\leftarrow})\subseteq \mc P_{t-1}\cup N(\mc P_{t-1})$ and,
if it is non-Markovian, update-sequence coordinates with updated vertex in $N_{\mc P}^2(v_s)$.
The $t$-th forward edit changes the input only on the side branch $N_{\mc P}^2(v_t)$ and, after
its time-$t$ discrepancy is created, on the descendant component of
$\mc P\setminus \mc P_{t-1}$ containing $v_t$.

By \cref{lem:non-markovian-local-is-bijection}, any earlier update at time $s$ only reads data within $N^2_\mc{P}(v_s)$.  Persistent changes cannot affect earlier non-Markovian couplings as their changes are contained in $\mc{P}$.  Temporary edits cannot affect earlier non-Markovian couplings by \cref{lem:temporary-geometry}.

Thus these regions are disjoint from all local data seen at stages
$s<t$.
Consequently the first $t-1$ reverse local maps on input $\eta^{(t)}$ are exactly the same as on
input $\eta^{(t-1)}$.
By the induction hypothesis at time $t-1$, applying those maps to $\eta^{(t-1)}$ yields
$\vec\sigma$.
Therefore, after the first $t-1$ reverse stages are applied to $\eta^{(t)}=f_t(\eta^{(t-1)})$, the
intermediate sequence is
\begin{equation}
\label{eq:reverse-before-stage-t-rewrite}
f_t(\vec\sigma).
\end{equation}

We now split according to the relevant subcase of
\cref{def:global-coupling}.

\smallskip
\noindent\emph{Subcase 3a: $|H|=1$.}
Then $f_t=\id$ and $D_t=D_{t-1}$.
There is nothing to prove beyond the induction hypothesis.
The reverse run sees the same singleton set $H$ and also chooses the identity map.
Hence item~(i)--item~(iv) hold.

\smallskip
\noindent\emph{Subcase 3b: $|H|=2$ and $c_t\neq Y_{t-1}^{t-1}(p)$.}
Then
\[
f_t=\Jerrum_{t,H}.
\]
If $c_t\notin H$, this map is the identity on the time-$t$ coordinate.
If $c_t\in H$, then necessarily $c_t=X_{t-1}(p)$, and the Jerrum map replaces the forward proposal
$c_t$ by $Y_{t-1}^{t-1}(p)$.
In either event, the two chains see the same update status at $v_t$: if the proposal is available,
both chains recolor $v_t$ to the same color; if it is blocked, both chains keep the old color.
So no new persistent discrepancy is created and
\[
D_t=D_{t-1}.
\]
This proves item~(iii)--item~(iv) in Subcase~3b.

For the reverse construction, the first $t-1$ stages produce the intermediate sequence
$f_t(\vec\sigma)$ by \eqref{eq:reverse-before-stage-t-rewrite}.
At time $t$, the reverse run sees the same unordered pair $H$, now with the roles of the two colors
swapped.
Hence stage $t$ of the reverse construction again chooses the same Jerrum map $\Jerrum_{t,H}$.
Since $\Jerrum_{t,H}$ is an involution,
\[
g_t^{(t)}\bigl(f_t(\vec\sigma)\bigr)
=
\Jerrum_{t,H}\bigl(\Jerrum_{t,H}(\vec\sigma)\bigr)
=
\vec\sigma.
\]
This proves item~(i), and item~(ii) follows because both runs keep the same discrepancy set.

\smallskip
\noindent\emph{Subcase 3c: danger zone, so $|H|=2$ and $c_t=Y_{t-1}^{t-1}(p)\neq X_{t-1}(p)$.}
Now the branch is determined by the intersection
\[
H\cap X_{t-1}(N_p(v_t)) = H\cap X_{t-1}(N_{\mc P}(v_t))
\]
and by the well-definedness condition for the local non-Markovian map.  We record an important fact.

By \eqref{eq:global-local-equality-rewrite}, the forward $Y$-run sees exactly the same local data.
Hence the forward construction lands in the same subcase as the reference definition.

If the chosen map is Jerrum, then the discrepancy behavior is exactly the one encoded in Step~(4) of
\cref{def:global-coupling}: in the first branch both proposals are blocked, so
$D_t=D_{t-1}$; in the remaining Jerrum branches exactly one chain succeeds at $v_t$, so
$D_t=D_{t-1}\cup\{v_t\}$.
Because
\[
H\subseteq Z_{t-1}(p)\subseteq H(Z_{t-1},v_t),
\]
the two post-update colors at every persistent-discrepancy vertex lie in the corresponding
bounding-chain sets, so item~(iii)--item~(iv) hold.
The reverse run sees the same unordered pair $H$ and the same Jerrum branch, so the same argument as
in Subcase~3b gives item~(i)--item~(ii).

Finally, assume that the chosen map is non-Markovian, so
\[
f_t=\NM_{t,H}.
\]
Because of the local agreement \eqref{eq:global-local-equality-rewrite}, the proof of
\cref{lem:non-markovian-local-is-bijection} applies verbatim to the pair of runs
$\{Y_s^{t-1}\}_s$ and $\{Y_s^t\}_s$.
Thus the $t$-th forward edit creates discrepancies only on the temporary set $\Temp_t$ and,
possibly, at the vertex $v_t$ itself.
In the branch
\[
H\cap X_{t-1}(N_p(v_t))=\{c_t\},
\]
no persistent discrepancy is created at time $t$, so
\[
D_t=D_{t-1}.
\]
In the branch
\[
H\cap X_{t-1}(N_p(v_t))=H\setminus\{c_t\},
\]
a new persistent discrepancy is created at $v_t$, so
\[
D_t=D_{t-1}\cup\{v_t\}.
\]
In both branches, item~(iii)--item~(iv) hold because
\cref{lem:non-markovian-local-is-bijection} shows that all other discrepancies are temporary
and that any persistent discrepancy created at time $t$ is recorded by the bounding chain.

For the reverse construction, \eqref{eq:reverse-before-stage-t-rewrite} again shows that after the
first $t-1$ reverse stages the input to stage $t$ is $f_t(\vec\sigma)$.
At that point the local configuration at $p$ and on $N_p(v_t)$ is the same as in the forward run,
with the two colors of $H$ interchanged.
Hence stage $t$ of the reverse construction also chooses the map $\NM_{t,H}$.

By \cref{item:local-coupling-involution},
\[
g_t^{(t)}\bigl(f_t(\vec\sigma)\bigr)
=
\NM_{t,H}\bigl(\NM_{t,H}(\vec\sigma)\bigr)
=
\vec\sigma.
\]

So item~(i) holds, and item~(ii) follows because the same persistent-discrepancy update rule is used
in both directions.

\medskip
The three cases exhaust the possibilities, so the induction is complete.
\end{proof}

\begin{proposition}[global coupling is valid]
\label{prop:coupling-is-valid-rewrite}
For every pair $X_0,Y_0\in \wh\Omega$ that differ at exactly one vertex,
\[
F_{Y_0,X_0}\bigl(F_{X_0,Y_0}(\vec\sigma)\bigr)=\vec\sigma
\qquad\text{for every }\vec\sigma\in (V\times [k])^{\Tcouple}.
\]
In particular, $F_{X_0,Y_0}$ is a bijection.
\end{proposition}

\begin{proof}
If $\BC(X_0,Y_0,\vec\sigma)=\tt{False}$, then
$F_{X_0,Y_0}(\vec\sigma)=\vec\sigma$, and by
\cref{lem:bounding-chain-is-symmetric} the reverse run also uses the identity map.
So the claim is immediate.

Assume now that $\BC(X_0,Y_0,\vec\sigma)=\tt{True}$.
Then the forward construction produces the final intermediate sequence
\[
\eta^{(\Tcouple)}=F_{X_0,Y_0}(\vec\sigma).
\]
Applying \cref{prop:main-global-induction-rewrite} with $t=\Tcouple$ gives
\[
F_{Y_0,X_0}\bigl(\eta^{(\Tcouple)}\bigr)=\vec\sigma.
\]
Since $\eta^{(\Tcouple)}=F_{X_0,Y_0}(\vec\sigma)$, this is exactly the required identity.
Therefore $F_{X_0,Y_0}$ is invertible, hence bijective.
\end{proof}

\section{Analysis of the coupling}
\label{sec:analysis-coupling}

In this section, we provide a probabilistic analysis of the coupling defined in \cref{sec:nonmarkovian}; this will complete the proof of the main coupling result stated in \cref{thm:main-nonmarkovian}.

Throughout the previous section, we have defined the coupling with several constants. We briefly recall what the constants are.  The input parameter $\delta$ is such that $k \geq (1+\delta)\Delta$.  The parameter $\Tcouple = \Ccouple n$ is the number of steps of the coupling.  The parameter $\eps$ is the local uniformity parameter we will input into \cref{thm:local-uniformity-discrete}. The parameter $\Delta_0$ is a lower bound on $\Delta$. 
We will assume the following relationship holds between these constants:
\begin{equation}\label{eq:parameter-inequalities} 1/\Delta_0 \ll \eps \ll 1/\Ccouple \ll \delta. \end{equation}
In particular, this relationship is consistent with \cref{thm:local-uniformity-discrete}. 
These parameters are fixed globally.

For our analysis, we isolate the following good event. 
\begin{definition}\label{def:good-event-G}
    Let $\mc{G} = \mc{G}_{(X_0,Y_0)}$ be the event that both the bouding chain is well-behaved and no non-Markovian updates fail due to \cref{assumption:non-markovian-succeed}. Formally,
    \begin{multline*}
    \mc{G} = \mc{G}_{(X_0,Y_0)} := \{ \vecsigma \in (V\times [k])^{\Tcouple}: \BC(X_0,Y_0,\vecsigma) = {\tt True} \} \\ \cap \{ \vecsigma \in (V \times [k])^\Tcouple : \nexists t \text{ satisfying \cref{condition:nm-prelim} but not \cref{assumption:non-markovian-succeed}} \}. \end{multline*}
\end{definition}

\subsection{Analysis of the bounding chain}\label{sub:controlling-bounding-chain}

First, we control the probability that $|\mc{P}|$ becomes too large. This requires a careful argument, since $|\mc{P}|$ is heavy-tailed and its expectation is not bounded by a function of $\Ccouple$.  

\begin{proposition}\label{lem:P-total-size-bound}
    $\mb{P}[|\mc{P}| \geq \exp(\exp(O(\Ccouple)))] \leq \exp(-\Ccouple^2)$.
\end{proposition}

Let $\gamma = \exp(-\Ccouple^2)$ be the desired probability bound. We discretize time and control both $|Z_t(v)|$ and $|\mc{P}_t|$ iteratively.  Initialize
\[ T_0 = 0, \qquad \ell_0 = 2, \qquad m_0 = 1 \]
and for $0 \leq j < 100\Ccouple$, let
\[ T_{j+1} = T_j + \frac{n}{100}, \qquad \ell_{j+1} = 2\ell_j + \frac{\log(200 \Ccouple m_j/\gamma)}{\log(25)}, \qquad m_{j+1} = \frac{200 \Ccouple}{\gamma} m_j \exp(\ell_{j+1}/100). \]

Define the good event
\[ \mc{E}_j = \{ |\mc{P}_{\leq T_j}| \leq m_j \} \cap \left\{ \max_{v \in V, t \leq T_j} |Z_t(v)| \leq \ell_j \right \}. \]

\begin{lemma}\label{lem:controlling-P-one-block}
    For all $0\leq j < 100\Ccouple$, we have
    \[ \mb{P}[\mc{E}_{j+1} \mid \mc{E}_j] \geq 1 - \frac{\gamma}{100\Ccouple}. \]
\end{lemma}

\begin{proof}
    We work with the two failure conditions separately. First, suppose there is $T_j \leq t \leq T_{j+1}$ a vertex $v \in V$ with $|Z_t(v)| \geq \ell_{j+1}$.  Let
    \[ d_j := \ell_{j+1} - \ell_j = \ell_j + \frac{\log(200 \Ccouple m_j/\gamma)}{\log(25)}. \]

    By $\mc{E}_j$, we can control (in particular) $\max_{v \in \mc{P}_{T_j}} |Z_{T_j}(v)|$.  To see how quickly this maximum can grow, write
    \[
        M_t:=\max_{u\in V}|Z_t(u)|.
    \]
    Whenever a hazardous update changes $Z_{t-1}(v_t)$, we have $v_t\notin\mc{P}_{t-1}$ and hence $|Z_{t-1}(v_t)|=1$.  The chosen witness $w\in\mc{P}_{t-1}$ satisfies $|Z_{t-1}(w)|\leq M_{t-1}$, and therefore
    \[
        |Z_t(v_t)|=|Z_{t-1}(w)\cup Z_{t-1}(v_t)|\leq M_{t-1}+1.
    \]
    The other update rules do not increase the maximum, so $M_t\leq M_{t-1}+1$ at every step.

    Moreover, each newly created nonsingleton set has a witness at the preceding time.  Tracing these witnesses backwards from a set of size at least $\ell_{j+1}=\ell_j+d_j$, we eventually reach a vertex of $\mc{P}_{T_j}$; uncertainty cannot be created after time $T_j$ without propagating from a vertex that is already in $\mc{P}$.  Since each propagation can increase the size along this genealogy by at most one, the genealogy contains at least $d_j$ hazardous updates, in chronological order, during $[T_j,T_{j+1}]$.  Taking its first $d_j$ steps gives a starting point $v \in {\mc{P}_{T_j}}$ and a walk of length $d_j$ of the required type.

    {Fix a walk $u_1,u_2,\ldots,u_{d_j}$.  In fact, we have the slightly stronger statement that for each step in this walk, we have $u_{i-1}$ is the chosen ``witness'' for $u_{i}$'s hazardous update so that $Z_t(u_i) = Z_{t-1}(u_i) \cup Z_{t-1}(u_{i-1})$.  In particular, this means that the chosen color is in not just $H(Z_{t-1},u_i)$ but in $Z_{t-1}(u_{i-1})$.  Thus the probability this walk was taken as hazardous updates in $[T_j,T_{j+1}]$ (an interval of length $n/100$) is at most}
    \[ \binom{n/(100)}{d_j} \prod_{i=1}^{d_j} \paren{\frac{\ell_j + i}{kn}} = \frac{(n/100)_{d_j}}{n^{d_j}} k^{-d_j} \binom{\ell_j + d_j}{d_j} \leq (100k)^{-d_j} 4^{d_j}, \]
where the last inequality used $d_j \geq \ell_j$.  Now, we union bound over all possible walks of length $d_j$ starting at a vertex in $\mc{P}_{T_j}$.  This is at most $m_j \Delta^{d_j}$ since $\mc{E}_j$ holds and $G$ has max degree $\Delta$.  Thus
    \begin{multline}
        \label{first-failure}
    \mb{P}\left[\max_{v \in V, t \leq T_{j+1}} |Z_t(v)| > \ell_{j+1} \mid \mc{E}_j\right] \\
    \leq m_j \Delta^{d_j} (100k)^{-d_j} 4^{d_j} \leq m_j 25^{-d_j} = \frac{\gamma}{200\Ccouple} 25^{-\ell_j} \leq \frac{\gamma}{200\Ccouple}. 
    \end{multline} 
    by the definition of $d_j$.

    Second, we control $|\mc{P}_{{\leq} T_{j+1}}|$ given $\mc{E}_j$ and $|Z_t(v)| \leq \ell_{j+1}$ for all $t \leq T_{j+1}$.  Consider that
    \[ \mb{E}[|\mc{P}_{{\leq} t+1}| \mbf1\{ \max_{v \in V} |Z_t(v)| \leq \ell_{j+1} \} \mid \mc F_t] \leq |\mc{P}_{\leq t}| + \frac{|\mc{P}_{\leq t}| \Delta \ell_{j+1}}{kn} \leq |\mc{P}_{\leq t}| e^{\ell_{j+1}/n} \]
    and so
   \[
    \mb{E}[|\mc{P}_{\leq T_{j+1}}| \mbf1\{ \max_{v \in V, t \leq T_{j+1}} |Z_t(v)| \leq \ell_{j+1} \} \mid \mc E_j] \leq m_j e^{\ell_{j+1}/100}. 
    \]
    Thus, by Markov's inequality, we have the desired bound for growth by plugging in the recursive definition of $m_{j+1}$:
     \begin{equation}
        \label{second-failure}
        \mb{P}\left[|\mc{P}_{\leq T_{j+1}}|>m_{j+1},\ \max_{v\in V,\ t\leq T_{j+1}}|Z_t(v)|\leq\ell_{j+1} \mid \mc{E}_{j}\right] \leq \frac{m_j\exp(\ell_{j+1}/100)}{m_{j+1}} \leq \frac{\gamma}{200\Ccouple}.
    \end{equation}
    Finally, by combining \cref{first-failure,second-failure} with the definition of $\mc{E}_{j+1}$ we have the following: 
    \[ \mb{P}[\mc{E}_{j+1} \mid \mc{E}_j] \geq 1 - \frac{\gamma}{100\Ccouple}. \qedhere \]
\end{proof}

\begin{proof}[Proof of \cref{lem:P-total-size-bound}]
    By \cref{lem:controlling-P-one-block}, we have $\mb{P}[\mc{E}_{100\Ccouple}] \geq 1 - \gamma$.  Thus except with probability $\gamma$, we have $|\mc{P}| \leq m_{100\Ccouple}$.  By a computation, if $\gamma = \exp(-\Ccouple^2)$, then $m_{100\Ccouple} \leq \exp(\exp(O(\Ccouple)))$.  See \cref{app:computation} for details of this computation. 
\end{proof}

Given this, we can now control the probability of $\BC$ being false. 

\begin{lemma}\label{lem:small-error-if-nonmarkovian-fails}
    Let $X_0,Y_0 \in \wh\Omega$ be neighboring labelings.  Then
    \[ \mb{P}_\vecsigma[\BC(X_0,Y_0,\vecsigma)] \geq 1 - 2\exp(-\Ccouple^2). \]
\end{lemma}

\begin{remark}
\label{rmk:girth}
    We will need the probability to be sufficiently small as a function of $\Ccouple$.  The exact choice $\exp(-\Ccouple^2)$ is arbitrary and can be made into any function of $\Ccouple$ using the machinery in \cref{sub:controlling-bounding-chain}. Moreover, the analysis of the event $\mc{B}_3$ in the proof below is where the requirement that the girth is at least $7$ enters the picture.  
\end{remark}

\begin{proof}
    Recall that $\BC(X_0,Y_0,\vecsigma) = \tt{False}$ if either (i) the set $\mc{P}$ created is too large, (ii) the set $\mc{P}$ is close to completing a cycle or (iii) there is a vertex $v \in V$ which has at least two hazardous updates.

    We begin by decomposing $\BC(X_0,Y_0,\vecsigma)$ into these three potential issues.
    Accordingly, we define a series of bad events.
    Informally, our first bad event $\mc{B}_1$ is that $\mc{P}_{\leq\Tcouple}$ is too large.  Our second event $\mc{B}_2$ is the event that $\mc{P}_{\leq\Tcouple}$ remains small and still contains a repropagation.  Our third event $\mc{B}_3$ is the event that $\mc{P}_{\leq\Tcouple}$ is small and still contains a cycle.
    
    \begin{itemize}
        \item $\mc{B}_1 = \{ |\mc{P}| \geq \exp(\exp(O(\Ccouple))) \}$.
        \item $\mc{B}_2 = \overline{\mc{B}_1} \cap \{ \vecsigma : \overline{\BC(X_0,Y_0,\vecsigma)} \text{ fails condition }(2)\text{ (repropagation)} \}$.
        \item $\mc{B}_3 = \overline{\mc{B}_1} \cap \overline{\mc{B}_2} \cap \{ \vecsigma : \overline{\BC(X_0,Y_0,\vecsigma)} \text{ fails condition }(1)\text{ (acyclicity)} \}$
    \end{itemize}

    Notice that $\{ \BC(X_0,Y_0,\vecsigma) = \tt{False} \} \subset \mc{B}_1 \cup \mc{B}_2 \cup \mc{B}_3$ by design.

    We have $\mb{P}[\mc{B}_1] \leq \exp(-\Ccouple^2)$ by \cref{lem:P-total-size-bound}.

    For $\mc{B}_2$, notice that at step $t$, to repropagate, we must have $v_t \in \mc{P}_{\leq t-1}$ and $c_t \in H(Z_{t-1},v_t)$.  However, $|\mc{P}_t| \leq \exp(\exp(O(\Ccouple)))$ by $\ol{\mc{B}_1}$ and $|H(Z_t,v_t)| \leq |\mc{P}_t| \cdot \max_{w \in V} |Z_t(w)|$.  If we have not yet repropagated, then $|Z_t(w)| \leq |\mc{P}_t|$, so $|H(Z_t,v_t)| \leq \exp(\exp(O(\Ccouple)))$ by $\ol{\mc{B}_1}$.  Thus the expected number of repropagations is at most $\Tcouple \exp(\exp(O(\Ccouple))) / (kn) = \exp(\exp(O(\Ccouple))) / k = o(1)$ as $k > \Delta_0 \gg \Ccouple$.  The result follows by Markov's inequality.

    For $\mc{B}_3$, suppose that an attempt $(v,c)$ would break the cycle condition.  Then we know $v \in N(p)$ for some $p \in \mc{P}_{t-1}$ and furthermore there is some $v' \in V$ so that $G[\mc{P}_{t-1} \cup \{v,v'\}]$ contains a cycle; thus there is some $w \in \mc{P}_{t-1}$ with $v' \in N(v) \cap N(w)$.  Thus we have $p,w \in \mc{P}_{t-1}$ with the path $pvv'w$.  As the graph has girth $\geq 7$, there is at most one path of length $3$ between two vertices.  Thus $p$ and $w$ uniquely determine $v$.  Thus there are at most $|\mc{P}_{t-1}|^2$ choices of $v$ and $|\mc{P}_{t-1}|$ choices of $c$ (by $\ol{\mc{B}_2}$) and so by a union bound, the probability of creating a cycle is at most $\Tcouple \exp(\exp(O(\Ccouple))) /(kn) = o_{\Delta\to\infty}(1)$.

    By enlarging $\Delta_0=\Delta_0(\Ccouple)$, the sum of the bounds for $\mc B_2$ and $\mc B_3$ is at most $e^{-\Ccouple^2}$. Together with $\mb P(\mc B_1)\le e^{-\Ccouple^2}$, the union bound proves the claim.
\end{proof}

\subsection{Putting it all together}

We will use the following result to turn high probability bounds into expectation bounds.

\begin{lemma}[{\cite[Observation 21]{hayesvigoda-nonmarkovian}}]\label{lem:HV-exponential-RV-bound}
    Let $X$ be an exponential random variable with mean $\mu$ and let $A$ be an event of probability $p$.  Then
    \[ \mb{E}[X \mbf1\{A\}] \leq p(\mu\log(e/p) + 1). \]
\end{lemma}

We will need three propositions to control the contribution from non-Markovian errors.

\begin{proposition}
\label{lem:probability-nonmarkovian-fails}
    Let $v \in V$ and suppose $\LU(X_0,\vecsigma,\eps,v)$ holds for $X_0 \in \wh{\Omega}$. 
    Let $0 < \Cmod < \Ccouple$, let $t > \Cmod n$, and let $H \in \binom{[k]}{2}$. With probability $1 - \exp(-\Delta^{2/3})$, $v$ has at most $e^{-\Cmod/100}\Delta$ neighbors that follow all of \cref{condition:nm-prelim} but the last bullet, i.e.~they are blocked by a vertex who has never been updated.
\end{proposition}

The reason for isolating $\Cmod$ is that in the first few steps, it is unlikely that non-Markovian updates will succeed as there will be a blocker who has never been recolored.  However, after a buffer of $\Cmod n$ steps, with high probability (in $\Cmod$), all blockers have been updated at least once and so are eligible to be recolored. While the above proposition is valid for any $\Cmod < \Ccouple$, we will later assume that $\Cmod$ has size
\[1/\Delta_0 \ll \eps \ll 1/\Ccouple \ll 1/\Cmod \ll \delta.\]

\begin{proof}
    We will simply count the number of neighbors of $v$ who are blocked for color $c_b$ that have not been recolored since time $0$.  Let $\gamma = e^{-\Cmod}$.  By \eqref{eq:two-nbd-uniformity}, there are at most $\gamma^{1/10} \Delta$ neighbors $w \in N(v)$ who have more than $400\gamma^{-1/10}$ neighbors $z \in N(w)$ initially colored $c_b$.  We will assume these vertices all fail the last bullet.

    Then, for the remaining neighbors $w \in N(v)$, we will use a union bound over $z \in N(w)$.  Using \eqref{eq:main-number-colors}, each $z$ has been updated except with probability
    \[ \prod_{s=1}^t \paren{1 - \frac{|A(X_{s-1},z)|}{kn}} \leq \exp\paren{ - \frac{t}{en} } \leq \gamma^{1/e}. \]

    Thus by a union bound, all have been updated except with probability $\gamma^{1/e - 1/10} \leq \gamma^{1/4}$.  Since these are negatively correlated, we may use a Chernoff bound.  Thus, at most $2\gamma^{1/4}\Delta$ vertices are blocked by an unupdated vertex except with probability $\exp(-\Delta^{3/4})$.  Combining with the $\gamma^{1/10}\Delta$ neighbors earlier excluded, we get the result.
\end{proof}

\begin{restatable}{proposition}{propprobg}
\label{prop:non-markovians-all-succeed}
Let $X_0$ and $Y_0$ be neighboring labelings in $\wh{\Omega}$.  Then
\[\mb{P}[\ol{\mc{G}} \cap \{\BC(X_0,Y_0,\vecsigma) = {\tt True} \} \cap \mc{U}] \leq \eps^{1/3}.\]
\end{restatable}

We defer the proof to \cref{sec:non-markovian-inequalities}, as we will need to use ideas developed in our local uniformity proofs in \cref{sec:local-uniformity} to approximately decouple the future randomness in $\vecsigma$ from various non-Markovian events which necessarily look into the future.

\begin{restatable}{proposition}{proptemp}\label{prop:temporary-errors-are-mostly-fixed}
    Let $\mc U:=\{\LU(X_0,\vecsigma,\eps,z^*)\}$. Then
    \[ \mb{E}[|(X_\Tcouple \oplus Y_\Tcouple) \cap \Temp_{\leq \Tcouple}| \mbf1\{\mc U\cap\mc{G}\}] \leq \frac19. \]
\end{restatable}

We defer the proof of this result to \cref{sec:non-markovian-inequalities} as well.

\begin{proof}[Proof of \cref{thm:main-nonmarkovian}]
    Let
    \[
        \mc{U}
        :=
        \left\{
            \operatorname{LU}(X_0,\vecsigma,\eps,z^*)
        \right\}.
    \]
    We must prove
    \[
        \mb{E}\left[
            |X_\Tcouple \oplus Y_\Tcouple| \mbf1\{\mc{U}\}
        \right]
        \leq \frac13.
    \]

    \medskip
    \noindent\textbf{Step 1: separating the event $\ol{\mc{G}}$.}
    On the event $\{\BC(X_0,Y_0,\vecsigma)=\tt{False}\}$ the global coupling is, by definition, the identity coupling.  Hence the resulting Hamming distance is stochastically dominated by the standard branching-process bound for the identity coupling; see \cite{jerrum1995very}.  Since \cref{lem:small-error-if-nonmarkovian-fails} gives
    \[
        \mb{P}[\BC(X_0,Y_0,\vecsigma)={\tt False}] \leq 2e^{-\Ccouple^2},
    \]
    \cref{lem:HV-exponential-RV-bound} implies
    \begin{equation}\label{eq:bad-BC-expectation}
        \mb{E}\left[
            |X_\Tcouple\oplus Y_\Tcouple|
            \mbf1\{\mc{U}\cap\{\BC={\tt False}\}\}
        \right]
        \leq
        e^{O(\Ccouple)} e^{-\Ccouple^2}
        =
        o(1).
    \end{equation}
    We must also consider the event that there is some possible failed non-Markovian update.  In this case, we will bound that $D_t \subset \mc{P}$ which has size at most $\exp(\exp(O(\Ccouple)))$.  To control $\Temp_{\leq t}$, notice that there are at most $|\mc{P}|$ values $t$ for which $\Temp_t \ne \emptyset$.  For each of those values of $t$, $|\Temp_t|$ is stochastically dominated by $2$ times a binomial random variable $\on{Binom}(\Delta,e/k)$.  Thus $|\Temp_{\leq t}|$ is stochastically dominated by $2$ times a Poisson random variable with mean $e \exp(\exp(O(\Ccouple)))$, which is in turn stochastically dominated by an exponential random variable with mean $\exp(\exp(O(\Ccouple)))$.  Thus by \cref{lem:HV-exponential-RV-bound,prop:non-markovians-all-succeed}, $\eps \ll 1/\Ccouple$, and \eqref{eq:bad-BC-expectation} we have
    \begin{equation}\label{eq:bad-olG-expectation} \mb{E}[\mbf1\{ \mc U\cap\ol{\mc{G}} \} |X_\Tcouple \oplus Y_\Tcouple|] = o(1). \end{equation}
    
    \medskip
    \noindent\textbf{Step 2: persistent and temporary discrepancies on $\{\BC=\tt{True}\}$.}
    On $\{\BC(X_0, Y_0, \vecsigma)=\tt{True}\}$ the sets $D_t$ and $\Temp_{\leq t}$ are defined for every $t$.  Since these sets need not be defined on $\ol{\mc G}$, we define the following random variables piecewise:
    \[
        \rho(t)
        :=
        \begin{cases}
        \mbf1\{\mc U\}\displaystyle\sum_{w\in D_t}e^{d(w)/k},&\vecsigma\in\mc G,\\
        0,&\vecsigma\notin\mc G,
        \end{cases}
        \qquad
        \nu(t)
        :=
        \begin{cases}
        \mbf1\{\mc U\}|(X_t\oplus Y_t)\cap\Temp_{\leq t}|,&\vecsigma\in\mc G,\\
        0,&\vecsigma\notin\mc G.
        \end{cases}
    \]
    By \cref{prop:coupling-is-valid-rewrite}, whenever $\mc{U}$ and $\BC(X_0,Y_0,\vecsigma)=\tt{True}$ hold,
    \[
        X_\Tcouple \oplus Y_\Tcouple \subset D_\Tcouple \cup \Temp_{\leq \Tcouple}.
    \]
    Therefore
    \begin{equation}\label{eq:decompose-final-disagreements}
        |X_\Tcouple\oplus Y_\Tcouple|\mbf1\{\mc{U}\}
        \leq
        |X_\Tcouple\oplus Y_\Tcouple|\mbf1\{\mc{U}\cap\ol{\mc{G}}\}
        + \rho(\Tcouple) + \nu(\Tcouple).
    \end{equation}

    We next record the measurability input needed for the drift calculation.

    \begin{claim}\label{claim:measurability-of-Y}
        Let $\mc{F}_t$ denote the $\sigma$-algebra generated by the first $t$ coordinates of $\vecsigma$.  There are $\mc F_t$-measurable sets $D_t^\circ$ and colors $\{Y_t^\circ(v):v\in D_t^\circ\}$, defined for every update sequence, such that on $\mc G$ they agree with $D_t$ and $\{Y_t(v):v\in D_t\}$, respectively.
    \end{claim}

\begin{proof}
    The arguments in \cref{sec:validity} show these results; we sketch the proof here for completeness.  We induct on time $t$.  At $t=0$, set $D_0^\circ=\{z^*\}$ and $Y_0^\circ(z^*)=Y_0(z^*)$.

    For the induction step, follow the recursion defining $D_t$ and the colors on $D_t$, with one convention: whenever \cref{condition:nm-prelim} holds, use the intended successful non-Markovian branch, without inspecting the future-dependent \cref{assumption:non-markovian-succeed}.  If some data required by this recursion are not defined, use the identity branch and leave $D_t^\circ$ unchanged.  If $v_t \in D_{t-1}^\circ$ and remains in $D_t^\circ$, then this update was blocked in the $\{X_s\}$ chain, so the corresponding color on the successful branch is unchanged.

    If $v_t \in D_t^\circ \setminus D_{t-1}^\circ$, then $H=\{X_{t-1}(p),Y_{t-1}^\circ(p)\}$ is $\mc F_{t-1}$-measurable by induction.  Furthermore, \cref{condition:nm-prelim} is $\mc F_t$-measurable.  Thus the successful-branch convention determines $D_t^\circ$ and $Y_t^\circ(v_t)$ from the first $t$ updates.

    On $\mc G$, \cref{condition:nm-prelim} implies \cref{assumption:non-markovian-succeed}, so this convention agrees with the actual recursion: in the non-Markovian case $Y_t(v_t)=X_{t-1}(p)$, while in the Jerrum case the update can be simulated from the prefix.  This proves both the asserted measurability and agreement on $\mc G$.
\end{proof}

    For the drift calculation, set
    \[
        R(t):=\sum_{w\in D_t^\circ}e^{d(w)/k}.
    \]
    This is an $\mc F_t$-measurable random variable defined on the whole probability space, and the preceding claim gives
    \begin{equation}\label{eq:rho-agrees-with-prefix-potential}
        \rho(t)=\mbf1\{\mc U\cap\mc G\}R(t).
    \end{equation}

    \medskip
    \noindent\textbf{Step 3: a generic one-step bound and the tail estimate.}
    Regardless of whether the successful-branch recursion at time $t$ uses the identity, Jerrum, or non-Markovian rule, a single persistent discrepancy can create new persistent discrepancies only at neighboring vertices, and each such new disagreement has weight at most $e$ times the weight of its parent.  Consequently,
    \begin{equation}\label{eq:generic-rho-growth}
        \mb{E}[R(t+1)\mid \mc{F}_t] \leq \left(1+\frac{\Delta e}{nk}\right)R(t) \leq \left(1+\frac{e}{n}\right)R(t)
        \qquad\text{for every }t<\Tcouple.
    \end{equation}

    \medskip
    \noindent\textbf{Step 4: negative drift after buffer.}
    We now assume $t>\Cmod n$ and work at a sample point in $\mc U\cap\mc G$.  By \cref{claim:measurability-of-Y}, $D_t^\circ=D_t$ at this sample point.  Fix $p\in D_t^\circ$ and consider the ways in which the discrepancy at $p$ can create a new persistent discrepancy at time $t+1$ in the successful-branch recursion.

    There are two contributions.
    \begin{enumerate}
        \item[(a)] \emph{The intended propagating branch.}  By \eqref{eq:main-blockers}, applied with the two disagree colours at $p$, the total weighted contribution of vertices $v\in N(p)\setminus \mc{P}_{\leq t}$ for which the actual local rule would propagate the disagreement from $p$ is at most $d(p)+O(\eps\Delta)$.
        \item[(b)] \emph{Failure of the last bullet of \cref{condition:nm-prelim}.}  By \cref{lem:probability-nonmarkovian-fails}, except with probability $e^{-\Delta^{2/3}}$ there are at most $e^{-\Cmod/100}\Delta$ such vertices.
    \end{enumerate}

    Since each such vertex can matter only when it is chosen together with the unique disagree colour at $p$, the expected increase in $R$ produced by descendants of $p$ is at most
    \begin{equation}\label{eq:births-from-p}
        \frac{d(p)+O(\eps\Delta)+e^{-\Omega(\Cmod)}\Delta}{nk}
        \leq
        \frac1n\left(\frac{1}{1+\delta}+O(\eps)+e^{-\Omega(\Cmod)}\right)
        \leq
        \frac1n\left(1-\frac{\delta}{2}\right),
    \end{equation}
    after first choosing $\Cmod$ sufficiently large and then $\eps$ sufficiently small.

    On the other hand, the discrepancy at $p$ disappears whenever we update $p$ with an available colour.  By the first conclusion of \cref{thm:local-uniformity-discrete},
    \[
        \frac{|A(X_t,p)|}{kn} e^{d(p)/k}
        \geq
        \frac1n(1-\eps e).
    \]
    Summing over $p\in D_t^\circ$ and using \eqref{eq:births-from-p}, we obtain, at every sample point in $\mc U\cap\mc G$,
    \begin{align}
        \mb{E}[R(t+1)-R(t)\mid \mc{F}_t]
        &\leq
        \sum_{p\in D_t^\circ}
        \left[
            \frac1n\left(1-\frac{\delta}{2}\right)
            -
            \frac1n(1-\eps e)
        \right]
        \notag\\
        &\leq
        -\frac{\delta}{3n}|D_t^\circ|
        \leq
        -\frac{\delta}{3en}R(t).
        \label{eq:negative-drift-rho}
    \end{align}

    We now combine these estimates without conditioning any one-step expectation on the future event $\mc U\cap\mc G$.  Put
    \[
        a_-:=1-\frac{\delta}{3en},
        \qquad
        a_+:=1+\frac en,
    \]
    and define the $\mc F_t$-measurable multiplier
    \[
        a_t:=
        \begin{cases}
        a_-,&t>\Cmod n\text{ and }\mb E[R(t+1)\mid\mc F_t]\leq a_-R(t),\\
        a_+,&\text{otherwise}.
        \end{cases}
    \]
    The generic estimate \eqref{eq:generic-rho-growth} gives
    \[
        \mb E[R(t+1)\mid\mc F_t]\leq a_tR(t)
    \]
    on the whole probability space.  Hence, with
    \[
        Q(0):=1,
        \qquad
        Q(t):=\prod_{s=0}^{t-1}a_s,
        \qquad
        M(t):=\frac{R(t)}{Q(t)},
    \]
    the process $M(t)$ is a nonnegative supermartingale.  In particular,
    \[
        \mb E[M(\Tcouple)]\leq M(0)=R(0)=e^{d(z^*)/k}\leq e.
    \]

    At every sample point in $\mc U\cap\mc G$, \eqref{eq:negative-drift-rho} implies $a_t=a_-$ for every $t>\Cmod n$.  Therefore, on this event,
    \[
        Q(\Tcouple)
        \leq
        a_+^{\Cmod n}a_-^{\Tcouple-\Cmod n}.
    \]
    Using \eqref{eq:rho-agrees-with-prefix-potential} only at the terminal time, we conclude that
    \begin{align}
        \mb E[\rho(\Tcouple)]
        &=\mb E\left[\mbf1\{\mc U\cap\mc G\}M(\Tcouple)Q(\Tcouple)\right]
        \notag\\
        &\leq
        a_+^{\Cmod n}a_-^{\Tcouple-\Cmod n}\mb E[M(\Tcouple)]
        \notag\\
        &\leq
        e\left(1+\frac{e}{n}\right)^{\Cmod n}
        \left(1-\frac{\delta}{3en}\right)^{\Tcouple-\Cmod n}
        \notag\\
        &\leq
        \exp\left(1+e\Cmod-\frac{\delta}{4e}\Ccouple\right).
        \label{eq:rho-final-bound}
    \end{align}
    Therefore, by choosing $\Ccouple$ sufficiently large after $\Cmod$ has been fixed, we may ensure
    \begin{equation}\label{eq:bad-rho-expectation}
        \mb{E}[\rho(\Tcouple)] \leq \frac19.
    \end{equation}

    \medskip
    \noindent\textbf{Step 5: temporary discrepancies.}
    By \cref{prop:temporary-errors-are-mostly-fixed},
    \[
        \mb{E}[\nu(\Tcouple)] \leq 1/9.
    \]

    \medskip
    \noindent\textbf{Step 6: Putting everything together.}
    Using \eqref{eq:bad-olG-expectation}, \eqref{eq:bad-rho-expectation}, and \cref{prop:temporary-errors-are-mostly-fixed} in \eqref{eq:decompose-final-disagreements}, we obtain
    \[
        \mb{E}\left[
            |X_\Tcouple \oplus Y_\Tcouple| \mbf1\{\mc{U}\}
        \right]
        \leq
        \frac19 + \frac19 + \frac19 = \frac13
    \]
    for all sufficiently large $\Delta$, which is exactly the expectation bound in \cref{thm:main-nonmarkovian}.  Together with the tail estimate from Step~3, this completes the proof.
\end{proof}

\newcommand{\Pbb}{\mathbb P}
\newcommand{\Ebb}{\mathbb E}
\newcommand{\ind}{\mathbf 1}
\newcommand{\Gin}{G_{\mathrm{in}}}
\newcommand{\Yst}{Y^{*}}
\newcommand{\Xst}{X^{*}}
\newcommand{\Zst}{Z^{*}}
\newcommand{\set}[1]{\left\{#1\right\}}
\newcommand{\abs}[1]{\left|#1\right|}
\newcommand{\brac}[1]{\left[#1\right]}
\newcommand{\Po}{\mathrm{Pois}}
\newcommand{\Unif}{\mathrm{Unif}}

\section{Local Uniformity Properties of the Metropolis Glauber Dynamics}\label{sec:local-uniformity}
In this section, we establish local uniformity properties for the Metropolis Glauber dynamics. Local uniformity properties are those local properties (i.e.~properties depending on a ``small'' neighborhood of a vertex) which hold with high probability for a uniformly random coloring. We show that the labelings generated by the Metropolis Glauber dynamics satisfy these local uniformity properties with high probability, after sufficiently many steps.

Similar results were obtained in \cite{hayes2013local} for the heat-bath version of the Glauber dynamics. While we follow the general approach of Hayes, there is a key aspect of the Metropolis dynamics which requires new ingredients to analyze, namely, the successful-refresh rate at a vertex $w$, which is $|A(X_t, w)|/k$, is correlated with the colors in $S_2(v)$. In \cref{def:Zstar}, we introduce an auxiliary process $Z^*$ to decouple this dependence. In fact, even for the heat-bath version, there seem to be some gaps in \cite{hayes2013local}, which can be fixed using our techniques.

All chains are run on the extended state space $\wh{\Omega} = [k]^V$. For a labeling $X \in \wh{\Omega}$ and a vertex $u \in V$, we write
\[A(X,u) := [k] \setminus X(N(u))\]
for the set of colors available at $u$ under $X$. We also use the variant
\[A_v(X,u) := [k]\setminus X(N(u)\setminus \{v\}),\]
which ignores the color at $v$.

Our main result establishes that the discrete-time Metropolis Glauber dynamics achieves local uniformity after sufficiently many time steps.  We first define the $i$-times blocked subset.

\begin{definition}
\label{def:i-times-blocked-subset}
    For a labeling $X \in \wh{\Omega}$, a vertex $v \in V$, a subset of vertices $S \subseteq N(v)$, 
    a color $c \in [k]$, and a non-negative integer $i$, we define 
    \[ S_{c,i}(X) = \{ w \in S : |(N(w) \setminus \{v\}) \cap X^{-1}(c)| = i \}. \]
    We say that $S_{c,i}(X)$ is the subset of $S$ which is $i$ times blocked for $c$. 
\end{definition}

\begin{theorem}\label{thm:local-uniformity-discrete}
    Given $\delta, \eps > 0$, there exists constants $C = C(\delta)$ and $\Delta_0 = \Delta_0(\delta, \eps)$ such that the following holds.
    Let $G = (V,E)$ be a graph of maximum degree $\Delta > \Delta_0$ and girth at least $7$, and let $k \geq (1+\delta)\Delta$.  Let $\{X_t\}_{t \geq 0}$ be the discrete-time Metropolis dynamics on $\wh{\Omega}$. Fix a vertex $v \in V$, an interval length $T\geq n = |V|$, and a starting time $T_0 \geq Cn\log \Delta$. 
    
    Then,
    \begin{equation*}
        \mb{P}\left[ \exists t \in [T_0,T_0+T] : \big||A(X_t,v)| - k e^{-d(v)/k} \big| > \eps k\right] \leq \frac Tn e^{-\Delta/C}.
    \end{equation*}

    Furthermore, for any $S \subseteq N(v)$, $c_1 \ne c_2 \in [k]$, and non-negative integers $i_1,i_2$,
\begin{equation*}
   \mb{P}\left[ \exists t \in [T_0,T_0+T] : \left|  | S_{c_1,i_1}(X_t) \cap S_{c_2,i_2}(X_t)| - \sum_{w \in S} \frac{e^{-2d(w)/k}}{i_1 ! i_2 !} \paren{\frac{d(w)}{k}}^{i_1+i_2} \right| > \eps\Delta \right] \leq \frac Tn e^{-\Delta/C}.
\end{equation*}
Moreover, for any $c \in [k]$,
\begin{equation*}
\mb{P}[\exists t \in [T_0, T_0 + T]: |X_t^{-1}(c) \cap B_2(v)| > 400 \Delta] \leq \frac{T}{n}\exp(-\Delta/C).
\end{equation*}
\end{theorem}

\medskip

The third conclusion of \cref{thm:local-uniformity-discrete} follows directly from  the Metropolis versions of \cite[Lemma~31]{hayes2013local}; this can be proved in exactly the same way. Therefore, we will focus on proving the first two conclusions. We will deduce these from a similar statement about the continuous-time Metropolis  dynamics, which is easier to analyze. 

\begin{definition}
Let $G = (V,E)$ be a graph, let $k$ be the number of colors, and let $X_0 \in \wh{\Omega}$. In the continuous-time Metropolis dynamics, each vertex $u \in V$ has an independent Poisson clock of rate~$1$. When the clock of $u$ rings at time $t$, we choose a candidate color $c \in [k]$ uniformly, and we set
\[X_t(u)=\begin{cases}c,&\text{if }c\in A(X_{t^-},u),\\ X_{t^-}(u),&\text{otherwise.}\end{cases}\]
and all other coordinates are unchanged. 
\end{definition}

\begin{remark}
    It will be convenient to use the following equivalent successful-update description: 
    vertex $u$ updates successfully at rate $|A(X_t,u)|/k$, and conditional on such a successful update, the new color is uniform on $A(X_{t}, u)$.
\end{remark}

\begin{observation}
\label{obs:bound-update-rates}
Let $\rho_{\delta} := \frac{\delta}{1+\delta}.$
For every $X \in \wh{\Omega}$ and every $u \in V$,
\[|A(X,u)| \geq k-\Delta \geq \rho_{\delta} k.\]
Consequently, every successful-update clock has instantaneous rate in $\left[\rho_{\delta}, 1\right].$
\end{observation}

We will prove the following bounded-interval, continuous-time analogue of \cref{thm:local-uniformity-discrete} and then transfer the result to discrete time. The statement of the continuous-time result requires the following notion of $C$-above-suspicion from \cite{hayes2013local}.

\begin{definition}
\label{def:above-suspicion}
    Let $G = (V, E)$ be a graph of maximum degree $\Delta$, and let $C > 0$. For any vertex $w \in V$ and positive integer $R$, let $N^R(w)$ denote all vertices, other than $w$, which are within distance at most $R$ of $w$, i.e.~$N^R(w) = B_R(w) \setminus \{w\}$.
    
    Let $f: V \to [k]$ be a labeling, let $c \in [k]$, and let $v \in V$. We say $f$ is {$C$-light} for color $c$ at $v$ if
    \[
    |f^{-1}(c) \cap N^2(v)| \le C \Delta \quad \text{and} \quad |f^{-1}(c) \cap N(v)| \le C\Delta/\log \Delta.
    \]
  Then we say $f$ is {$C$-above-suspicion} for radius $R$ at $v$ if all $w\in B_R(v)$ are $C$-light for every $c\in [k]$.
\end{definition}

\begin{theorem}\label{thm:local-uniformity-continuous}
      Given $\delta,\eps > 0$, there exist constants $\Delta_0 = \Delta_0(\delta,\eps)$, $C = C(\delta,\eps)$, and  $R = R(\eps, \delta)$ such that the following holds. Let $G = (V,E)$ be a graph of maximum degree $\Delta > \Delta_0$ and girth at least $7$, and let $k \geq (1+\delta)\Delta$.  Let $\{X_t\}_{t \geq 0}$ be the continuous-time Metropolis dynamics on $\wh{\Omega}$. Fix a vertex $v \in V$ and assume that $X_0$ is $400$-above-suspicion for radius $R$ at $v$. 
    Then,
\begin{equation}\label{eq:continuous-number-colors}
        \mb{P}\left[\exists T \in [C, C+1]: \big||A(X_T,v)| - k e^{-d(v)/k} \big| > \eps k\right] \leq e^{-\Delta/C},
    \end{equation}
and for any $S \subseteq N(v)$, $c_1 \ne c_2 \in [k]$, and non-negative integers $i_1,i_2$,
    \begin{equation}\label{eq:continuous-blockers}
        \mb{P}\left[ \exists T \in [C, C+1] : \left| |S_{c_1,i_1}(X_T) \cap S_{c_2,i_2}(X_T)| - \sum_{w \in S} \frac{e^{-2d(w)/k}}{i_1 ! i_2 !} \paren{\frac{d(w)}{k}}^{i_1+i_2} \right| > \eps\Delta \right] \leq  e^{-\Delta/C}.
    \end{equation}
\end{theorem}

\begin{remark}
    We will assume $\eps \ll \delta$; this loses no generality since enlarging $\eps$ only weakens the conclusion. 
\end{remark}

The deduction of \cref{thm:local-uniformity-discrete} from \cref{thm:local-uniformity-continuous} follows as in \cite{hayes2013local}; we omit the straightforward details. The remainder of this section is devoted to the proof of \cref{thm:local-uniformity-continuous}.

\medskip

\begin{remark}[overview of the proof]\label{par:local-uniformity-overview} For the reader's convenience, we outline the structure of the proof. 

\begin{enumerate}
    \item We define a modification $X^*$ (\cref{def:Xstar}) that will prove easier to analyze, and show that $X^*_t \approx X_t$ (\cref{prop:comparison}).
    \item We prove a tight lower bound on $|A(X_t^*,v)|$ and use the above comparison to transfer this to a tight lower bound on $|A(X_t,v)|$ (\cref{cor:universal-lower-bound}). 
    \item 
    \label{detailed-new-local-uniformity} We define a quantity $P(X,v,c)$ (\cref{def:bias-field}) which is essentially the ``expected size'' of $X^{-1}(c) \cap N(v)$, and, using a recurrence for $P(X,v,c)$, show concentration $P(X,v,c) \approx d(v)/k$ (\cref{prop:concentration-of-P}).  This involves an auxiliary chain $Z^*$ (\cref{def:Zstar}) which is not necessary in the heat-bath case of \cite{hayes2013local} but crucial for the Metropolis dynamics. This part contains the key innovations of this section (see \cref{rem:heatbath-vs-metropolis}). 
    \item Finally, we use concentration of $P(X,v,c)$ to derive \cref{thm:local-uniformity-continuous}.
\end{enumerate}

\end{remark}

\medskip

Before proceeding further, we gather a few concentration results for later use.  

\begin{lemma}[see {\cite[Lemma 23]{hayes2013local}}] 
\label{lem:dyer-concentration}
  Let $k,d\ge 1$, and let $\xi_1,\dots,\xi_d$ be independent random variables taking values in $[k]$. Let
\[
p:=\max_{1\le j\le d}\max_{a\in[k]}\Pbb[\xi_j=a],
\qquad
A:=[k]\setminus\{\xi_1,\dots,\xi_d\}.
\]
Then
\[
\Ebb\brac{\abs A}\ge k\paren{\frac{1-p}{e}}^{d/k},
\]
and, for every $a\ge 0$,
\[
\Pbb\brac{\abs{\abs A-\Ebb\brac{\abs A}}\ge a}\le 2e^{-a^2/(2k)}.
\]

\end{lemma}

\begin{proof}
    This result was originally proven by Dyer and Frieze~\cite{dyer2003randomly} in the iid case.   Hayes~\cite{hayes2013local} extended it to the independent setting, but in the second inequality, only stated the result for the lower tail. His proof of this result, which uses a Chernoff bound, immediately implies the stated bound on the upper tail as well.
\end{proof}

We shall also use the elementary Poisson tail bound
\begin{equation}\label{eq:poisson-tail}
\Pbb\brac{\Po(\mu)\ge b}\le \paren{\frac{e\mu}{b}}^b
\qquad\text{for }b>\mu.
\end{equation}

Finally, we will need the following observation from \cite{hayes2013local}.

\begin{observation}[{\cite[Observation 11]{hayes2013local}}]\label{obs:concentration}
    Let $X,Y$ be two non-negative random variables with $\min Y > 0$.  Let $0 \leq \theta < \min Y/2$, and suppose $p \geq \mb{P}[|Y - \mb{E}Y| \geq \theta]$.  Then with probability at least $1-p$,
    \[ \mb{E}\left[ \frac XY \right] \in \frac{\mb{E} X}{Y \pm 2\theta} \pm \frac{p \max X}{\min Y}. \]
\end{observation}

\subsection{The auxiliary processes \texorpdfstring{$X^*$ and $Z^*$}{X* and Z*}}
\label{subsec:modified-process}

Fix for the moment a vertex $v \in V$.

\begin{definition}
\label{defn:directed-neighborhoods}
Given a graph $G$ of girth at least $7$ and a vertex $v$, define $G^* = G_{\mr{in}}(v,3)$ to be the directed graph constructed as follows:
\begin{itemize}
    \item The vertex set is $V(G)$.
    \item For any edge $\{x, y\}$ in the original graph $G$:
    \begin{itemize}
        \item If both $x, y \in B_3(v)$ and $d(x, v) > d(y, v)$, we include the directed edge $(x, y)$. (This directs edges towards the center $v$.)
        \item Otherwise, we include the pair of directed edges $(x,y)$ and $(y,x)$. 
    \end{itemize}
\end{itemize}
Equivalently, $G_{\mr{in}}(v,3)$ is obtained by replacing each edge in $G$ by the corresponding pair of directed edges, and then deleting all edges in directed paths of length $3$ starting at $v$. 
\end{definition}

In this directed context, neighbor means in-neighbor: for a vertex $w \in V$, we define
\[N_{G^*}(w) = \{u \in V \mid (u,w) \in G^*\}.\]
Accordingly, the available colors for a vertex $w$ in a labeling $X$ are determined only by its in-neighbors:

\[A_{G^*}(X,w) = [k]\setminus X(N_{G^*}(w)).\]

\begin{definition}[the natural directed chain $\Yst$]\label{def:Ystar}
Let $G^*:=\Gin(v,3)$. The chain $\Yst=(\Yst_t)_{t\ge 0}$ is the continuous-time Metropolis dynamics on $G^*$ started from the original initial labeling $X_0 \in \wh{\Omega}$ on all vertices.
\end{definition}

\begin{definition}[the recursive directed chain $X^*$] 
\label{def:Xstar}
Let $G^* = G_{\mr{in}}(v,3)$ and let $X_0 \in \wh{\Omega}$ be a labeling as before. 
Let $X^* = \{X^*_t\}_{t \in \mb{R}}$ denote the continuous-time process on $G^* = G_{\mr{in}}(v,3)$ obtained using the following construction. 

\begin{enumerate}
    \item Boundary phase: Note that the graph induced by $G^*$ on $B_2(v)^c = V\setminus B_2(v)$ is an undirected graph. 
    Starting from $X_0^*|_{B_2(v)^c} := X_0|_{B_2(v)^c}$, run the standard continuous-time Metropolis Glauber dynamics forward for $t\geq 0$. Independently run a second copy forward from the same labeling and use its trajectory, with time reversed, to define $X_t^*$ for $t<0$; by reversibility, this is the same dynamics run backward from time $0$.

    \item  Inward propagation to $S_2(v)$: Having fixed the trajectories on $S_3(v)$, define the trajectories on $S_2(v)$ as follows. For each $z\in S_2(v)$, conditionally on the already constructed trajectories of $N(z)\subseteq S_3(v)$, let $\mc T_z := \mc T_z^{X^*}$ be an inhomogeneous Poisson process on $\mb R$ with instantaneous rate $|A(\Xst_t,z)|/k$. At each time of $\mc T_z$, choose a color uniformly from $A(\Xst_t,z)$. Since the rate is bounded below by $\rho_\delta$, there are almost surely infinitely many points of $\mc T_z$ in $(-\infty,t)$ for every $t$. Define $\Xst_t(z)$ to be the color chosen at the last point of $\mc T_z$ before time $t$.
\item Inward propagation to $S_1(v)$: Having fixed the trajectories on $S_2(v)$, define the trajectories on $S_1(v) = N(v)$ in the same way, using the already constructed trajectories of $S_2(v)$.
\item Inward propagation to $v$: Finally, define the trajectory at $v$ in the same way, using the already constructed trajectories of $S_1(v)$.
\end{enumerate}
In particular, $\Xst_0$ agrees with $X_0$ outside $B_2(v)$, but inside the ball its time-$0$ values are generated by the two-sided recursive construction and need not equal $X_0$.
\end{definition}

\begin{remark}
The role of \(\Yst\) is that it starts from the actual initial labeling $X_0$ and is therefore the natural object to compare to the original chain $(X_t)_{t\geq 0}$. The role of \(\Xst\) is different: because it is generated recursively from the outer trajectories, every term in the later definition of the bias field (\cref{def:bias-field}) comes from a genuine refresh, and no correction term is needed for a vertex in the ball that has not refreshed since time $0$ and is still equal to its initial value. Moreover, the construction of $X^*$ ensures that for any vertex $u$ inside the ball $B_2(v)$, its state depends only on the history of vertices strictly farther from $v$ than $u$. This conditional independence will be key in the analysis. 
\end{remark}

\begin{definition}[the auxiliary process $\Zst$]\label{def:Zstar}
Fix $v \in V$ and the recursively defined process $\Xst$ as in \cref{def:Xstar}. Let
\[
\mc F_v:=\sigma\paren{\Xst_t(x): x\notin B_2(v),\ t\in\mb R}.
\]
Define a process $\Zst=(\Zst_t)_{t\in\mb R}$ by the following rules.
\begin{enumerate}
\item Outside $B_1(v)$, set $\Zst_t(x):=\Xst_t(x)$ for every $t\in\mb R$.
\item For each $w\in N(v)$, conditionally on $\mc F_v$, let $\mc T_w^{\Zst}$ be an inhomogeneous Poisson process on $\mb R$ of instantaneous rate
\[
\lambda_w(t):=\frac{1}{k}\mb E\brac{|A(\Xst_t,w)|\mid \mc F_v}.
\]
At each point of $\mc T_w^{\Zst}$, choose a color uniformly from $A(\Xst_t,w)=A(\Zst_t,w)$. Since $\lambda_w(t)\ge \rho_\delta$, there are almost surely infinitely many such points to the left of every $t$, and we define $\Zst_t(w)$ to be the color chosen at the last point before time $t$.
\item Fix the root color deterministically, say $\Zst_t(v)=1$ for all $t$.
\end{enumerate}
\end{definition}

\begin{remark}
\label{rem:heatbath-vs-metropolis}
    In the heat-bath dynamics, one can condition on the update times in $B_1(v)$ without leaking information about the colors in $S_2(v)$.  For the Metropolis dynamics, this is not true as the update rate of $w \in B_1(v)$ is proportional to $|A(X_t,w)|$, and so revealing update times in $S_1(v)$ leaks information about colors in $S_2(v)$.  The purpose of the chain $Z^*$ is to fully decouple the update times in $B_1(v)$ with the colors of $S_2(v)$: compare step (2) in the definition of $Z^*$, where the rate is a conditional expectation over $\mc{F}_v$, versus the rates in the definition of $X^*$, which depend directly on the random variable $|A(X_t^*,\cdot)|$. 
\end{remark}

We now show that the process $(X^*_t)$ is ``close'' to the continuous-time Metropolis  dynamics $(X_t)$ for the quantities we care about. To state this precisely, we need some notation.

\begin{definition}
    For directed graphs $G, G^*$ on the same vertex set, we use $G \oplus G^*$ to denote the symmetric difference of their edgesets. In particular, when $G^* = G_{\mr{in}}(v,3)$, then $G\oplus G^*$ is a directed tree of depth $3$ rooted at $v$ and oriented towards its leaves. 
\end{definition}

\begin{definition}
    For labelings $X, X' : V \to [k]$, we use $X\oplus X'$ to denote the disagreement set $\{u\in V: X(u)\neq X'(u)\}$. 
\end{definition}

By Poisson thinning, we can equivalently view the continuous-time Metropolis dynamics on $k$-colorings as follows: assign to each vertex-color pair $(v,c)$, an independent Poisson clock of rate $1/k$. When this clock rings, attempt to update the current labeling by assigning the color $c$ to $v$, accepting this update if and only if it is valid. A similar equivalent reformulation also holds for both modified process. By using the same Poisson clocks for all the vertex-color pairs $(v,c)$, we therefore obtain a natural coupling $(X_t, X_t^*, Y_t^*)$, which is what will be used in the statement below. This statement is an analogue of \cite[Theorems 32, 33] {hayes2013local}.

\begin{proposition}
\label{prop:comparison}
For every $\delta, \eps > 0$, there exist constants $C = C(\eps, \delta)$, $R = R(\eps,\delta)$ and $\Delta_0 = \Delta_0(\eps, \delta)$ such that the following holds. Suppose $G = (V,E)$ has girth $g\geq 7$ and maximum degree $\Delta \geq \Delta_0$. Let $k \geq (1+\delta)\Delta$. Let $v \in V$ and $G^* = G_{\mr{in}}(v,3)$. Let $(X_t, X_t^*)_{t\geq 0}$ denote the above coupling with initial configuration $X_0$ which is $400$-above-suspicion at $v$ for $R$.
Then, for any $T \in [C, C+1]$,
    \[ \mb{P}[
    \forall u \in V, |(X_T^* \oplus X_T) \cap N(u)| \leq \eps\Delta] \geq 1-\exp(-\Delta/C) \]
and for every color $c \in [k]$,
   \[ \mb{P}\bigg[ 
    \forall u \in V, 
    \#\{ z \in B_2(u) \mid z \in X_T \oplus X_T^*, c \in \{X_T(z),X_T^*(z)\} \} \leq \eps\Delta \bigg] \geq 1-\exp(-\Delta/C). \]
\end{proposition}

\begin{proof}

By the triangle inequality, and by adjusting constants, it suffices to show that the two conclusions of \cref{prop:comparison} hold for $X_T \oplus Y_T^*$ and $X_T^* \oplus Y_T^*$. The tension in the proof comes from the processes $X_t^*$ and $Y_t^*$ needing a sufficient amount of time to couple, and the processes $Y_t^*$ and $X_t$ beginning at $Y_0^* = X_0$ slowly diverging from each other.

The corresponding statements for $X_T \oplus Y_T^*$ follow the proofs of \cite[Theorem 32 and Theorem 33]{hayes2013local} without substantial modification, using \cref{obs:bound-update-rates}.
The only change this causes in Hayes's result is that certain constants need to be adjusted by factors of $\delta$, which is immaterial, since all of our implicit constants are allowed to depend on $\delta$. The main work is in proving the corresponding statement for $X_T^* \oplus Y_T^*$. 

Recall that
$\rho = \rho_{\delta}:=\frac{\delta}{1+\delta}$, 
so every successful-update rate in either chain is at least $\rho$, since $|A(\cdot,u)|\geq k-\Delta\geq \rho k$ for every vertex $u$.  
Choose $\eta:=\eps/100$. Next choose $a=a(\eps,\delta)$ so that $e^{-\rho a/2}\leq \eta$, and then choose $C=C(\eps,\delta)$ so large that
\[
e^{-\rho(C-a)}\leq \eta,\qquad \frac{e^{-\rho C}}{\rho(1+\delta)}\leq \eta,\qquad 400e^{-\rho C}\leq \eta.
\]
Fix $T\in [C,C+1]$.

For $z\in S_2(v)$, the available-color sets in the two chains are always equal, because both depend only on the colors in $S_3(v)$, where $X^*$ and $Y^*$ agree identically. Hence the successful-update processes at such $z$ are identical in the two chains, and once $z$ has one successful update after time $0$, the two chains agree there forever.

Now fix $w\in N(v)$ and write $\mc C(w):=N(w)\cap S_2(v)$. For $z\in \mc C(w)$, let $I_z$ be the indicator that $z$ has no successful update in $[0,T-a]$, and set $M_w:=\sum_{z\in \mc C(w)} I_z$. Conditional on the outside trajectory, the variables $\{I_z:z\in \mc C(w)\}$ are independent, and each satisfies $\mb P[I_z=1]\leq e^{-\rho(T-a)}\leq \eta$. Since $|\mc C(w)|\leq \Delta$, a Chernoff bound and then a union bound over $w\in N(v)$ give an event $\mc E_2$ with probability at least $1-\exp(-\Delta/C)$ on which
\[
M_w\leq 2\eta\Delta\qquad\text{for every }w\in N(v).
\]

Assume $\mc E_2$ holds. Then for every $w\in N(v)$ and every $t\in [T-a,T]$, the sets $A(X_t^*,w)$ and $A(Y_t^*,w)$ differ in at most $2\eta\Delta$ colors, because only those children of $w$ that have not refreshed by time $T-a$ can still be discrepant. Since both sets have size at least $\rho k$, their intersection has size at least $\rho k/2$ for $\Delta$ large enough. Therefore, once we condition on the histories up to time $T-a$, the probability that a given $w\in N(v)$ is still discrepant at time $T$ is at most $e^{-\rho a/2} + (2\eta\Delta)/(\rho k/2) \leq \eta^{1/2}$, and these events are independent over distinct $w$. Another Chernoff bound gives an event $\mc E_1$ such that
\[
\mb P[\mc E_1^c\mid \mc E_2]\leq \exp(-\Delta/C),
\]
and on $\mc E_1$,
\[
\#\{w\in N(v):X_T^*(w)\neq Y_T^*(w)\}\leq 2\eta\Delta.
\]
Hence
\[
\mb P[\mc E_1\cap \mc E_2]\geq 1-\exp(-\Delta/C).
\]

We now prove the first conclusion. On $\mc E_1\cap \mc E_2$, every neighborhood contains at most $\eps\Delta$ discrepant vertices. Indeed:
\begin{itemize}
\item if $u\notin B_3(v)$, then $X_T^*=Y_T^*$ on $N(u)$;
\item if $u\in S_3(v)$ or $u\in S_2(v)$, then by girth at least $7$, $u$ has at most one neighbor in $B_2(v)$, so $|(X_T^*\oplus Y_T^*)\cap N(u)|\leq 1\leq \eps\Delta$;
\item if $u=v$, then all discrepancies in $N(v)$ are counted by $\mc E_1$, so the bound is $2\eta\Delta$;
\item if $u=w\in N(v)$, then the only possible discrepant neighbors are $v$ and those children in $\mc C(w)$ that failed to refresh by time $T-a$, so
\[
|(X_T^*\oplus Y_T^*)\cap N(w)|\leq 1+M_w\leq 1+2\eta\Delta\leq \eps\Delta
\]
for $\Delta$ large enough.
\end{itemize}
This proves the neighborhood-discrepancy estimate.

For the color-specific estimate, fix $c\in [k]$. Since $X_T^*=Y_T^*$ outside $B_2(v)$, it is enough to show that
\[
\#\{z\in B_2(v): z\in X_T^*\oplus Y_T^*,\ c\in\{X_T^*(z),Y_T^*(z)\}\}\leq \eps\Delta
\]
with probability at least $1-\exp(-\Delta/C)$.

The contribution from $N(v)$ is already at most $2\eta\Delta$ on $\mc E_1$, so only vertices in $S_2(v)$ need further attention. Split them into two classes:
\[
A_c:=\#\{z\in S_2(v): X_T^*(z)=c\neq Y_T^*(z)\},\qquad
B_c:=\#\{z\in S_2(v): Y_T^*(z)=c\neq X_T^*(z)\}.
\]

For $A_c$, note that if $X_T^*(z)=c\neq Y_T^*(z)$ then $z$ has had no successful update in $[0,T]$ (in either chain, since those successful-update processes are identical at vertices of $S_2(v)$). Conditional on the outside trajectory, this has probability at most $e^{-\rho T}\leq e^{-\rho C}$; given that event, the value of $X_T^*(z)$ is the color used at the last successful update of the two-sided chain $X^*$ before time $0$, so the conditional probability that it equals $c$ is at most $1/(\rho k)$. Thus each $z\in S_2(v)$ contributes to $A_c$ with conditional probability at most $e^{-\rho C}/(\rho k)$. Since $|S_2(v)|\leq \Delta^2$ and $k\geq (1+\delta)\Delta$, the conditional expectation of $A_c$ is at most $\eta\Delta$, so a Chernoff bound gives $\mb P[A_c>2\eta\Delta]\leq \exp(-\Delta/C)$.

For $B_c$, if $Y_T^*(z)=c\neq X_T^*(z)$ then again $z$ has had no successful update in $[0,T]$, so $Y_T^*(z)=Y_0^*(z)=X_0(z)=c$. Hence only vertices of $B_2(v)$ that already have color $c$ at time $0$ can contribute. Since $X_0$ is $400$-above-suspicion at $v$, there are at most $400\Delta$ such vertices. Each survives with no successful update until time $T$ with probability at most $e^{-\rho C}$, so $\mb E[B_c]\leq 400e^{-\rho C}\Delta\leq \eta\Delta$, and another Chernoff bound gives $\mb P[B_c>2\eta\Delta]\leq \exp(-\Delta/C)$.

Combining the contributions from $N(v)$, $A_c$, and $B_c$, we get the required bound, since $2\eta\Delta+2\eta\Delta+2\eta\Delta<\eps\Delta$.
\end{proof}

\subsection{Lower bound on the number of available colors}
Recall that \eqref{eq:continuous-number-colors} asserts a two-sided bound on the number of available colors $A(X_T, v)$. As a first step, we will establish one side of the bound. We will first prove the result for $X^*$ at a fixed time and then transfer it to $X$ over an entire bounded-length interval. 

\begin{lemma}
\label{lem:available-colors-fixed-time}
     For every $\delta, \eps > 0$, there exist constants $\Delta_0 = \Delta_0(\delta, \eps)$ and $C = C(\delta, \eps)$ such that the following holds. Let $v \in V$ and let $T \in [C, C+1]$. 
    Then,
    \[ \mb{P}[|A(X^*_T,v)| \leq (1-\eps) k e^{-d(v)/k}] \leq \exp(-\Delta/C).\]
\end{lemma}

\begin{proof}
    Condition on
\[
\mc F_v:=\sigma\paren{\Xst_t(x):x\notin B_2(v),\ t\in\mb R}.
\]
Given $\mc F_v$, the colors $\{\Xst_T(w):w\in N(v)\}$ are independent, because in the directed graph $\Gin(v,3)$ the branches below distinct neighbors of $v$ are disjoint.

Fix $w\in N(v)$ and $c\in[k]$. Let $\tau_w(T)$ be the last successful refresh time of $w$ before time $T$ in $\Xst$; this time exists almost surely. Conditional on $\mc F_v$ and on $\tau_w(T)$,
\[
\mb P[\Xst_T(w)=c\mid \mc F_v,\tau_w(T)]
=
\frac{\mathbf 1\{c\in A(\Xst_{\tau_w(T)^-},w)\}}{|A(\Xst_{\tau_w(T)^-},w)|}
\le \frac1{\rho k}
\]
by \cref{obs:bound-update-rates}. Hence
\[
\sup_{c\in[k]}\mb P[\Xst_T(w)=c\mid \mc F_v]\le \frac1{\rho k}.
\]
Applying \cref{lem:dyer-concentration} conditionally on $\mc F_v$ gives
\[
\mb E\brac{|A(\Xst_T,v)|\mid \mc F_v}\ge k\paren{\frac{1-1/(\rho k)}{e}}^{d(v)/k}
= ke^{-d(v)/k}\paren{1-O\paren{\frac1\Delta}}.
\]
Since $d(v)/k\le 1/(1+\delta)$, the error term $O(1/\Delta)$ is uniform. For $\Delta$ large enough,
\[
\mb E\brac{|A(\Xst_T,v)|\mid \mc F_v}\ge (1-\eps/2)ke^{-d(v)/k}.
\]
A second application of \cref{lem:dyer-concentration} yields
\[
\mb P\brac{|A(\Xst_T,v)|\le (1-\eps)ke^{-d(v)/k}\mid \mc F_v}\le e^{-\Delta/C},
\]
and averaging over $\mc F_v$ proves the statement for $X^*_T$. 
\end{proof}

\begin{lemma}\label{cor:universal-lower-bound}
    For every $\delta, \eps > 0$, there exist constants $\Delta_0 = \Delta_0(\delta, \eps)$, $C = C(\delta, \eps)$, and $R = R(\eps,\delta)$ such that the following holds. Let $v \in V$ and suppose that $X_0$ is $400$-above-suspicion for radius $R$ at $v$. Then, 
    \[ \mb{P}[\exists t \in [C, C+1]: |A(X_t,v)| \leq (1-2\eps) k e^{-d(v)/k}] \leq \exp(-\Delta/C).\]
    Moreover, the same bound also holds with $X$ replaced by $X^*$ on $G_{\mathrm{in}}(v,3)$ (in fact, without any assumptions on $X_0$). 
\end{lemma}

\begin{proof}

We first prove the statement for $\Xst$. Choose a mesh size $h:=\frac{\eps}{100\Delta}.$
The interval $[C, C+1]$ contains $O(\Delta/\eps)$ mesh points. By \cref{lem:available-colors-fixed-time}, the probability that the lower bound fails at one of those mesh points is at most $e^{-\Delta/C}$. Between two consecutive mesh points, the quantity $|A(\Xst_t,v)|$ can change only when a vertex in $B_2(v)$ refreshes successfully, and each such refresh changes $|A(\Xst_t,v)|$ by at most $1$. The total successful-refresh rate in $B_2(v)$ is at most $2\Delta^2$, so the number of such refreshes in an interval of length $h$ is Poisson with mean $O(\eps\Delta)$. By \eqref{eq:poisson-tail} and a union bound over the $O(\Delta/\eps)$ mesh intervals, with probability at least $1-e^{-\Delta/C}$ each mesh interval contains at most $\eps\Delta/4$ such refreshes. On this event, if the lower bound holds at the mesh points with parameter $\eps$, then it holds on the whole interval with parameter $2\eps$.

For the original chain $X$, the statement now follows from \cref{prop:comparison} and a similar mesh argument.
\end{proof}

\subsection{The bias field and its properties}

This is the most technically difficult part in the proof of \cref{thm:local-uniformity-continuous} and also the one which differs the most from \cite{hayes2013local} due to the different bias field $P$ and the auxiliary process $Z^*$. 

The following notion of ``bias field' is similar, but not identical to the function $P(X,v,c)$ in \cite[Section~2.4]{hayes2013local}, as we explain below. 

\begin{definition}[Bias field]
\label{def:bias-field}
Let $U = (U_t)$ be either the original metropolis chain $X$, the recursive directed chain $X^*$, or the auxiliary process $Z^*$. Let $v \in V$, let $T > 0$, and let $c \in [k]$. For each $w \in N(v)$, let $\tau_w^{U}(T)$ be the last successful refresh time of $w$ before time $T$; if no such refresh has occurred, set $\tau_w^U(T) = -\infty$. Define
\[P(U_T, v, c) =  \sum_{w \in N(v)} \mathbf{1}\{\tau_w^U(T) > -\infty\}\frac{\mathbf{1}\{c \in A_v(U_{\tau_w^U(T)^-},w)\}}{|A_v(U_{\tau_w^U(T)^-},w)|}.
\]
\end{definition}

\begin{remark}
    This is slightly different from Hayes' function $P(X_T,v,c)$ \cite[Section~2.4]{hayes2013local}, in which all colors on the right hand side are only accessed at time $T$ with no mention of $\tau_w(T)$.  We believe that the version above is what is actually needed for the later arguments in \cite{hayes2013local} as well.
\end{remark}

\begin{remark}
    In all our applications, $\tau_w^U(T) > -\infty$ almost surely, so we will omit this factor for notational convenience. Additionally, almost surely, $A_v(U_{\tau_w^U(T)^-},w) = A_v(U_{\tau_w^U(T)},w)$, so we will not distinguish between $\tau^-$ and $\tau$. Finally, note that for the processes $X^*$ and $Z^*$ on the directed graph $G^* = G_{\mathrm{in}}(v,3)$, for all $w \in N(v)$,
    \[A(X_t^*, w) = A_{G^*}(X_t^*, w) = A_{G^*}(Z_t^*, w) = A_v(Z_t^*, w),\]
    where we have used that $X_t^* = Z_t^*$ outside $B_1(v)$.
\end{remark}

The next lemma provides a comparison between the bias fields for $X^*$ and $Z^*$.

\begin{lemma}
\label{lem:Xstar-vs-Zstar}
For every $\delta,\eps>0$ there exist constants $\Delta_0=\Delta_0(\delta,\eps)$ and $C=C(\delta,\eps)$ such that the following holds. Let $v\in V$ and $T\in[C, C+1]$. Then one can couple $\Xst$ and $\Zst$ so that
\[
\mb P\brac{\bigl|\set{w\in N(v):\ \tau_w^{\Xst}(T)\ne \tau_w^{\Zst}(T)}\bigr|>\eps\Delta}\le e^{-\Delta/C}.
\]
Consequently,
\[
\mb P\brac{\exists c\in[k]:\ \bigl|P(\Xst_T,v,c)-P(\Zst_T,v,c)\bigr|>\eps}\le e^{-\Delta/C}.
\]
\end{lemma}

\begin{proof}
Condition on $\mc F_v$. For each $w\in N(v)$, couple the refreshes of $\Xst$ and $\Zst$ by a common Poisson point process on $[0,T]\times[0,1]$ of intensity $1$: at a point $(t,s)$, the chain $\Xst$ refreshes $w$ if $s\le |A(\Xst_t,w)|/k$, the chain $\Zst$ refreshes $w$ if $s\le \lambda_w(t)$, and whenever both refresh we use the same new color.

Let $B_w$ be the event that either before time $T$ there is a point $(t,s)$ at which exactly one of the two chains refreshes $w$, or there is no point in $[0,T]$ at which both chains refresh $w$. If $B_w$ does not occur, then the chains have at least one common refresh in $[0,T]$ and all their refreshes there are common, so $\tau_w^{\Xst}(T)=\tau_w^{\Zst}(T)$.

Conditional on $\mc F_v$, the number of mismatched refreshes is dominated by a Poisson random variable with mean
\[
\mu_w:=\frac1k\int_0^T \bigl||A(\Xst_t,w)|-\mb E[|A(\Xst_t,w)|\mid \mc F_v]\bigr|\,dt.
\]
Hence the conditional probability of at least one mismatched refresh is at most $\mb E[\mu_w\mid \mc F_v]$.
For each fixed $t$, conditional on $\mc F_v$ the variable $|A(\Xst_t,w)|$ is a Lipschitz function of the independent colors of the vertices in $N(w)\setminus\{v\}$, and \cref{lem:dyer-concentration} implies
\[
\mathrm{Var}\bigl(|A(\Xst_t,w)|\mid \mc F_v\bigr)=O(k).
\]
Therefore
\[
\mb E\brac{\bigl||A(\Xst_t,w)|-\mb E[|A(\Xst_t,w)|\mid \mc F_v]\bigr|\mid \mc F_v}=O(\sqrt{k}),
\]
and integrating over $t\in[0,T]$ bounds the probability of a mismatched refresh by $C/\sqrt\Delta$. Moreover, the instantaneous rate of common refreshes is
\[
\min\bigl\{|A(\Xst_t,w)|/k,\lambda_w(t)\bigr\}\geq \rho_\delta,
\]
so the probability of no common refresh in $[0,T]$ is at most $e^{-\rho_\delta T}$. Thus
\[
\mb P[B_w\mid \mc F_v]\le \frac{C}{\sqrt{\Delta}}+e^{-\rho_\delta T}.
\]
For distinct $w\in N(v)$, the events $B_w$ are conditionally independent given $\mc F_v$, because the branches below the different neighbors of $v$ are disjoint. Hence the total number
\[
B:=\sum_{w\in N(v)} \mathbf 1_{B_w}
\]
is conditionally stochastically dominated by a binomial random variable with parameters $\Delta$ and $C/\sqrt\Delta+e^{-\rho_\delta T}$. Increase $C$ so that $e^{-\rho_\delta C}\leq \eps/4$, and then increase $\Delta_0$ so that $C/\sqrt\Delta\leq \eps/4$. Its mean is then at most $\eps\Delta/2$, so a Chernoff bound gives
\[
\mb P[B>\eps\Delta\mid \mc F_v]\le e^{-\Delta/C},
\]
and the statement follows from the law of total probability. 
On the complement of this event, the definitions of $P(\Xst_T,v,c)$ and $P(\Zst_T,v,c)$ differ in at most $\eps\Delta$ summands, and by \cref{cor:universal-lower-bound} each summand is bounded by $C/\Delta$. Thus
\[
\sup_{c\in[k]}\bigl|P(\Xst_T,v,c)-P(\Zst_T,v,c)\bigr|\le \eps.
\]
This proves the claim.
\end{proof}

Next, we compare the bias fields for $X$ and $X^*$. This is the analogue of \cite[Corollary 30]{hayes2013local} for our version of bias field.

\begin{lemma}\label{lem:comparison-of-P}
For every $\delta,\eps>0$ there exist constants $\Delta_0=\Delta_0(\delta,\eps)$ and $C=C(\delta,\eps)$ such that the following holds. Let $v\in V$ and $T\in[C, C+1]$. Then, 
\[
\mb P\brac{\exists c\in[k]:\ \bigl|P(\Xst_T,v,c)-P(X_T,v,c)\bigr|>\eps}\le e^{-\Delta/C}. \qedhere
\]
\end{lemma}

\begin{proof}
For lightness of notation, let $\tau_y := \tau_y^{X}(T)$ and $\tau_y^* := \tau_y^{X^*}(T)$.
Let $D = \{ y \in N(v) \mid \tau_y \ne \tau_y^* \text{ or }\tau_y < 0 \}$. By the triangle inequality,
\begin{align*}
    |P(X_T^*,v,c) - P(X_T,v,c)| & \leq \sum_{y \in N(v)} \left| \frac{\mbf1\{c \in A(X_{\tau_y^*}^*,y)\}}{|A(X_{\tau_y^*}^*,y)|} - \frac{\mbf1\{c \in A_v(X_{\tau_y},y)\}}{|A_v(X_{\tau_y},y)|} \right| \\
    & = \paren{ \sum_{y \in N(v) \cap D} + \sum_{y \in N(v) \setminus D}} \left| \frac{\mbf1\{c \in A(X_{\tau_y^*}^*,y)\}}{|A(X_{\tau_y^*}^*,y)|} - \frac{\mbf1\{c \in A_v(X_{\tau_y},y)\}}{|A_v(X_{\tau_y},y)|} \right|  \\
\end{align*}

The same disagreement-history estimate used in the proof of \cref{prop:comparison}, now applied to the coupled successful-refresh clocks and counting a vertex with no refresh in $[0,T]$ as bad, gives $\mb P[|D|>\eps\Delta]\leq e^{-\Delta/C}$.  On the complementary event, we bound the first sum by
\[ \sum_{y \in N(v) \cap D} \left| \frac{\mbf1\{c \in A(X_{\tau_y^*}^*,y)\}}{|A(X_{\tau_y^*}^*,y)|} - \frac{\mbf1\{c \in A_v(X_{\tau_y},y)\}}{|A_v(X_{\tau_y},y)|} \right| \leq |N(v) \cap D| \cdot \frac{2}{A_\mr{min}} \leq \frac{2e\Delta\eps}{k} \leq 6\eps. \]

For the second sum, notice that if $y \notin D$, then $0 < \tau_y = \tau_y^*$. At this point, we conclude using the same argument as in the proof of \cite[Corollary~30]{hayes2013local}.
\end{proof}

We will prove local relations for the bias fields for $X, X^*, Z^*$. We will first establish such a relation for $Z^*$, and then transfer it to $X^*$ and $X$ via the previously established comparison results.

\begin{lemma}\label{lem:P-relative-for-Z*}
For every $\delta, \eps > 0$, there exist constants $\Delta_0 = \Delta_0(\delta, \eps)$ and $C = C(\delta, \eps)$ such that the following holds. Let $v \in V$, let $T \in [C, C+1]$, and let $c \in [k]$. Then,
    \[ \mb{P}\left[ \left| P(Z_T^*,v,c) - \sum_{w \in N(v)} \frac{\exp(- P(X_{\tau_w}^*,w,c) )}{|A(X_{\tau_w}^*,w)|} \right| > \eps \right] \leq \exp(-\Delta/C),\]
    where $\tau_w := \tau_w^{Z^*}(T)$.
\end{lemma}

\begin{proof}
Condition on $\mc F_v$ and on the collection of last ring times
$
\set{\tau_w:w\in N(v)}.
$
For $w\in N(v)$, write
\[
\xi_w:=
\frac{\mathbf 1\{c\in A_v(\Zst_{\tau_w},w)\}}{|A_v(\Zst_{\tau_w},w)|}.
\]
Then
\[
P(\Zst_T,v,c)=\sum_{w\in N(v)}\xi_w.
\]
Because $\Zst$ and $\Xst$ coincide on $S_2(v)$, the denominator is
\[
|A_v(\Zst_{\tau_w},w)|=|A(\Xst_{\tau_w},w)|,
\]
where we emphasize that the time $\tau_w$ on the right hand side is still $\tau_w^{Z^*}(T)$. 
Moreover, for fixed $w$ the event $\{c\in A_v(\Zst_{\tau_w},w)\}$ depends only on the colors of the vertices in $N(w)\setminus\{v\}$, and those vertices lie in pairwise disjoint branches for different $w$. Hence, conditioned on $\mc F_v$ and on the times $\tau_w$, the random variables $\xi_w$ are independent over $w\in N(v)$.

Fix $w\in N(v)$ and write $\tau= \tau_w = \tau_w^{\Zst}(T)$. For each $z\in N(w)\setminus\{v\}$, let $\sigma_z$ be the last successful refresh time of $z$ before $\tau$ in the chain $\Zst$ (or $\Xst$, since these chains coincide outside of $B_1(v)$). Conditional on $\mc F_v$ and the times $\{\sigma_z:z\in N(w)\setminus\{v\}\}$, the indicators
\[
\mathbf 1\{\Xst_{\tau}(z)=c\}
\qquad (z\in N(w)\setminus\{v\})
\]
are independent, and
\[
\mb P[\Xst_{\tau}(z)=c\mid \mc F_v,\{\sigma_y\}]
=
\frac{\mathbf 1\{c\in A(\Xst_{\sigma_z},z)\}}{|A(\Xst_{\sigma_z},z)|}.
\]
Therefore
\begin{align*}
&\mb P\brac{c\in A(\Xst_{\tau},w)\mid \mc F_v,\tau,\{\sigma_z\}}\\
&\qquad=
\prod_{z\in N(w)\setminus\{v\}}\paren{1-
\frac{\mathbf 1\{c\in A(\Xst_{\sigma_z},z)\}}{|A(\Xst_{\sigma_z},z)|}}
=
\exp\paren{-P(\Xst_{\tau},w,c)+O\paren{\sum_{z\in N(w)\setminus\{v\}}\frac1{|A(\Xst_{\sigma_z},z)|^2}}}.
\end{align*}
On the event from \cref{cor:universal-lower-bound}, every denominator is at least $(1-2\eps)ke^{-1/(1+\delta)}$, so the quadratic error is $O(\Delta/k^2)=O(1/\Delta)$. Hence
\[
\mb P\brac{c\in A(\Xst_{\tau},w)\mid \mc F_v,\tau,\{\sigma_z\}}
=e^{-P(\Xst_{\tau},w,c)}+O\paren{\frac1\Delta}.
\]
Dividing by the denominator $|A(\Xst_{\tau},w)|$ and taking conditional expectations, together with concentration of $|A(X^*_{\tau_w},w)|$ around its conditional mean, gives
\[
\Biggl|\mb E\brac{\xi_w\mid \mc F_v,\{\tau_y :y\in N(v)\}}
-
\frac{e^{-P(\Xst_{\tau_w},w,c)}}{|A(\Xst_{\tau_w},w)|}\Biggr|
\le \frac{C}{\Delta^2}.
\]
Summing over $w\in N(v)$ yields
\begin{equation}\label{eq:cond-exp-local-relation-new}
\Biggl|\mb E\brac{P(\Zst_T,v,c)\mid \mc F_v,\{\tau_y\}}
-
\sum_{w\in N(v)}\frac{e^{-P(\Xst_{\tau_w},w,c)}}{|A(\Xst_{\tau_w},w)|}\Biggr|
\le \frac{C}{\Delta}.
\end{equation}
Finally, on the same lower-bound event each $\xi_w$ is bounded by $C/\Delta$, and the family $(\xi_w)_{w\in N(v)}$ is conditionally independent. Hoeffding's inequality therefore gives
\[
\mb P\brac{\Bigl|P(\Zst_T,v,c)-\mb E[P(\Zst_T,v,c)\mid \mc F_v,\{\tau_y\}]\Bigr|>\eps/2\ \middle|\ \mc F_v,\{\tau_y\}}
\le e^{-\Delta/C}.
\]
Combining this with \eqref{eq:cond-exp-local-relation-new} proves the lemma.
\end{proof}

\begin{corollary}
\label{lem:P-relative-for-X*}
    For every $\delta, \eps > 0$, there exist constants $\Delta_0 = \Delta_0(\delta, \eps)$ and $C = C(\delta, \eps)$ such that the following holds. Let $v \in V$, let $T \in [C, C+1]$, and let $c \in [k]$. Then,
    \[ \mb{P}\left[ \left| P(X_T^*,v,c) - \sum_{w \in N(v)} \frac{\exp(- P(X_{\tau_w}^*,w,c) )}{|A(X_{\tau_w}^*,w)|} \right| > \eps \right] \leq \exp(-\Delta/C),\]
    where $\tau_w := \tau_w^{X^*}(T)$.
\end{corollary}

\begin{proof}
    This follows by combining \cref{lem:P-relative-for-Z*} and \cref{lem:Xstar-vs-Zstar}, using \cref{cor:universal-lower-bound} to control the change in sum when the last-refresh times differ on at most $\eps\Delta$ vertices. 
\end{proof}

\begin{corollary}\label{lem:P-local-relation-Xt}
 For every $\delta, \eps > 0$, there exist constants $\Delta_0 = \Delta_0(\delta, \eps)$, $C = C(\delta, \eps)$, and $R = R(\delta, \eps)$ such that the following holds. Let $v \in V$, and assume that $X_0$ is $400$-above-suspicion for radius $R$ at $v$. Then, with probability at least $1-\exp(-\Delta/C)$, the following holds simultaneously for every vertex $u \in B_R(v)$, every time $T \in [C, C+1]$, and every color $c \in [k]$
\[\left|P(X_T,u,c) - \sum_{w \in N(u)} \frac{\exp(- P(X_{\tau_w},w,c) )}{|A_u(X_{\tau_w},w)|} \right| \leq \eps,\]
where $\tau_w := \tau_w^X(T)$.
\end{corollary}

\begin{proof}
For a fixed time $T \in [C, C+1]$, this follows by combining \cref{lem:P-relative-for-X*} with the (proof of) \cref{lem:comparison-of-P}. To upgrade this to a bound over the entire interval $[C,C+1]$, we use a mesh argument as in the proof of \cref{cor:universal-lower-bound}.
\end{proof}

\begin{proposition}\label{prop:concentration-of-P}
For every $\delta, \eps > 0$, there exist constants $\Delta_0 = \Delta_0(\delta, \eps)$, $C = C(\delta, \eps)$, and $R = R(\delta, \eps)$ such that the following holds. Let $v \in V$, let $T \in [C, C+1]$, and assume that $X_0$ is $400$-above-suspicion for radius $R$ at $v$. Then, with probability at least $1-e^{-\Delta/C}$, the following holds simultaneously for every $u \in B_2(v)$ and every color $c \in [k]$:
    \[ |P(X_T,u,c) - d(u)/k| \leq \eps, \qquad |P(X_T^*, u,c) - d(u)/k| \leq \eps,\]
    where in the second estimate, $X^*$ denotes the recursive process centered at $u$.
    Moreover, for each $u$, the second estimate holds simultaneously for all $c\in[k]$ with conditional probability at least $1-e^{-\Delta/C}$ given the corresponding outer-history sigma-field $\mc F_u$, uniformly in the outer history.
\end{proposition}

\begin{proof}
We will only prove the result for $X$, since the result for $X^*$ follows identically.  Indeed, all the preceding estimates for $X^*$ were proved after conditioning on an arbitrary realization of $\mc F_u$, with constants uniform in that realization.

Fix $u \in B_2(v)$ and notice that, if $\tau_w > 0$ for all $w$,
    \[ \sum_{c \in [k]} P(X_T,u,c) = \sum_{c \in [k]} \sum_{w \in N(u)} \frac{\mbf{1}\{c \in A_u(X_{\tau_w},w)\}}{|A_u(X_{\tau_w},w)|} = \sum_{w \in N(u)} \frac{\sum_{c \in [k]} \mbf{1}\{ c \in A_u(X_{\tau_w},w) \}}{|A_u(X_{\tau_w},w)|} = d(u). \]

    Since $|\{ w : \tau_w < 0 \}| \preceq \on{Binom}(d(u),e^{-\rho_\delta T})$, by a Chernoff bound, this is at most $\eps\Delta/2$ except with probability $e^{-\Omega(\Delta)}$.  Thus, we need only show $|P(X_T,u,c) - P(X_T,u,c')| < \eps$ for any $c,c'$ with probability at least $1-\exp(-\Delta/C)$.  We can then union bound over the $\binom k2$ pairs.
    
    We will discretize time by cutting it into chunks of size $\eps^{-1}$ and define
    \[ \alpha_\ell := \max_{\substack{w \in B_\ell(v) \\ T - \ell \eps^{-1} \leq t \leq T}} |P(X_t,w,c) - P(X_t,w,c')|. \]

  Throughout, we will work on the event in \cref{cor:universal-lower-bound}, invoked with $\eps/2$, for all $w \in B_R(v)$. Set $A_\mr{min}:=(1-\eps)k/e$. As we may assume $\eps\leq 1/20$, on this event $|A(X_t,w)| \geq A_\mr{min} \geq k/3$; moreover, $A_y(X_t,w)\supseteq A(X_t,w)$ for every neighbor $y$ of $w$, so every denominator below is at least $A_\mr{min}$. Consequently, $P(X_{t},w,c) \leq 3$ for any $c$, and thus $\alpha_R \leq 3$. We will now inductively shrink $\alpha_\ell$ until we have a sufficient bound on $\alpha_2$.

By \cref{lem:P-local-relation-Xt},
    \begin{align*}
        P(X_{\tau_w},w,c') & \leq \eps + \paren{\sum_{z \in N(w)} \frac{\exp(- P(X_{\tau_z},z,c') )}{|A_w(X_{\tau_z},z)|}} \\
        & = \eps + \paren{\sum_{\substack{z \in N(w) \\ \tau_z \leq T - (\ell + 1) \eps^{-1}}} + \sum_{\substack{z \in N(w) \\ \tau_z > T - (\ell + 1) \eps^{-1}}}} \frac{\exp(- P(X_{\tau_z},z,c') )}{|A_w(X_{\tau_z},z)|} \\
        & \leq \eps + \frac{\#\{ z \in N(w) \mid \tau_z \leq T - \ell\eps^{-1} \}}{\delta\Delta} + \sum_{\substack{z \in N(w) \\ \tau_z > T - (\ell + 1) \eps^{-1}}} \frac{\exp(- (P(X_{\tau_z},z,c)-\alpha_{\ell+1}) )}{|A_w(X_{\tau_z},z)|} \\
        & \leq 2\eps + \paren{\sum_{z \in N(w)} \frac{\exp(- (P(X_{\tau_z},z,c) - \alpha_{\ell+1}) )}{|A_w(X_{\tau_{z}},z)|}} \\
        & \leq e^{\alpha_{\ell+1}} (P(X_{\tau_w},w,c) + \eps) + 2\eps \\
        & \leq e^{\alpha_{\ell+1}} P(X_{\tau_w},w,c) + 21\eps.
    \end{align*}

    For a vertex $y$, let $P^+(y) = \max\{ P(X_{\tau_y},y,c),P(X_{\tau_y},y,c') \}$ and $P^-(y) = \min\{ P(X_{\tau_y},y,c),P(X_{\tau_y},y,c') \}$.  The above then shows that \[P^+(y) \leq e^{\alpha_{\ell+1}} P^-(y) + 21\eps.\]
     Moreover, we have the numerical inequality
    \[\exp(-P^-(y)) - \exp(-P^-(y) e^{\alpha_{\ell+1}}) \leq \frac{\alpha_{\ell+1}}{e}.\]
    
    Thus,
    \begin{align*}
        P^+(w) - P^-(w) 
        & \leq 2\eps + \sum_{z \in N(w)} \frac{\exp( - P^-(z) ) - \exp( - P^+(z) )}{|A_w(X_{\tau_z},z)|} \\
        & \leq 2\eps + \sum_{z \in N(w)} \frac{\exp( - P^-(z) ) - \exp( - e^{\alpha_{\ell+1}} P^-(z) - 21\eps )}{A_\mr{min}} \\
        & \leq 2\eps + \sum_{z \in N(w)} \left(\frac{\alpha_{\ell+1}}{e A_\mr{min}} + \frac{1 - e^{-21\eps}}{A_\mr{min}} \right)\\
        & \leq 2\eps + \frac{\alpha_{\ell+1} \Delta}{e A_\mr{min}} + \frac{21\eps\Delta}{A_\mr{min}} \\
        & \leq 60\eps + \alpha_{\ell+1} \cdot \frac{\Delta}{(1-\eps)k} \\
        & \leq 60\eps + \alpha_{\ell+1} \cdot \frac{1}{(1-\eps)(1+\delta)} \leq 60\eps + (1-\delta/2) \alpha_{\ell+1}.
    \end{align*}
    This gives
    \[ \alpha_2 \leq 60 \eps \frac{1 - (1 - \delta/2)^{R-2}}{\delta/2} + 3(1-\delta/2)^{R-2}. \]

    Choosing a sufficiently large $R$, we get the result for $130\eps/\delta$.  Reparametrizing gives the desired result.
\end{proof}

\subsection{Local uniformity for \texorpdfstring{$X$}{X}}
We now have all the ingredients to prove \cref{thm:local-uniformity-continuous}. These follow exactly the same argument as in Hayes \cite{hayes2013local} with a slight correction to the recurrence, which is also why we needed to modify our function $P$. 

\begin{proof}[Proof of \cref{thm:local-uniformity-continuous} \eqref{eq:continuous-number-colors}]
    Fix $T\in[C,C+1]$ and a color $c$.  Condition on $\mc{F}_v$. By conditional independence of $\{X_T^*(w) = c\}_w$,
    \begin{align*}
        \mb{P}[c \in A(X_T^*,v) \mid \mc{F}_v] & = \prod_{w \in N(v)} \mb{P}[X_T^*(w) \ne c \mid \mc{F}_v] \\
        & = \prod_{w \in N(v)} \mb{E}_{\{\tau_w\}} \left[ \paren{1+O\paren{\frac{1}{(\delta\Delta)^2}}} \exp\left\{ -\frac{\mbf{1}\{c \in A(X_{\tau_w}^*,w)\}}{|A(X_{\tau_w}^*,w)|} \right\} \right] \\
        & = \paren{1+O\paren{\frac{1}{\delta^2\Delta}}} \mb{E}_{\{\tau_w\}} \exp\left\{ - \sum_{w \in N(v)} \frac{\mbf{1}\{c \in A(X_{\tau_w}^*,w)\}}{|A(X_{\tau_w}^*,w)|} \right\} \\
        & = \paren{1+O\paren{\frac{1}{\delta^2\Delta}}} \mb{E}_{\{\tau_w\}} \exp\{-P(X_T^*,v,c) \}.
    \end{align*}

    Thus by \cref{prop:concentration-of-P},
    \begin{equation}\label{eq:probability-of-one-color} \mb{P}[c \in A(X_T^*,v) \mid \mc{F}_v] = \paren{1 + O(\Delta^{-1} \delta^{-2}) + e^{-\Omega(\Delta)} \pm 2\eps} \exp\left\{ -\frac{d(v)}{k} \right\} = (1 \pm 3\eps) \exp\left\{ -\frac{d(v)}{k} \right\}. \end{equation}

    We get concentration about the mean $k e^{-d(v)/k}$ via \cref{lem:dyer-concentration}.  Since the preceding conditional estimates are uniform in the realization of $\mc{F}_v$, the tower law gives the same concentration bound for the unconditional random variable $|A(X_T^*, v)|$.  By \cref{prop:comparison}, we have the result for $X_T$.  Finally, applying this fixed-time estimate on the mesh and using the refresh-count interpolation from the proof of \cref{cor:universal-lower-bound} (with $\eps$ replaced by a sufficiently small constant multiple of $\eps$) gives the result simultaneously for every $T\in[C,C+1]$.
\end{proof}

\begin{proof}[Proof of \cref{thm:local-uniformity-continuous} \eqref{eq:continuous-blockers}] 
Condition on $\mc{F}_v$. Fix two colors $c_1,c_2$ and let $\eta_z = \mbf1\{X_t^*(z) = c_1\}$ and $\nu_z = \mbf1\{X_t^*(z) = c_2\}$.  Then
    \[ \sum_{z \in N(w)} \mb{E}[\eta_z \mid \mc{F}_v] = \sum_{z \in N(w)} \mb{E}_{\{\tau_z\}} \left[ \frac{\mbf1\{c_1 \in A(X_{\tau_z}^*,z)}{|A(X_{\tau_z}^*,z)|} \mid \mc{F}_v \right] = \mb{E}[P(X_t^*,w,c_1) \mid \mc{F}_v]. \]

    Thus $\{\eta_z\}_z$ are independent conditioned on $\mc{F}_v$ with $\sum_z \mb E[\eta_z\mid\mc F_v] = d(w)/k\pm \eps$ with high probability, and so $\sum_z \eta_z$ is $O(1/\Delta)$-total variation distance from Poisson (see, e.g.~\cite[Lemma 20]{hayes2013local}) and hence takes the desired values.  The same holds for $\nu_z$.

    However, clearly $\eta_z$ and $\nu_z$ are disjoint events and so not independent.  This is not a serious obstacle as the total-variation distance between the law of each pair $(\eta_z,\nu_z)$ and the product of its marginals is $O(1/k^2)$ and there are $O(\Delta^2)$ many vertices in $S_2(v)$ for $O(1)$ total errors, which is irrelevant as our error term is $\eps\Delta$.  Thus we may couple the true distribution of the $(2d(w))$-tuple $(\eta_z,\nu_z)$ with a collection of independent samples with only $O(1)$ many errors, giving us the result.
\end{proof}

\subsection{Weighted local uniformity}

We now record the exact local uniformity statement necessary in \cref{sec:analysis-coupling}.

\begin{corollary}
    In the setup of \cref{thm:local-uniformity-discrete}, for any color $c \in [k]$,
    \begin{equation}\label{eq:weighted-local-uniformity-statement} \mb{P}\left[ \exists t \in [T_0,T_0+T] : \left| \paren{\sum_{\substack{w \in N(v) \\ c \in A_v(X_t,w)}} e^{d(w)/k}} - d(v) \right| > \eps\Delta \right] < \frac Tn e^{-\Delta/C}. \end{equation}
\end{corollary}

\begin{proof}
    Note first that the previous proof of \cref{thm:local-uniformity-continuous} \eqref{eq:continuous-blockers} evidently holds for one color as well.  Thus for any $S \subset N(v)$, evaluating \eqref{eq:continuous-blockers} at $i=0$,
    \begin{equation}\label{eq:Sc0-control-one-color} \left| |S_{c,0}(X_t)| - \sum_{w \in S} e^{-d(w)/k} \right| \leq \eps\Delta \end{equation}
    except with probability $e^{-\Omega(\Delta)}$.  We will use the collection of sets
    \[ S^{\geq\ell} := \{ w \in N(v) : d(w) \geq \ell \}, \qquad S^{\geq\ell}_c := (S^{\geq\ell})_{c,0}(X_t) = \{ w \in S^{\geq\ell} : |N_v(w) \cap X_t^{-1}(c)| = 0 \}. \]

    There are at most $\Delta$ many such sets $S^{\geq\ell}$, so we may union bound that \eqref{eq:Sc0-control-one-color} holds for all sets $S^{\geq\ell}$.  We now use Abel summation
    \begin{align*}
        \sum_{\substack{w \in N(v) \\ c \in A_v(X_t,w)}} e^{d(w)/k} & = \sum_{w \in S^{\geq0}_c} e^{d(w)/k} = |S^{\geq0}_c| + \sum_{\ell=1}^\Delta (e^{\ell/k} - e^{(\ell-1)/k}) |S^{\geq\ell}_c| \\
        & = \paren{\sum_{w \in N(v)} e^{-d(w)/k} \pm \eps\Delta } + \sum_{\ell=1}^\Delta (e^{\ell/k} - e^{(\ell-1)/k}) \paren{\sum_{w \in S^{\geq\ell}} e^{-d(w)/k} \pm \eps\Delta} \\
        & = \paren{ \sum_{w \in N(v)} e^{-d(w)/k} + \sum_{\ell=1}^\Delta (e^{\ell/k} - e^{(\ell-1)/k}) \sum_{w \in S^{\geq\ell}} e^{-d(w)/k} } \pm e\eps\Delta \\
        & = \paren{ \sum_{w \in N(v)} e^{-d(w)/k} \paren{1 + \sum_{\ell=1}^{d(w)} \left(e^{\ell/k} - e^{(\ell-1)/k}\right)} } \pm e\eps\Delta = d(v) \pm e\eps\Delta.
    \end{align*}

    We may rescale $\eps$ by the constant $e$ and extend the result to discrete time exactly as in \cite{hayes2013local} to finish the proof.
\end{proof}

\section{Deferred proofs from \texorpdfstring{\cref{sec:analysis-coupling}}{non-Markovian}}
\label{sec:non-markovian-inequalities}
Finally, we deal with the non-Markovian failures. We begin with the proof of \cref{prop:non-markovians-all-succeed}, which we recall here for the reader's convenience. 

\propprobg*

        We will need the following claim regarding order statistics, whose proof is deferred to \cref{app:proof-of-bernoulli}.

        \begin{claim}\label{claim:bernoulli-sequence-nice-event}
            Fix $\gamma,\eps,m$ with $m \gg \eps^{-1} \gg \gamma^{-1}$.  Let $1 \geq p_1 \geq p_2 \geq \cdots \geq p_m > \gamma$.  Let $X_i \sim \Bernoulli(p_i)$ and $Y_i \sim \Bernoulli(p_i)$ be all independent.  Define the index (in $\{Y_i\}$)
            \[ \index(R) = \min \left\{ j : \sum_{i \leq j} Y_i \geq R \right\}. \]

            Let $j \in [m]$ be uniformly at random, and let $R = \sum_{i \leq j} X_i$ be the rank of $j$ in $\{X_i\}$.  
            Then
            \[ p_j - \eps^{1/2} \leq p_{\index(R) + \eps m} \leq p_{\index(R) - \eps m} \leq p_j + \eps^{1/2} \]
            except with probability
            \[ 4\eps + 4\eps^{1/2} + 6\exp\paren{- \frac{\gamma^2\eps^2 m}{8}} \ll 1. \]
        \end{claim}

Before proceeding to the proof of \cref{prop:non-markovians-all-succeed}, we formalize the conditional $\sigma$-algebras and auxiliary processes used to decouple the random variables in the local neighborhood of a non-Markovian update. As in Section 5, it is more convenient to transition to the continuous-time version of the Metropolis dynamics, from which results translate seamlessly to discrete time via standard bounds on the total number of updates.

To rigorously analyze the failure probability of the non-Markovian update without circular dependencies, we must decouple the update times of a vertex's neighbors from their actual colors. Because the Metropolis update rate of a vertex $w$ is $|A(X_t, w)|/k$, revealing update times typically leaks information about the 2-neighborhood. We resolve this by introducing an auxiliary conditional-expectation process.

\begin{definition}
Let $(\mathcal{P}_t)_{t \ge 0}$ denote the potential persistent discrepancy set defined by the bounding chain. For each integer $\ell \ge 1$, let $\tau_\ell$ be the continuous-time stopping time at which $|\mathcal{P}_t|$ first reaches size $\ell$. On the event $\{\tau_\ell \le C_{cp}\}$, let $v$ be the uniquely identified $\ell$-th vertex added to $\mathcal{P}$, and let $p \in \mathcal{P}_{<\tau_\ell} \cap N(v)$ be its unique parent in the tree-like bounding chain. 

We define $\mathcal{H}_\ell$ to be the $\sigma$-algebra generated by the following information:
\begin{itemize}
    \item The exact clock ring times and proposed colors for all vertices $u \in B_1(\mathcal{P}_t)$ up to time $\tau_\ell$. (This completely determines the evolution of $\mathcal{P}_t$, the stopping time $\tau_\ell$, and the vertices $v$ and $p$. Since $p \in \mathcal{P}_{<\tau_\ell}$, this also explicitly reveals the clocks for $p$ and its neighborhood $N(p)$ up to time $\tau_\ell$).
    \item The exact clock ring times and proposed colors for all vertices $x \notin B_2(v)$ up to time $\tau_\ell$.
\end{itemize}
\end{definition}

Notice that $\mathcal{H}_\ell$ is a slight refinement of the standard $\sigma$-algebra $\mathcal{F}$ of information outside $B_2(v)$. In addition to revealing the identity of $v$ and of $p$, we have revealed all clocks in $N_v(p)$ and the clock of $v$ in $[\tau_{\ell-1},\tau_\ell]$.  To analyze the probability of \cref{assumption:non-markovian-succeed}, we define a decoupled auxiliary chain $W^*$.

\begin{definition}[Auxiliary Process $W^*$]
\label{def:W-star}
Conditioned on $\mathcal{H}_\ell$ (on the event $\{\tau_\ell \le \Ccouple\}$), let $G^* = G_{\mathrm{in}}(v, 3)$ be the directed graph oriented towards $v$. We define the continuous-time process $W^* = (W^*_t)_{t \in [0, \Ccouple]}$ with initial configuration $W^*_0 = X_0$ via the following construction:
\begin{enumerate}
    \item \textbf{Revealed trajectories:} For all $x \notin B_2(v)$ and for all $x \in N_v(p) \cup \{p\}$, the clock rings and proposed colors are exactly those revealed in $\mathcal{H}_\ell$. 
    \item \textbf{Decoupled updates for $N_p(v)$:} Equip each vertex $w \in N_p(v) := N(v) \setminus \{p\}$ with an inhomogeneous Poisson clock of rate $\lambda_w(t) := \frac{1}{k} \mathbb{E}[|A(X_t, w)| \mid \mathcal{H}_\ell, \mc{F}_t]$ for times $t > \tau_\ell$. When this clock rings, choose a color uniformly at random from the currently available colors $A_{G^*}(W^*_{t^-}, w)$.
    \item \textbf{Standard clocks elsewhere:} Equip each unrevealed vertex $z \in S_2(v) \setminus N(p)$ with the usual $k$ independent Poisson clocks of rate $1/k$, each attempting an update $(z, c)$ for some $c \in [k]$. A proposal $c$ is accepted if it is available in $G^*$.
\end{enumerate}
(Note that the vertex $v$ itself also has its clocks revealed by $\mathcal{H}_\ell$, since $v \in B_1(p) \subseteq B_1(\mathcal{P}_t)$ for $t < \tau_\ell$, and we use those exact clock rings).
\end{definition}

This is a small modification of $Z^*$ that serves the same purpose of ensuring the trajectories of distinct vertices $w \in N_p(v)$ are conditionally (on $\mc{H}_\ell$) independent of one another.  We formally bound the discrepancy between this auxiliary process and the true dynamics.

\begin{lemma}[Coupling $W^*$ and $X$]
\label{lem:W_X_coupling}
The auxiliary process $W^*_t$ and the true continuous-time Metropolis dynamics $X_t$ can be coupled on the interval $t \in [0, \Ccouple]$ such that, with probability at least $1 - O(\varepsilon^{1/3})$, the two processes differ on at most $O(\varepsilon^{2/3}\Delta)$ vertices in $B_2(v)$ at time $\Ccouple$.
\end{lemma}
\begin{proof}
We begin by coupling $W^*$ to $Y^*$ using identical reasoning to \cref{lem:Xstar-vs-Zstar}.  The initial configuration is identical, the error rate at a given time is $\eps$, and we are only running for time $\Ccouple \ll \eps^{-1}$.  We then couple $Y^*$ to $X$ by \cref{prop:comparison}.  
We omit further details.
\end{proof}

\begin{proof}[Proof of \cref{prop:non-markovians-all-succeed}]
Our goal is to bound the probability of $\overline{\mathcal{G}}$, the event that some valid non-Markovian update fails \cref{assumption:non-markovian-succeed}, assuming the bounding chain is well-behaved and local uniformity ($\mathcal{U}$) holds. We will union bound the failure probability over all times the discrepancy set $\mathcal{P}_t$ grows. 
    
Fix $\ell \le \exp(\exp(O(C_{cp})))$ and condition on the boundary $\sigma$-algebra $\mathcal{H}_\ell$ at the stopping time $\tau_\ell$ (assuming $\tau_\ell \le \Ccouple$). Let $v = v_\ell$ and $p = p_\ell$ be the $\ell$th vertex to enter $\mc{P}$ and its parent. To cleanly analyze the probabilistic conditions of the non-Markovian edit at vertex $v$, we substitute the true process $X_t$ with the conditionally decoupled auxiliary process $W^*_t$ (\cref{def:W-star}). By \cref{lem:W_X_coupling}, we can discard the coupling failure event (which occurs with probability $\ll \varepsilon^{1/3}$) and assume $W^*_t$ accurately models $X_t$ for all $0 < t \leq \Tcouple$ up to an $O(\varepsilon^{2/3}\Delta)$ neighborhood error tolerance.

The following claim will perform the bulk of the probabilistic work within the rigorously decoupled environment of $W^*_t$.

    \begin{claim}
    \label{claim:nm}
        Let $\{W^*_t\}$, $v$, and $p$ be as defined above.  Fix a pair $H^* \subset H(Z_t,v)$ and suppose $\BC(X_0,Y_0,\vecsigma) = \tt True$.  Then
        \[ \mb{P}[H^* \text{ fails \cref{assumption:non-markovian-succeed}}] < \eps^{2/5}. \]
    \end{claim}

    \begin{proof}[Proof of \cref{claim:nm}]
        For convenience, we restate the bullets of \cref{assumption:non-markovian-succeed}.
        \begin{enumerate}
            \item\label{itm:BC} $\BC(X_0, Y_0, \vecsigma) = \tt{True}$;
            \item\label{itm:unblocked} $W^*_{\tau^-}(N_p(v) \cap \mc{P}) \cap H^* = \emptyset$;
            \item\label{itm:avoid} $B \cap \mathsf{Avoid}_{\mc{P}}(t) = \emptyset$;
            \item\label{itm:alpha} the mapping $\alpha$ is defined on $B$;
            \item\label{itm:beta} for each $w\in B$ with $\alpha(w)\ne w$, the color
            $\beta_w(X_{t-1}(\alpha(w)))$ is defined and satisfies  $\beta_w(X_{t-1}(\alpha(w))) \notin H^*$; 
            \item\label{itm:no-intersection} $\alpha(B)\cap B=\{ w \in B : \alpha(w) = w \}$.
        \end{enumerate}

        We begin by bounding the size of the set $B = (W^*_t)^{-1}(c_b) \cap N_p(v)$.  As $W^*_t(w)$ is uniformly distributed among available colors at the last ring time of $w$, of which there are at least $k/e$ by \cref{cor:universal-lower-bound}, we have that $|B|$ is stochastically dominated by a binomial random variable with parameters $\Delta$ and $e/k$.  Thus $|B| \leq \eps^{-1/10}$ except with probability $\exp(-\eps^{-\Omega(1)}) \ll \eps^{2/5}$, which can be swallowed into the error bound.


        Similarly, we assume $\deg v \geq \eps^{1/2} \Delta$.  If not, then $|B|$ has expectation at most $e\eps^{1/2}$ and so is zero except with probability $e\eps^{1/2} \ll \eps^{2/5}$, in which case \cref{assumption:non-markovian-succeed} trivially holds.

        We now systematically work through the bullets of \cref{assumption:non-markovian-succeed}.  \cref{itm:BC} is a hypothesis of \cref{claim:nm}.

        For \cref{itm:unblocked},
        notice that each vertex $w \in B$ ends up in $\mc{P}$ only if one of the clocks $(w,c)$ for $c \in Z_t(v)$ attempts an update before time $\Ccouple$.  Each of these clocks ring at rate $1/k$.  By Markov's inequality, we may bound by the expected number of such rings
        \[ |B| \cdot \frac{ |Z_t(v)| \cdot \Ccouple}{k} \leq \frac{\eps^{-1/10} \exp(\exp(O(\Ccouple))) \Ccouple}{\Delta} \ll \eps. \]

\cref{itm:avoid} is similar.  Recall the definition
\begin{equation*}
    \Avoid(t) = \{ w \in N(v_t) : \exists s \in I^\circ(w,t) \text{ where } v_s \in (N(w) \cap N_w^2(\mc{P})) \cup \{ v_t \}, c_s = X_{s-1}(w) \} \cup N_{v_t}^2(\mc{P}).
\end{equation*}

        First, we must bound $|\Avoid(t)|$.  We begin by bounding the size of the second set.  Suppose $z \in \mc{P}$.  We claim that $|N_{v_t}^2(z) \cap N(v_t)| \leq 1$ due to the girth condition.
        
        If $z \notin N^2(v_t)$, then there can be at most one path from $z$ to $v_t$ containing two interior vertices, as otherwise, we would have a $4$- or $6$-cycle in $G$, which has girth $7$.  Thus $|N_{v_t}^2(z) \cap N(v_t)| \leq 1$.  If $z \in N(v_t)$, then there can be no point in $N_{v_t}^2(z) \cap N(v_t)$, or this would give a $3$- or $4$-cycle in $G$.  If $z \in N^2(v_t) \setminus N(v_t)$, then the intermediate vertex is clearly in $N_{v_t}^2(z) \cap N(v_t)$.  There can be no other vertex in $N_{v_t}^2(z) \cap N(v_t)$ or this would give a $4$- or $5$-cycle in $G$.  Thus
        \[ |N_{v_t}^2(\mc{P}) \cap N(v_t)| \leq |\mc{P}|. \]

        Now, we bound the size of the first part of $\Avoid(t)$.  Fix $w \in N(v_t)$.  By the same reasoning as in the preceding paragraph, there are at most $|\mc{P}|+1$ many vertices $u \in (N(w) \cap N_w^2(\mc{P})) \cup \{v_t\}$.  By identical reasoning as for \cref{itm:unblocked}, the expected number of ring attempts is at most
        \[ (|\mc{P}| + 1) \cdot \frac{\Ccouple}{k} \ll \Delta^{-1/2}. \]

        Thus we have, whp, at most $2\deg v \Delta^{-1/2}$ vertices in $\Avoid(t)$.  Once we have conditioned on the interval sizes, this is all independent of the colors, so we may then sample $B$ (which is stochastically dominated by independent sampling with probability $e/k$) to get at most $O(\deg v/\Delta^{3/2}) = o(1)$ probability of failing \cref{itm:avoid}.

        The remaining bullets (\cref{itm:alpha,itm:beta,itm:no-intersection}) require finer analysis.  We begin by sampling $\tau_w^+$ and $\tau_w^-$ for each $w \in N_p(v)$.  Due to the structure of $W^*$, this has revealed no information about the ring times or update times of any $z \in N_p^2(v)$.  Once $\tau_w^+$ and $\tau_w^-$ are revealed, we now check whether each is $c$-swappable for $c = c_u$ or $c = c_b$.  Define
        \[ \Pure(w,t) = \{ c \in [k]
        : \mbox{for all }s\in I^\circ(w,t), \mbox{ if } v_s\in N(w) \mbox{ then } c_s\neq c\}. \]

        Notice that $(\Pure(w,t) \cap A(W^*_{\tau_w^-},w)) \oplus \Exchange(w,t) \subset \{ W^*_{t-1}(w), c_t \}$, as the only difference between the definitions is these two elements.  We can now exactly compute the probability of $c$-swappability 
        \[ \mb{P}[c \in \Pure(w,t)] = \exp\paren{-\frac{|I(w,t)| \cdot \deg w}{k}}. \]

        We must now compute $\mb{P}[c \in A(W^*_{\tau_w^-},w)]$.  This is independent of the previous event, as the previous event depended solely on the Poisson clocks of $z$ in $I(w,t)$, whereas this event depends solely on Poisson clocks in $[0,\tau_w^-)$.  By analogous reasoning to \cref{eq:probability-of-one-color}, we have the same output
        \[ \mb{P}[c \in A(W^*_{\tau_w^-},w)] = (1 \pm \eps) \exp \paren{- \frac{\deg w}{k}}. \]

        Thus
        \[ \mb{P}[c \in \Pure(w,t) \cap A(W^*_{\tau_w^-},w)] = (1 \pm \eps) \exp \paren{ - \frac{\deg w}{k} \paren{1 + |I(w,t)|} }. \]
        Let
        \[ p_w = \exp \paren{ - \frac{\deg w}{k} (1+|I(w,t)|) }. \]

        These events are also independent over all $w$.  Notice $1 \geq p_w > e^{-\Ccouple}$.  Recall that as we used Poisson clocks of rate $1$ instead of $1/n$, there will be a factor of $1/n$ on $I(w,t)$ when working in discrete time.  Now $\alpha$ is constructed as follows.

        Order $N_p(v) = \{ w_1,\ldots,w_{|N_p(v)|} \}$ by decreasing $p_{w_i}$.  We generate two sequences $S_i^b \sim \Bernoulli(p_{w_i})$ and $S_i^u \sim \Bernoulli(p_{w_i})$ of independent Bernoulli random variables and couple these with the chain $W^*$ so that $c_u \in \Pure(w_i,t) \cap A(W^*_{\tau_w^-},w) \iff S_i^u = 1$ and similarly $c_b \in \Pure(w_i,t) \cap A(W^*_{\tau_w^-},w) \iff S_i^b = 1$.  This is possible by the conditional independence under $\mc{H}_\ell$.  We will then sample $W^*_{\tau_w^-}(w)$, which are independent of all previously revealed data and of each other.  This will tell us $B$; in particular, $B$ is stochastically dominated by independent sampling with probability $e/k$.

        We will apply \cref{claim:bernoulli-sequence-nice-event} with $\gamma = e^{-\Ccouple}$, $\eps = \eps$, and $m = |N_p(v)| = \deg v - 1 \geq \eps^{1/2} \Delta - 1$.  Observe that the mapping $\alpha$ corresponds to exactly the scenario described.  As $B$ is stochastically dominated by independent sampling with probability $e/k$, we may sample a dominating set by first sampling $|B| \sim \on{Binom}(\deg v, e/k)$ and then choose the elements of $B$ uniformly at random from $N_p(v)$.  For each element $w \in B$, the mapping $\alpha(w)$ is exactly $\on{index}(R)$ with respect to $S_i^u$ where $R$ is the rank of $w$ in $S_i^b$.

        Fix $w \in B$.  Then by \cref{claim:bernoulli-sequence-nice-event}, except with probability $O(\eps^{1/2})$, $|p_w - p_{\on{index}(\on{rank}(w))}| \leq 2\eps^{1/2}$.  However, notice that $\alpha(w)$ may not align exactly with this idealized process due to the fact that we are working in $W^*_t$ and not $X_t$.  There are up to $\eps \Delta$ many discrepancies, and we must also delete $\mc{P} \cap N(w)$, which (very crudely) has size at most $|\mc{P}| = \exp(\exp(O(\Ccouple)))$.  Still, \cref{claim:bernoulli-sequence-nice-event} is sufficiently robust to handle this and say that $|p_w - p_{\alpha(w)}| \leq 2\eps^{1/2}$ except with probability $O(\eps^{1/2})$.  If this holds, then $\alpha(w)$ exists.  By a union bound, \cref{itm:alpha} holds except with probability
        \[ |B| \cdot O(\eps^{1/2}) = O(\eps^{2/5}). \]
        
        We now work on \cref{itm:beta}. Fix $w\in B$ with $\alpha(w)\ne w$.
        For any $w \in N_p(v)$, $|\Pure(w,t) \cap A(W^*_{\tau_w^-},w)|$ satisfies a Chernoff bound as it is a sum of $k$ different i.i.d.~$\Bernoulli((1\pm\eps)p_w)$ random variables,
        and thus is concentrated about its mean $p_w k$ with high probability (up to $2\eps$ to cover the fluctuations due to the probability not being exactly $p_w$).  As $|\Exchange(w,t)| = |\Pure(w,t) \cap A(W^*_{\tau_w^-},w)| \pm O(1)$, we have concentration of $|\Exchange(w,t)|$ as well.  Thus by a union bound, all $w \in N_p(v)$ have $|\Exchange(w,t)| = (1 \pm \eps) p_w k$.  By the triangle inequality and \cref{claim:bernoulli-sequence-nice-event},
        \begin{align*}
            \big||\Exchange(w,t)| - |\Exchange(\alpha(w),t)| \big| & \leq \big| |\Exchange(w,t)| - p_w k \big| + |p_w k - p_{\alpha(w)} k | + \big| |\Exchange(\alpha(w),t)| - p_{\alpha(w)} k \big| \\
            & \leq 2\eps k + 2\eps^{1/2} k + 2\eps k.
        \end{align*}

        Thus $\beta$ is defined unless $|\Exchange(\alpha(w),t)| > |\Exchange(w,t)|$ and $X_{t-1}(\alpha(w))$ falls in the unmatched suffix.
        Since $X_{t-1}(\alpha(w))$ is uniform from $\Exchange(\alpha(w),t)$, the probability it lands in these final entries (and so \cref{itm:beta} fails) is
        \[ \frac{\big||\Exchange(w,t)| - |\Exchange(\alpha(w),t)| \big|}{|\Exchange(\alpha(w),t)|} \leq \frac{(4\eps + 2\eps^{1/2})k}{(e^{-\Ccouple} - \eps) k} \ll \eps^{2/5}. \]

        In addition, we need that $\beta_w(X_{t-1}(\alpha(w))) \notin H^*$.  As $|H^*| = 2$ and $|\Exchange(\alpha(w),t)| \geq (e^{-\Ccouple}-\eps) k$, the probability of hitting $H^*$ is $\ll \eps^{2/5}$.

        Finally, we show \cref{itm:no-intersection}.  This will follow by a union bound.  For each $w \in N(v)$, the probability of being in $B$ is at most $e/k$, and can be determined solely by the sequence $S_i^b$.  The probability that $w \in \alpha(B \setminus \{w\})$ is also at most $e/k$ and is independent (as it depends on $S_i^u$).  Thus the expected number of vertices in $B \cap \alpha(B)$ that are not fixed points of $\alpha$ is at most
        \[ \deg v \cdot \paren{\frac ek}^2 = O(1/\Delta) \ll \eps. \]

        Thus by Markov's inequality, we have \cref{itm:no-intersection}, and we can union bound over all six items.
    \end{proof}

    The proposition follows from the \cref{claim:nm} and a union bound over all non-Markovian updates.  As $\BC(\cdot)$ holds, there are no repropagations, so there are at most $|\mc{P}| = \exp(\exp(O(\Ccouple)))$ times satisfying \cref{condition:nm-prelim}.  Similarly, we can union bound over all pairs contained in $Z_t(v)$, which has size at most $|\mc{P}|$ again.  Thus the probability of a non-Markovian update failing is at most
    \[ |\mc{P}| \cdot \binom{|\mc{P}|}{2} \cdot O(\eps^{2/5}) = \exp(\exp(O(\Ccouple))) \eps^{2/5} \ll \eps^{1/3}. \qedhere \]
\end{proof}

Finally, we handle temporary discrepancies introduced during the non-Markovian coupling.

\proptemp*

\begin{proof}
    For this result, we break into two time intervals.  Let $I_1 = [0, \Tcouple/10]$ and $I_2 = [\Tcouple/10,\Tcouple]$, and let $\gamma = e^{-\Ccouple}$.

    First, we control the size of $\Temp_t = B \cup \alpha(B)$.  As in the proof of \cref{claim:nm}, we have $|B|$ is stochastically dominated by a binomial distribution with parameters $\Delta$ and $e/k$.  

    First, notice that $|\Temp_t| \leq \gamma^{-1/10}$ with high probability as the set $|B|$ from \cref{def:local-NM-coupling} is stochastically dominated by a binomial distribution with parameters $\Delta$ and $e/k$.  If $\Temp_t$ is too large, we appeal to \cref{lem:HV-exponential-RV-bound} to say that we still see $o(1)$ contribution.

    We begin by handling the contribution from $I_2$.  In this case, we ignore the possibility of correcting them and simply bound the size of
    \[ \Temp_{I_2} = \bigcup_{t \in I_2} \Temp_t. \]
    
    Suppose there are $\ell$ non-Markovian updates in $I_2$.  Then $|\Temp_{I_2}|$
    is stochastically dominated by a sum of $\ell$ i.i.d.~$\on{Binom}(\Delta,e/k)$ random variables multiplied by $2$ (which are independent of $\ell$ as well).  In particular,
    \[ \mb{E}\left[\left|\bigcup_{t \in I_2} \Temp_t\right| \mid \ell \right] \leq 2e \ell. \]
    
    As argued in \cref{sec:analysis-coupling}, the probability of a non-Markovian error is at most $2\Delta \rho(t) / kn$.  Let $t = s + \Tcouple/10$.  Then
    \[ \mb{E}[\rho(t)] \leq \exp(2e \Cmod) \cdot \paren{1 - \frac{\delta}{3en}}^{\Tcouple/10 - \Cmod n + s} \leq \exp(-\delta \Ccouple/100) . \]

    Thus the expected number of attempted non-Markovian errors is at most 
    \[ \sum_{s = 0} ^{9\Tcouple/10} \frac{2\Delta}{kn} \exp(-\delta \Ccouple/100) \leq \frac{9\Tcouple}{10} \cdot \frac{2\exp({-\delta\Ccouple/100})}{n} \leq \exp(-\delta\Ccouple/200). \]    

    Handling $I_1$ requires a different approach.  In this region, we cannot yet control the probability of non-Markovian updates being attempted, but each temporary discrepancy has time $9\Tcouple/10$ to be ``resolved.''  Heuristically, we create $\exp(O(\Cmod))$ many temporary discrepancies and each will be fixed with probability $1-\exp(-\Omega(\Ccouple))$.
    
    Unfortunately, as $\alpha(B)$ depends on the update time $\tau_w^+$ of $w \in \alpha(B)$, we are badly lacking independence and must use auxiliary processes as before to gain independence.  As in the previous proof, we will work with the auxiliary process $\{W^*_t\}$ at some time $\tau$ when $v$ is added to $\mc{P}$.  If $\deg v \leq e^{-\Ccouple/10}\Delta$, then except with probability $e^{-\Omega(\Ccouple)}$, $|B| = 0$ and so $\Temp_t = \emptyset$ and there is nothing further to do.
    
    Otherwise, as in the previous analysis, we reveal $\tau_w^+$ and $\tau_w^-$ for all $w \in N(v)$.  Except with probability exponentially small in $\Delta$, we have at most $2e^{-9\Ccouple/10} \deg v$ vertices $w$ with $\tau_w^+ > \Tcouple$.  Call these vertices $w$ and their corresponding neighbors $\alpha^{-1}(w)$ ``risky.''  Notice $\Temp_t \cap (X_\Tcouple \oplus Y_\Tcouple)$ is empty if and only if $B$ contains no risky vertices.

    We now sample $B$, which is stochastically dominated by independent sampling with probability $e/k$.  Each risky vertex in $B$ contributes at most two discrepancies (in the unlikely event that both $w$ and $\alpha(w)$ are risky), so the expected number of discrepancies created that survive to the end is at most
    \[ 2e^{-9\Ccouple/10} \deg v \cdot \frac ek = \exp(-\Omega(\Ccouple)). \]

    Finally, we must sum this over all non-Markovian updates in $I_1$.  As before, the probability of attempting a non-Markovian update at any time step is at most $2\Delta\rho(t)/kn \leq 2\rho(t)/n$.  Further, the expectation of $\rho(t)$ is always at most $\exp(O(\Cmod))$.  Thus the total expectation is
    \[ \mb{E}\left[ \left| \bigcup_{t \in I_1} \Temp_t \cap (X_\Tcouple \oplus Y_\Tcouple) \right| \right] \leq |I_1| \cdot \frac{\exp(O(\Cmod))}{n} \cdot \exp(-\Omega(\Ccouple)) \ll 1. \]

    Thus we combine the two bounds over $I_1$ and $I_2$ to get a universal bound of $1/9$.
\end{proof}

\section{Fast Mixing: proof of \texorpdfstring{\cref{thm:main-mixing}}{Theorem 1.3}}
\label{sec:DFHV}
In this section we prove \cref{thm:main-mixing} by combining the one-block contraction from the non-Markovian coupling with the local-uniformity inputs from \cref{sec:local-uniformity}.  The overall strategy follows the burn-in and contraction approach from previous works \cite{dyer2003randomly, Molloy, hayes-old-girth5}, especially Dyer--Frieze--Hayes--Vigoda \cite{DFHV} to handle the constant degree case. As many of the details are quite similar, we will frequently refer the reader to \cite{DFHV} for various computations. 

\medskip 

Recall that we work on the extended state space
\[
    \wh\Omega := [k]^V,
\]
and, for every pair $X,Y\in \wh\Omega$, we fix once and for all a Hamming interpolation
\[
    X = Z_0 \sim Z_1 \sim \cdots \sim Z_{|X \oplus Y|} = Y
\]
of minimal length.  The intermediate labelings need not be proper colorings even if both $X$ and $Y$ are proper colorings, which is exactly the reason why the argument is carried out on $\wh\Omega$.

\medskip 

To handle various ``failure cases'', we will need certain weak-estimates for the identity coupling. 

\begin{lemma}
\label{lem:identity-estimates}
  For every $\delta > 0$ and every $C\ge 3$, there exists $\Delta_0=\Delta_0(\delta,C)$ such that the following holds.
Let $G=(V,E)$ be a graph on $n$ vertices with maximum degree $\Delta\ge \Delta_0$ and let $k\ge (1+\delta)\Delta$.
Let $X_0,Y_0\in\wh\Omega$ be neighboring labelings and let $(X_t,Y_t)_{t\ge 0}$ evolve under the identity coupling for the Metropolis dynamics on $\wh\Omega$. Then, for any $T\geq Cn$,
\begin{enumerate}
    \item \[
        \mb E\bigl[|X_{T} \oplus Y_{T}|\bigr] \le \exp(2T/n).
    \]
    \item Let $\ell \geq \exp(20T/n)$. Then,
    \[
        \mb E\bigl[|X_T \oplus Y_T|\mathbf 1(|X_T\oplus Y_T|>\ell^{2/3})\bigr] \le \exp(-\sqrt{\ell}).
    \]   
\end{enumerate}
\end{lemma}

\begin{proof}
   The first part follows by noting that the rate of spread of disagreements at any time is at most $1+2\Delta/(kn) \leq \exp(2/n)$, regardless of past history. The second part follows using the same argument as \cite[Lemma~3]{DFHV}.
\end{proof}

The main work in the proof of \cref{thm:main-mixing} is the following proposition, which provides a contractive coupling for $O(n\log \Delta)$ steps of the Metropolis dynamics, starting from a single disagreement. Once this proposition is established, \cref{thm:main-mixing} follows immediately via a standard path-coupling argument. On the other hand, similar to \cite{DFHV}, the proof of \cref{lem:twelve} itself uses path-coupling. 

\begin{proposition}
\label{lem:twelve}
For every $\delta > 0$, there exist constants $\Delta_0(\delta)$ and $C_{\mathrm{blk}} (\delta)$ such that the following holds. 
Let $G=(V,E)$ be a graph on $n$ vertices with maximum degree $\Delta\ge \Delta_0$ and girth at least $7$, and let $k\ge (1+\delta)\Delta$.
Let $X_0,Y_0\in\wh\Omega$ be neighboring labelings, with $X_0\oplus Y_0=\{z^*\}$. Set $T_{\mathrm{blk}} = C_{\mathrm{blk}}n\log\Delta$.

 Then, there exists a $T_\mathrm{blk}$-step coupling of the Metropolis dynamics starting from $(X_0, Y_0)$ such that
\[
    \mb E\bigl[|X_{T_\mathrm{blk}} \oplus Y_{T_\mathrm{blk}}|\bigr]\le \Delta^{-1}.
\]
\end{proposition}

\begin{proof}
We will choose constants in the following order
\[ 1/\Delta_0\ll 1/C_{\mathrm{blk}} \ll \eps \ll 1/\Ccouple \ll \delta,\]
where $\Ccouple$ and $\eps$ are as in the proof of \cref{thm:main-nonmarkovian}. We will also set 
\[\gamma := 1/(4\Ccouple).\]

We begin by describing our coupling. The coupling is defined block by block. Block 0 has length $\gamma T_{\mathrm{blk}}$. Subsequent blocks have length $\Tcouple$, so that there are  $m \approx ((1-\gamma)C_{\mathrm{blk}}/\Ccouple)\log\Delta$ such blocks.

For $i = 0, 1, 2,\dots,m$, suppose the pair at the beginning of block $i$ is $(X_{t(i)},Y_{t(i)})$. The high-level idea of the construction is the following. 

\begin{itemize}
    \item If at some earlier block boundary a ``bad'' event has already occurred, then from this block onward we use the identity coupling.

    \item Otherwise, we take the fixed Hamming interpolation
    \[
        X_{t(i)}=W^{(i)}_0\sim W^{(i)}_1\sim \cdots \sim W^{(i)}_{|X_{t(i)} \oplus Y_{t(i)}|}=Y_{t(i)},
    \]
    and we evolve each neighboring pair $(W^{(i)}_{r-1},W^{(i)}_r)$ over one block using our non-Markovian coupling from \cref{thm:main-nonmarkovian}; we then glue these couplings along the interpolation to obtain $(X_{t(i+1)}, Y_{t(i+1)})$. 
\end{itemize}

\medskip

For $0\le i\le m+1$, let
\[
    \on{Dist}_i:=|X_{t(i)} \oplus Y_{t(i)}|.
\]
For $1\le r\le \on{Dist}_i$, let $\on{Dist}_{i+1,r}$ be the Hamming distance after one block for the coupled evolution started from the neighboring pair $(W^{(i)}_{r-1},W^{(i)}_r)$.
By the triangle inequality,
\begin{equation}
\label{eq:rewrite-triangle}
    \on{Dist}_{i+1}\le \sum_{r=1}^{\on{Dist}_i} \on{Dist}_{i+1,r}.
\end{equation}

We now define our bad events, which are analogous to the ones in \cite{DFHV}. We will need the disagreement set of $X$ and $Y$. Accordingly, let
\[D^*_t := \cup_{0\leq s \leq t}\{v \in V: X_s(v) \neq Y_s(v)\}\]
Note that $D^*_t$ is an increasing sequence. We let $D^* = D^*_{T_{\mathrm{blk}}}$.

\medskip
\begin{itemize}
\item {\bf Large-growth.} Let
\[D_{\max} := \exp(100T_{\mathrm{blk}}/n) = \Delta^{100 C_{\mathrm{blk}}}.\]
and let $\mc{D}$ denote the event that
\[\on{Dist}_{m+1} \geq D_{\max}.\]

\item {\bf Local-uniformity fails.} Let $\vecsigma$ denote the sequence of updates for the $X$-chain after Block $0$. Let $\mc{NU}$ denote the event that $\LU(X_{t(1)}, \vecsigma, \eps, z^*)$ does not hold.

\item {\bf Disagreement escape.}
Let $\mc{E}$ denote the event that
\[D^* \not\subseteq B_{\Delta^{1/100}}(z^*).\]

\item 

\noindent {\bf Locally heavy disagreements.}
Let $\mc{H}$ denote the event that there exists $v \in B_{\Delta^{1/3}}(z^*)$ such that 
\[|B_2(v) \cap D^*| \geq \Delta^{2/3}.\]

\end{itemize}

Let
\[\mc{G} = \overline{\mc{NU} \cup \mc{E} \cup \mc{H}}.\]
Then, by \cref{thm:local-uniformity-discrete} to handle the $\mc{NU}$ term and the same paths-of-disagreements argument as in \cite[Lemma~4]{DFHV}  to handle the other two terms, we have
\[\mb{P}[\overline{\mc{G}}] \leq \exp(-\Delta^{1/10}).\]

Recall that our goal is to show
\[\mb{E}[\on{Dist}_{m+1}] \leq \Delta^{-1}.\]

We first write
\begin{align*}
    \mb{E}[\on{Dist}_{m+1}] &= \mb{E}[\on{Dist}_{m+1} \mathbf{1}_{\mc{D}}] + \mb{E}[\on{Dist}_{m+1} \mathbf{1}_{\overline{\mc{D}}}] \\
    &\leq \exp(-\Delta^{1/100}) + \mb{E}[\on{Dist}_{m+1} \mathbf{1}_{\overline{\mc{D}}}].
\end{align*}
The second line follows from the same standard argument as in \cite[Lemma~3]{DFHV}. This requires two inputs: firstly, the second conclusion of \cref{lem:identity-estimates} and secondly, a similar conclusion for the non-Markovian coupling,~i.e.~over $\Tcouple$ steps starting from neighboring labelings $X_0, Y_0 \in \wh{\Omega}$, the non-Markovian coupling satisfies the tail bound
 \[  \mb E\bigl[|X_\Tcouple \oplus Y_\Tcouple|\mathbf 1(|X_\Tcouple\oplus Y_\Tcouple|>\Delta^{2})\bigr] \le \exp(-\Delta^{1/2}).\]
On the event $\BC(X_0, Y_0, \vecsigma) = \tt{False}$, we use the identity coupling, so that this follows from \cref{lem:identity-estimates}; on the event $\BC(X_0, Y_0, \vecsigma) = \tt{True}$, this is deterministically bounded by $(2\Delta+1)|\mc{P}| \leq 3\Delta \cdot \exp(\exp(O(\Ccouple))) \leq \Delta^2$, after increasing $\Delta_0$.

Therefore, it suffices to control $\mb{E}[\on{Dist}_{m+1} \mathbf{1}_{\overline{\mc{D}}}]$. The idea, which is the same as in \cite{DFHV}, is the following: on the event $\overline{\mc{D}}$, we are guaranteed that the terminal distance is at most $D_{\max}$, which is polynomially large in $\Delta$. Therefore, we can absorb the failure of $\mc{G}$, which is exponentially small in a power of $\Delta$. On the other hand, on the event $\mc{G}$, the distance contracts by $1/3$ in expectation in each of the final $m$ stages. By choosing $\gamma$ sufficiently small, this easily offsets the growth during the first $\gamma T_{\mathrm{blk}}$ steps required for burn-in.

\medskip 

We proceed to formal details. We decompose
\begin{align*}
    \mb{E}[\on{Dist}_{m+1} \mathbf{1}_{\overline{\mc{D}}}]
    &\leq 
    \mb{E}[\on{Dist}_{m+1} \mathbf{1}_{\overline{\mc{D}}}\mathbf{1}_{\overline{\mc{G}}}] + \mb{E}[\on{Dist}_{m+1} \mathbf{1}_{\mc{G}}] \\
    &\leq D_{\max} \mb{P}[\overline{\mc{G}}] + \mb{E}[\on{Dist}_{m+1} \mathbf{1}_{\mc{G}}]\\
    &\leq \exp(-\Delta^{1/100}) + \mb{E}[\on{Dist}_{m+1} \mathbf{1}_{\mc{G}}],
\end{align*}
using the above estimates for $D_{\max}$ and $\mb{P}[\overline{\mc{G}}]$.

Finally, we control the dominant term $\mb{E}[\on{Dist}_{m+1} \mathbf{1}_\mc{G}]$. Consider the Hamming interpolation at the start of Block $m$:
\[
        X_{t(m)}=W^{(m)}_0\sim W^{(m)}_1\sim \cdots \sim W^{(m)}_{|X_{t(m)} \oplus Y_{t(m)}|}=Y_{t(m)},
    \]
Consider the adjacent pair $(W^{(m)}_{r-1}, W^{(m)}_r)$ with unique disagreement $v_{r-1}$ and let $\vecsigma_{r-1}$ denote the update sequence of $W^{(m)}_{r-1}$ in Block $m$ (obtained iteratively via the gluing lemma, starting from the update sequence for $X_0$). 

The key point is that on the event $\mc{G}$, the local uniformity event $\LU(W^{(m)}_{r-1}, \vecsigma_{r-1}, \eps, v_{r-1})$ holds for all $1\leq r \leq \on{Dist}_{m}$: $\overline{\mc{NU}}$ guarantees that the $X$-chain satisfies the $\eps$-local-uniformity properties for all times after the conclusion of Block $0$ in a ball of radius $\Delta^{1/10}$ around $z^*$. The event $\overline{\mc{E}}$ ensures that $W^{(m)}_{r}(t)$ coincides with $X_t$ outside $B_{\Delta^{2/100}}(z^*)$. Moreover, inside this ball, $\overline{\mc{H}}$ ensures that $2$-neighborhoods of $W^{(m)}_r(t)$ agree with those of $X_t$ up to a change of at most $\Delta^{2/3}$ assignments. Since both our local uniformity properties allow an $\eps \Delta \gg \Delta^{2/3}$ slack, this allows us to transfer local uniformity properties to $W^{(m)}_r(t)$. (We remark that here, we are following the somewhat simpler approach of \cite{DFHV}, but another approach is to modify the proof for local uniformity to apply directly to $W^{(m)}_{r-1}$, which follow an ``interpolated'' Metropolis dynamics, see Molloy \cite{Molloy}). Therefore, by \cref{thm:main-nonmarkovian} and path-coupling, it follows that
\[\mb{E}[\on{Dist}_{m+1} \mathbf{1}_{\mathcal{G}}] \leq (1/3)\mb{E}[\on{Dist}_{m} \mathbf{1}_{\mathcal{G}}].\]
Iterating this and using \cref{lem:identity-estimates} for the base case, we have
\begin{align*}
    \mb{E}[\on{Dist}_{m+1} \mathbf{1}_{\mathcal{G}}] 
    &\leq (1/3)^{m}\mb{E}[\on{Dist}_1 \mathbf{1}_{\mathcal{G}}] \\
    &\leq (1/3)^{m}\mb{E}[\on{Dist}_1]\\
    &\leq \Delta^{-C_{\mathrm{blk}}/\Ccouple} \exp(2\gamma T_{\mathrm{blk}}/n)\\
    &\leq \Delta^{-C_{\mathrm{blk}}/\Ccouple} \cdot \Delta^{2\gamma C_{\mathrm{blk}}}\\
    &\leq \Delta^{-1},
\end{align*}
by our choice of $\gamma$ and by taking $C_{\mathrm{blk}}$ sufficiently large compared to $\Ccouple$.
\end{proof}

The previous proposition easily implies \cref{thm:main-mixing}, again by a path-coupling argument.

\begin{proof}[Proof of \cref{thm:main-mixing}]
Let $T_{\mathrm{blk}}$ be as in \cref{lem:twelve} and let $K$ denote the $T_{\mathrm{blk}}$-step transition matrix on $\wh\Omega$.
By \cref{lem:twelve}, every neighboring pair $X, Y \in \wh{\Omega}$ admits a coupling $(X',Y')$ of $K(X,\cdot)$ and $K(Y,\cdot)$ such that
\begin{equation*}
\label{eq:v4-good-edge-H}
    \mb E\bigl[|X'\oplus Y'|\bigr]\le \Delta^{-1}.
\end{equation*}
Consequently, by path-coupling with respect to the Hamming metric, every pair $X,Y \in {\wh{\Omega}}$ admits a coupling $(X', Y')$ of $K(X, \cdot)$ and $K(Y, \cdot)$ such that
\[\mb{E} \bigl [|X' \oplus Y'| \bigr]\le \Delta^{-1}|X\oplus Y|.\]
Iterating this $r$ times, every pair $X,Y \in \wh{\Omega}$ admits a coupling $(X'', Y'')$ of $K^r(X,\cdot)$ and $K^r(Y,\cdot)$ such that
\[\mb{E}[|X'' \oplus Y''|] \leq \Delta^{-r} |X\oplus Y| \leq \Delta^{-r}n.\]
In particular, since $|X'' \oplus Y''| \geq \mathbf{1}[X''\neq Y'']$, it follows from the coupling characterization of total variation distance that for all $X, Y \in \Omega$,
\begin{align*}
    \left\|K^r(X,\cdot)-K^r(Y,\cdot)\right\|_{\mathrm{TV}}
    &\le \Delta^{-r} n.
\end{align*}
For any $\eta \in (0,1)$, choosing $r = \left\lceil\log(n/\eta)/\log\Delta\right\rceil$ makes the right hand side at most $\eta$.
Therefore, taking $Y \sim \pi$, we see that
\[T_{\on{mix}}(\eta) \leq \left\lceil\log(n/\eta)/\log\Delta\right\rceil T_{\mathrm{blk}} = O_{\delta}(n\log(n/\eta)),\]
as claimed.
\end{proof}

\bibliographystyle{alpha}
\bibliography{main}

\appendix
\crefalias{section}{appendix}

\section{Deferred computation from \texorpdfstring{\cref{sub:controlling-bounding-chain}}{bounding chain}}\label{app:computation}

In the proof of \cref{lem:P-total-size-bound} we defined
\[
\ell_0=2,\qquad m_0=1,
\]
and, for \(0\le j<100\Ccouple\),
\[
\ell_{j+1}
=
2\ell_j+\frac{\log(200\Ccouple m_j/\gamma)}{\log 25},
\qquad
m_{j+1}
=
\frac{200\Ccouple}{\gamma}\,m_j \exp(\ell_{j+1}/100).
\]
We now verify that when \(\gamma=\exp(-\Ccouple^2)\), this recursion gives
\[
m_{100\Ccouple}\le \exp(\exp(O(\Ccouple))).
\]

\begin{lemma}\label{lem:recursive-threshold-bound}
Assume \(\gamma=\exp(-\Ccouple^2)\). Then for every \(0\le j\le 100\Ccouple\),
\[
\ell_j,\ \log m_j
\le
3^j\bigl(\Ccouple^2+\log(200\Ccouple)+2\bigr).
\]
In particular,
\[
\ell_{100\Ccouple}\le \exp(O(\Ccouple)),
\qquad
m_{100\Ccouple}\le \exp(\exp(O(\Ccouple))).
\]
\end{lemma}

\begin{proof}
Set
\[
A:=\log\left(\frac{200\Ccouple}{\gamma}\right)
=\Ccouple^2+\log(200\Ccouple),
\qquad
u_j:=3^j(A+2).
\]
We claim that for all \(j\),
\[
\ell_j\le u_j,
\qquad
\log m_j\le u_j.
\]

This is immediate for \(j=0\), since
\[
\ell_0=2\le A+2=u_0,
\qquad
\log m_0=0\le u_0.
\]

Now assume \(\ell_j,\log m_j\le u_j\). Since \(u_j\ge u_0=A+2\), we also have \(A\le u_j\). Hence
\[
\ell_{j+1}
=
2\ell_j+\frac{A+\log m_j}{\log 25}
\le
2u_j+\frac{2u_j}{\log 25}
<
3u_j
=
u_{j+1},
\]
because \(\log 25>2\).

Also,
\[
\log m_{j+1}
=
A+\log m_j+\frac{\ell_{j+1}}{100}
\le
u_j+u_j+\frac{u_{j+1}}{100}
<
3u_j
=
u_{j+1}.
\]
This closes the induction.

Finally,
\[
\log m_{100\Ccouple}
\le
3^{100\Ccouple}(A+2)
=
\exp(O(\Ccouple)),
\]
since \(A=\Ccouple^2+\log(200\Ccouple)=O(\Ccouple^2)\). Exponentiating gives
\[
m_{100\Ccouple}\le \exp(\exp(O(\Ccouple))). \qedhere
\]
\end{proof}

\section{Proof of \texorpdfstring{\cref{claim:bernoulli-sequence-nice-event}}{order statistics}}
\label{app:proof-of-bernoulli}

Let \(\kappa:=\eps m\), and for notational simplicity assume that \(\kappa\in \mathbb{N}\); inserting floors and ceilings does not change the argument.  Also write
\[
A_t:=\sum_{i\le t} X_i,\qquad B_t:=\sum_{i\le t} Y_i,\qquad \mu_t:=\sum_{i\le t} p_i.
\]
Then \(R=A_j\), while \(\index(R)\) is the first time \(B_t\) reaches the value \(R\).

The proof has two ingredients.  First, for a uniformly random \(j\), the monotone sequence \((p_i)\) is typically almost constant on the window \([j-2\kappa,j+2\kappa]\).  Second, for such a \(j\), the random index \(\index(R)\) is typically within distance \(\kappa\) of \(j\).  Putting these together gives the desired comparison between \(p_j\) and the nearby values \(p_{\index(R)\pm \kappa}\).

We begin with the first point.  For \(2\kappa+1\le j\le m-2\kappa\), set
\[
\kappa_j:=p_{j-2\kappa}-p_{j+2\kappa}\ge 0.
\]
Since \((p_i)\) is nonincreasing, these \(\kappa_j\) measure the total variation of the sequence on windows of radius \(2\kappa\).  Summing over all interior \(j\), we obtain a telescoping estimate:
\[
\sum_{j=2\kappa+1}^{m-2\kappa}\kappa_j
=
\sum_{j=2\kappa+1}^{m-2\kappa} p_{j-2\kappa}
-
\sum_{j=2\kappa+1}^{m-2\kappa} p_{j+2\kappa}
=
\sum_{i=1}^{m-4\kappa} p_i
-
\sum_{i=4\kappa+1}^{m} p_i.
\]
Hence
\[
\sum_{j=2\kappa+1}^{m-2\kappa}\kappa_j
=
\sum_{i=1}^{4\kappa} p_i-\sum_{i=m-4\kappa+1}^{m} p_i
\le 4\kappa,
\]
because each \(p_i\le 1\).  Since \(j\) is uniform in \([m]\), this implies
\[
\mb{P}\left(j\notin [2\kappa+1,m-2\kappa]\right)\le \frac{4\kappa}{m}=4\eps,
\]
and, by Markov's inequality,
\[
\mb{P}\left(\kappa_j>\eps^{1/2}\right)
\le
\frac{1}{m\eps^{1/2}}
\sum_{j=2\kappa+1}^{m-2\kappa}\kappa_j
\le
\frac{4\kappa}{m\eps^{1/2}}
=
4\eps^{1/2}.
\]
Therefore, with probability at least \(1-4\eps-4\eps^{1/2}\), we have simultaneously
\[
j\in [2\kappa+1,m-2\kappa]
\qquad\text{and}\qquad
p_{j-2\kappa}-p_{j+2\kappa}\le \eps^{1/2}.
\]
On this event, monotonicity of \((p_i)\) yields
\begin{equation}
\label{equation:one}
p_j-\eps^{1/2}\le p_{j+2\kappa}\le p_j\le p_{j-2\kappa}\le p_j+\eps^{1/2}.
\end{equation}

We now turn to the second point.  Fix \(j\in [2\kappa+1,m-2\kappa]\).  We claim that with high probability,
\begin{equation}\label{equation:two}
|\index(R)-j|\le \kappa.
\end{equation}
To see this, let
\[
\lambda:=\frac{\gamma\kappa}{4},
\]
and consider the event
\[
E_j:=
\left\{
|A_j-\mu_j|\le \lambda,\ 
|B_{j-\kappa}-\mu_{j-\kappa}|\le \lambda,\ 
|B_{j+\kappa}-\mu_{j+\kappa}|\le \lambda
\right\}.
\]
Since \(A_j\), \(B_{j-\kappa}\), and \(B_{j+\kappa}\) are sums of independent Bernoulli variables, Hoeffding's inequality gives
\begin{equation}
\label{equation:three}
\mb{P}(E_j^c)
\le
6\exp\left(-\frac{2\lambda^2}{m}\right)
=
6\exp\left(-\frac{\gamma^2\kappa^2}{8m}\right)
=
6\exp\left(-\frac{\gamma^2\eps^2 m}{8}\right).
\end{equation}
Assume now that \(E_j\) occurs.  Because every \(p_i\ge \gamma\), we have
\[
\mu_j-\mu_{j-\kappa}=\sum_{i=j-\kappa+1}^j p_i\ge \gamma\kappa,
\qquad
\mu_{j+\kappa}-\mu_j=\sum_{i=j+1}^{j+\kappa} p_i\ge \gamma\kappa.
\]
Using \(R=A_j\), it follows that
\[
B_{j-\kappa}
\le \mu_{j-\kappa}+\lambda
\le \mu_j-\gamma\kappa+\lambda
= \mu_j-\frac{3\gamma\kappa}{4}
< \mu_j-\lambda
\le A_j=R,
\]
and similarly
\[
B_{j+\kappa}
\ge \mu_{j+\kappa}-\lambda
\ge \mu_j+\gamma\kappa-\lambda
= \mu_j+\frac{3\gamma\kappa}{4}
> \mu_j+\lambda
\ge A_j=R.
\]
Thus \(B_{j-\kappa}<R\le B_{j+\kappa}\).  Since \(B_t\) is nondecreasing and increases only by steps of size \(0\) or \(1\), it must hit the value \(R\) at some time between \(j-\kappa+1\) and \(j+\kappa\), i.e.~
\[
|\index(R)-j|\le \kappa,
\]
which proves \eqref{equation:two} on the event \(E_j\).

Finally, we combine the two ingredients.  Suppose that \(j\) satisfies the conclusion of the first part, and that \(E_j\) also occurs.  Then \eqref{equation:two} implies
\[
j\le \index(R)+\kappa\le j+2\kappa,
\qquad
j-2\kappa\le \index(R)-\kappa\le j.
\]
Since \((p_i)\) is nonincreasing, this gives
\begin{equation}\label{equation:four}
p_{j+2\kappa}\le p_{\index(R)+\kappa}\le p_{\index(R)-\kappa}\le p_{j-2\kappa}.
\end{equation}
Combining \eqref{equation:four} with \eqref{equation:one}, we conclude that
\[
p_j-\eps^{1/2}
\le
p_{\index(R)+\kappa}
\le
p_{\index(R)-\kappa}
\le
p_j+\eps^{1/2}.
\]
Recalling that \(\kappa=\eps m\), this is exactly
\[
p_j-\eps^{1/2}
\le
p_{\index(R)+\eps m}
\le
p_{\index(R)-\eps m}
\le
p_j+\eps^{1/2}.
\]

It remains only to estimate the failure probability.  The first part fails with probability at most \(4\eps+4\eps^{1/2}\), while by \eqref{equation:three}, conditional on any fixed \(j\), the event \(E_j\) fails with probability at most
\[
6\exp\left(-\frac{\gamma^2\eps^2 m}{8}\right).
\]
Hence
\[
\mb{P}(\text{failure})
\le
4\eps+4\eps^{1/2}
+
6\exp\left(-\frac{\gamma^2\eps^2 m}{8}\right).
\]
Under the regime \(m\gg \eps^{-1}\gg \gamma^{-1}\), this quantity tends to \(0\).  Therefore the stated inequalities hold with high probability. \qed

\end{document}